\PassOptionsToPackage{nosumlimits,nonamelimits}{amsmath}
\PassOptionsToPackage{colorlinks,linkcolor={blue},citecolor={blue},urlcolor={red},breaklinks=true,final}{hyperref}
\documentclass[runningheads,final]{llncs}
\usepackage{stel-common}
\newcommand{\by}[1]{\text{/$\mspace{-2mu}$/~#1}}		       %

\usepackage{bm} 

\newcommand{\mybar}[3]{%
  \mathrlap{\hspace{#2}\overline{\scalebox{#1}[1]{\phantom{\ensuremath{#3}}}}}\ensuremath{#3}
}

\newcommand{\barSigmas}{\mybar{0.9}{0pt}{\Sigmas}}

\newcommand{\gra}{\mathsf{gra}}
\newcommand{\dom}{\mathsf{dom}}

\providecommand{\catname}{\mathbf} 
\providecommand{\clsname}{\mathcal}
\providecommand{\oname}[1]{{\mathop{\mathsf{#1}}\xspace}}

\def\defcatname#1{\expandafter\def\csname B#1\endcsname{\catname{#1}}}
\def\defcatnames#1{\ifx#1\defcatnames\else\defcatname#1\expandafter\defcatnames\fi}
\defcatnames ABCDEFGHIJKLMNOPQRSTUVWXYZ\defcatnames

\def\defclsname#1{\expandafter\def\csname C#1\endcsname{\clsname{#1}}}
\def\defclsnames#1{\ifx#1\defclsnames\else\defclsname#1\expandafter\defclsnames\fi}
\defclsnames ABCDEFGHIJKLMNOPQRSTUVWXYZ\defclsnames

\def\defbbname#1{\expandafter\def\csname BB#1\endcsname{{\mathbb{#1}}}}
\def\defbbnames#1{\ifx#1\defbbnames\else\defbbname#1\expandafter\defbbnames\fi}
\defbbnames ABCDEFGHIJKLMNOPQRSTUVWXYZ\defbbnames

\def\Set{\catname{Set}}

\providecommand{\argument}{-}

\providecommand{\wave}[1]{\widetilde{#1}}		               

\DeclareOldFontCommand{\bf}{\normalfont\bfseries}{\mathbf}

\providecommand{\Id}{\operatorname{Id}}

\providecommand{\id}{\mathsf{id}}
\providecommand{\op}{\mathsf{op}}
\providecommand{\comp}{\mathbin{\circ}}
\providecommand{\iso}{\mathbin{\cong}}

\providecommand{\xto}[1]{\,\xrightarrow{#1}\,}

\providecommand{\To}{\mathrel{\Rightarrow}}			           

\providecommand{\dar}{\kern-1.2pt\operatorname{\downarrow}}	
\providecommand{\uar}{\kern-1.2pt\operatorname{\uparrow}}	

\providecommand{\bigor}{\bigvee}
\providecommand{\bigand}{\bigwedge}

\providecommand{\fst}{\oname{fst}}
\providecommand{\snd}{\oname{snd}}

\providecommand{\brks}[1]{\langle #1\rangle}

\providecommand{\inl}{\oname{inl}}
\providecommand{\inr}{\oname{inr}}
\providecommand{\inj}{\oname{in}}

\DeclareSymbolFont{Symbols}{OMS}{cmsy}{m}{n}
\DeclareMathSymbol{\iobj}{\mathord}{Symbols}{"3B}

\providecommand{\ev}{\oname{ev}}

\usepackage{stmaryrd}

\providecommand{\by}[1]{\text{/\!\!/~#1}}			             
\providecommand{\pacman}[1]{}					                     

\newcommand{\undefine}[1]{\let #1\relax}					                       

\providecommand{\mone}{{\text{\kern.5pt\rmfamily-}\mathsf{\kern-.5pt1}}}

\makeatletter
\def\mfix#1{\oname{#1}\@ifnextchar\bgroup\@mfix{}}	       
\def\@mfix#1{#1\@ifnextchar\bgroup\mfix{}}			           
\makeatother

\providecommand{\case}[3]{\mfix{case}{\mathbin{}#1}{of}{#2}{\kern-1pt;}{\mathbin{}#3}}

\DeclareMathSymbol{\mathinvertedexclamationmark}{\mathord}{operators}{'074}
\DeclareMathSymbol{\mathexclamationmark}{\mathord}{operators}{'041}
\makeatletter
\newcommand{\raisedmathinvertedexclamationmark}{%
  \mathord{\mathpalette\raised@mathinvertedexclamationmark\relax}%
}
\newcommand{\raised@mathinvertedexclamationmark}[2]{%
  \raisebox{\depth}{$\m@th#1\mathinvertedexclamationmark$}%
}
\makeatother

\newcommand{\sqdown}{\logp{\Downarrow}}

\newcommand{\wt}{\widetilde}

\newcommand{\Pt}{V}

\newcommand{\Pow}{\mathcal{P}}
\newcommand{\xTo}{\xRightarrow}

\newcommand{\hatini}{\hat{\ini}}

\newcommand{\qand}{\quad\text{and}\quad}
\newcommand{\qqand}{\qquad\text{and}\qquad}

\newcommand{\under}[1]{\lvert#1\rvert}

\newcommand{\Sigmas}{\Sigma^{\star}}

\newcommand{\ar}{\mathsf{ar}}

\newcommand{\seq}{\subseteq}
\newcommand{\ol}{\overline}
\newcommand{\out}{\mathsf{out}}
\newcommand{\outl}{\mathsf{outl}}
\newcommand{\outr}{\mathsf{outr}}

\providecommand{\C}{}

\renewcommand{\C}{{\mathcal{C}}}

\renewcommand{\id}{{\mathsf{id}}}
\newcommand{\Nat}{\mathds{N}}

\newcommand{\f}{\oname{f}}

\newcommand{\takeout}[1]{\empty}

\newcommand{\ini}{\iota}

\newcommand{\wh}{\widehat}

\DeclareMathOperator{\Alg}{\mathbf{Alg}}

\renewcommand{\rho}{\varrho}

\newcommand{\opp}{\mathsf{op}}

\newcommand{\pullbackangle}[2][]{\arrow[phantom,to path={
                     -- ($ (\tikztostart)!1cm!#2:([xshift=8cm]\tikztostart) $)
                        node[anchor=west,pos=0.0,rotate=#2,
                        inner xsep = 0]
                        {\begin{tikzpicture}[minimum
                        height=1mm,baseline=0,#1]
    \draw[-] (0,0) -- (.5em,.5em) -- (0,1em);
                        \end{tikzpicture}}}]{}}

\makeatletter
\newsavebox{\@brx}
\newcommand{\llangle}[1][]{\savebox{\@brx}{\(\m@th{#1\langle}\)}%
  \mathopen{\copy\@brx\kern-0.5\wd\@brx\usebox{\@brx}}}
\newcommand{\rrangle}[1][]{\savebox{\@brx}{\(\m@th{#1\rangle}\)}%
  \mathclose{\copy\@brx\kern-0.5\wd\@brx\usebox{\@brx}}}
\makeatother

\renewcommand{\comp}{\cdot}
\renewcommand{\c}{\colon}

\newcommand{\xra}[1]{\mathrel{\raisebox{-1.15pt}{$\xrightarrow{\;\smash{\raisebox{2.5pt}{\makebox(3,0)[t]{\scriptsize $#1$}}\;}}$}}}
\renewcommand{\xto}{\xra}

\newcommand{\monoto}{\rightarrowtail}

\renewcommand{\Nat}{\mathbb{N}}

\newcommand{\fset}{{\mathbb{F}}}

\newcommand{\mon}{\bullet}

\newcommand{\mS}{{\mu\Sigma}}

\renewcommand{\epsilon}{\varepsilon}

\newcommand*\xbar[1]{%
  \kern.2em\hbox{%
    \vbox{%
      \hrule height 0.5pt 
      \kern0.5ex
      \hbox{%
        \kern-0.2em
        \ensuremath{#1}%
        \kern-0.4em
      }%
    }%
  }\kern.4em %
} 

\makeatletter
\newcommand{\monto}{\@ifstar{\@mtolifted}{\@mto}}
\newcommand{\@mto}{\multimapdot}
\newcommand{\@mtolifted}{\mathbin{\xbar{\multimapdot}}}
\makeatother

\newcommand{\Ar}[1]{\mathbf{Ar}(#1)}

\newcommand{\app}[2]{\,}

\let\oldcheckmark\checkmark
\renewcommand{\checkmark}{\raisebox{-4pt}{\scalebox{1.2}[.65]{$\oldcheckmark$}}}

\usepackage[utf8]{inputenc}

\usepackage{hypcap}
\usepackage{microtype}
\usepackage[strict]{changepage}

\usepackage{enumitem}
\setlist[enumerate,1]{label=(\arabic*),font=\normalfont,align=left,leftmargin=0pt,labelindent=0pt,listparindent=\parindent,labelwidth=0pt,itemindent=!,topsep=2pt,parsep=0pt,itemsep=2pt,start=1}
\setlist[enumerate,2]{label=(\alph*),font=\normalfont,labelindent=*,leftmargin=*,start=1}
\setlist[itemize]{labelindent=*,leftmargin=*}
\setlist[description]{labelindent=*,leftmargin=*,itemindent=-1 em}

\usepackage{ifdraft}
\ifdraft{
  \usepackage[final]{showlabels}

  \usepackage[footnote,marginclue,nomargin]{fixme}
}{
 \usepackage[layout=footnote,final]{fixme}
}

\FXRegisterAuthor{hu}{ahu}{HU} %
\FXRegisterAuthor{sm}{asm}{SM} %
\FXRegisterAuthor{st}{ast}{ST} %
\FXRegisterAuthor{ls}{als}{LS} %
\FXRegisterAuthor{sg}{asg}{SG} %
\FXRegisterAuthor{as}{aas}{AS} %

\renewcommand{\comp}{\cdot}
\newcommand{\klstar}{\sharp}

\let\cedilla\c
\renewcommand{\c}{\colon}
\newcommand{\term}{1} %

\usepackage{seqsplit}
\usepackage{xstring}
\usepackage{xcolor}

\tikzstyle{shiftarr}=[
rounded corners,%
to path={--([#1]\tikztostart.center)
  -- ([#1]\tikztotarget.center) \tikztonodes
  -- (\tikztotarget)},
]

\tikzset{
  commutative diagrams/.cd,
  arrow style=tikz,
  diagrams={>={Straight Barb[length=1.75pt,width=3.85pt,inset=1.95pt]}}, %
  row sep=large,
  column sep = huge
}

\tikzset{cong/.style={draw=none,edge node={node [sloped, allow upside down, auto=false]{$\cong$}}},
  iso/.style={draw=none,every to/.append style={edge node={node [sloped, allow upside down, auto=false]{$\cong$}}}}}

\usetikzlibrary{decorations.pathmorphing}

\tikzcdset{scale cd/.style={every label/.append style={scale=#1},
    cells={nodes={scale=#1}}}}

\usepackage{stackengine}
\stackMath
\newcommand\tsup[2][2]{%
  \def\useanchorwidth{T}%
  \ifnum#1>1%
    \stackon[-1.05ex]{\tsup[\numexpr#1-1\relax]{#2}}{\scalebox{2}[1]{$\mathchar"307E$}\kern-.5pt}%
  \else%
    \stackon[-.9ex]{#2}{\scalebox{2}[1]{$\mathchar"307E$}\kern-.5pt}%
  \fi%
}

\usepackage{xspace}

\usepackage{etoolbox} %

\theoremstyle{plain}

\newtheorem{lem}[theorem]{Lemma}
\newtheorem{prop}[theorem]{Proposition}

\theoremstyle{definition}
\newtheorem{defn}[theorem]{Definition} %
\newtheorem{expl}[theorem]{Example} %
\newtheorem{rem}[theorem]{Remark} %
\newtheorem{notn}[theorem]{Notation} %

\newcommand{\Tr}{\mathsf{Tr}} %
\newcommand{\Trl}{\Lambda}

\newcommand{\synt}{\Sigma}
\newcommand{\stsc}{\textbf{xTCL}\xspace}
\newcommand{\stlc}{\textbf{STLC}\xspace}

\newcommand{\Ty}{\mathsf{Ty}}
\newcommand{\Tyl}{\mathsf{Ty}}

\newcommand{\arty}[2]{#1 \rightarrowtriangle #2} %

\newcommand{\Pred}[2][]{\mathbf{Pred}_{#1}(#2)}

\newcommand{\pred}[1]{\overset{#1}{\rightarrowtail}}
\newcommand{\utype}{\mathsf{unit}}

\newcommand{\sto}{\twoheadrightarrow}

\usepackage{tensor}

\newcommand{\iimg}[2]{#1{}^\star[#2]}
\newcommand{\fimg}[2]{#1{}_\star[#2]}

\newcommand{\DDar}{\rotatebox[origin=c]{-90}{$\Rrightarrow$}}

\newcommand{\logp}{\square}
\newcommand{\invp}{\square}

\newcommand{\tensorlift}{\mathbin{\ol\mon}}

\theoremstyle{definition}
\newtheorem{notation}[theorem]{Notation}
\newtheorem{assumptions}[theorem]{Assumptions}

\usepackage{scalerel,stackengine,graphicx}
\renewcommand{\blacksquare}{{\ThisStyle{\ensurestackMath{%
  \stackinset{c}{}{c}{-1.5\LMpt}%
  {\SavedStyle\scaleobj{1.5}{\cdot}}{\SavedStyle\square}}}}}

\begin{document}\allowdisplaybreaks
\title{Logical Predicates in Higher-Order\\ Mathematical Operational Semantics}
\author{Sergey Goncharov\inst{1}${}^{,}$\thanks{Supported by Deutsche Forschungsgemeinschaft (DFG, German Research Foundation) – project number 501369690}
     , Alessio Santamaria\inst{2}
     , Lutz Schröder\inst{1}${}^{,}$\thanks{Supported by Deutsche Forschungsgemeinschaft (DFG, German Research Foundation) – project numbers 419850228}
     , \\ Stelios Tsampas\inst{1}${}^{,}$\thanks{Supported by Deutsche
       Forschungsgemeinschaft (DFG, German Research Foundation) – project
       numbers 419850228 and 527481841}
   and Henning Urbat\inst{1}${}^{,}$\thanks{Supported by Deutsche Forschungsgemeinschaft (DFG, German Research Foundation) – project number 470467389}} 

\authorrunning{S.~Goncharov, A.~Santamaria, L.~Schröder, S.~Tsampas, H.~Urbat}
\institute{Friedrich-Alexander-Universität Erlangen-Nürnberg
  \email{$\{$name.surname$\}$@fau.de}
  \and University of Sussex
  \email{a.santamaria@sussex.ac.uk}}

\maketitle \thispagestyle{empty}

\begin{abstract}
  We present a systematic approach to logical predicates based on
  universal coalgebra and higher-order abstract GSOS, thus making a
  first step towards a unifying theory of logical relations. We first observe
  that logical predicates are special cases of
  \emph{coalgebraic invariants} on mixed-variance functors. We
  then introduce the notion of a \emph{locally maximal logical
    refinement} of a given predicate, with a view to enabling
  inductive reasoning, and identify sufficient conditions on the
  overall setup in which locally maximal logical refinements
  canonically exist. Finally, we develop induction-up-to techniques
  that simplify inductive proofs via logical predicates on systems
  encoded as (certain classes of) higher-order GSOS laws by
  identifying and abstracting away from their boiler-plate part.
\end{abstract}
\section{Introduction}

Logical relations are %
arguably the most widely used method for reasoning on higher-order
languages. Historically, early examples of logical
relations~\cite{MILNER1978348,DBLP:journals/iandc/OHearnR95,ongthesis,PITTS199666,plotkin1973lambda,sieber_1992,STATMAN198585,DBLP:journals/jsyml/Tait67} were based on
denotational semantics, before the method evolved into logical
relations based on operational
semantics~\cite{DBLP:journals/toplas/AppelM01,5230591,10.1145/2103621.2103666,DBLP:conf/lics/Pitts96,DBLP:journals/mscs/Pitts00,10.5555/645722.666533}.
Today, operationally-based logical
relations are ubiquitous and serve purposes ranging from strong
normalization proofs~\cite{altenkirch_et_al:LIPIcs:2016:5972} and
safety properties~\cite{10.1145/3623510,10.1145/3408996} to reasoning
about contextual equivalence~\cite{10.1007/11693024_6,10.1145/3158152} and formally verified
compilation~\cite{10.1145/1596550.1596567,10.1145/1926385.1926402,DBLP:conf/icfp/NewBA16,10.1145/3434302}, in a variety of
settings such as effectful~\cite{5571705},
probabilistic~\cite{10.1145/3571195,10.1007/978-3-662-46678-0_18,10.1145/3236782}, and
differential programming~\cite{10.1016/j.tcs.2021.09.027,DBLP:journals/pacmpl/LagoG22,dallago_et_al:LIPIcs:2019:10687}.

Unfortunately, despite the extensive literature, there is a distinct lack of
a general formal theory of (operational) logical relations. As a reasoning
method, logical relations are applied in a largely empirical manner,
more so because their core
principles are well understood on an intuitive level. For example, there is
typically no formal notion of a logical predicate or relation; instead, if a
predicate or relation is defined by induction on types and maps ``related inputs
to related outputs'', it then meets the informal criteria to be called
``logical''. However, the empirical character of logical relations is
problematic for two main reasons: (i) complex machinery associated to
logical relations needs to be re-established anew on a per-case basis, and (ii)
it is hard to abstract  and simplify said
machinery, even though certain parts of proofs via logical relations seem
generic.

Recently, \emph{Higher-order Mathematical Operational
  Semantics}~\cite{gmstu23}, or \emph{higher-order abstract GSOS}, has
emerged as a unifying approach to the operational semantics of
higher-order languages. In this framework, languages are represented
as \emph{higher-order GSOS laws}, a form of distributive law of a
syntax functor $\Sigma$ over a mixed-variance behaviour bifunctor
$B$. In further work~\cite{UrbatTsampasEtAl23}, an abstract form of
\emph{Howe's method}~\cite{DBLP:conf/lics/LagoGL17,DBLP:conf/lics/Howe89,DBLP:journals/iandc/Howe96} for higher-order
abstract GSOS has been identified, in which an otherwise complex and
application-specific operational technique is, at the same time,
lifted to an appropriate level of generality and reduced to a simple
\emph{lax bialgebra}
condition.%

In the present paper, we work towards establishing a theory of logical
relations based on coalgebra and higher-order abstract GSOS,
starting from \emph{logical predicates} (or \emph{unary} logical relations).
In more detail, we present the following contributions:

\begin{enumerate}[label=(\roman*)]
\item A systematization of the method of logical
  predicates (\Cref{sec:logpred}), achieved by
  \begin{enumerate}
  \item identifying logical predicates as certain coalgebraic
    invariants~(\Cref{def:logpred}), parametric in a predicate lifting of the underlying mixed-variance bifunctor,
  \item introducing the \emph{locally maximal logical refinement} $\logp P$ of a
	predicate $P$ (\Cref{def:refinement}), which enables
	inductive proofs of $\logp P$, and
      \item identifying an abstract setting in which locally
        maximal logical refinements of predicates exist and are unique
        (\Cref{sec:constructing-logical-predicates}).
  \end{enumerate}
\item The development of efficient reasoning techniques on logical
  predicates, which we call \emph{induction up-to} (Theorems
  \ref{th:main2} and \ref{thm:ind-up-to-blackquare}), for higher-order
  GSOS laws satisfying a \emph{relative flatness}
  condition~(\Cref{def:relativelyflat}).
\end{enumerate}
We illustrate (ii) by providing proofs of strong
normalization for typed combinatory logic and type safety for the
simply typed $\lambda$-calculus which, thanks to the use of our up-to techniques, are significantly shorter and simpler than standard arguments found in the literature. Finally, we exploit the genericity of
our framework to study strong normalization on the level of
higher-order GSOS laws (\Cref{thm:strong-normalization}).
We note that the implementation of typed languages as higher-order GSOS
laws as such is also novel.

\paragraph{Related work} While denotational logical relations have
been studied in categorical generality,
e.g.~\cite{DBLP:journals/mscs/Goubault-LarrecqLN08,DBLP:journals/entcs/HermidaRR14,hermida1993fibrations,DBLP:phd/ethos/Katsumata05},
general abstract foundations of operational logical relations are far
less developed. In recent
work~\cite{DBLP:conf/fscd/DagninoG22,DBLP:journals/corr/abs-2303-03271},
Dagnino and Gavazzo introduce a categorical notion of operational
logical relations that is largely orthogonal to ours, in particular
regarding the parametrization of the framework: In \emph{op.\ cit.},
the authors work with a fixed \emph{fine-grain
  call-by-value} language~\cite{LevyPowerEtAl03}, parametrized by a signature
of generic effects, while the notion of logical relation is kept
variable and in fact is parametrized over a fibration; contrastingly, we
keep to the traditional notion of logical relation but parametrize
over the syntax and semantics of the language. Moreover, we work with
a small-step operational semantics, whereas the semantics used in
\emph{op.\ cit.}  is an axiomatically defined categorical evaluation
semantics.

\section{Preliminaries}
\subsection{Category Theory}\label{sec:categories}
Familiarity with basic category theory~\cite{mac2013categories} (e.g.~functors, natural
transformations, (co)limits, monads) is
assumed. We review some concepts and notation.

\medskip\noindent\emph{Notation.}
Given objects
$X_1, X_2$ in a category~$\C$, we write $X_1\times X_2$ for the
product and $\langle f_1, f_2\rangle\c X\to X_1\times X_2$ for the
pairing of $f_i\c X\to X_i$, $i=1,2$. We let
$X_1+X_2$ denote the coproduct, $\inl\c X_1\to X_1+X_2$ and
$\inr\c X_2\to X_1+X_2$ the injections, $[g_1,g_2]\c X_1+X_2\to X$ the
copairing of $g_i\colon X_i\to X$, $i=1,2$, and
$\nabla=[\id_X,\id_X]\colon X+X\to X$ the codiagonal. The \emph{slice category} $\C/X$, where $X\in \C$, has as objects all pairs $(Y,p_Y)$ of an object $Y\in \C$
and a morphism $p_Y\c Y\to X$, and a morphism from $(Y,p_Y)$ to
$(Z,p_Z)$ is a morphism $f\c Y\to Z$ of $\C$ such that
$p_Y = p_Z\comp f$. The \emph{coslice category} $X/\C$ is defined dually.

\medskip\noindent\emph{Extensive categories.} A category $\C$ is
\emph{(finitely) extensive}~\cite{cbl93} if it has finite coproducts
and for every finite family of objects $X_i$ ($i\in I$) the functor
$E\colon \prod_{i\in I} \C/X_i \to \C/\coprod_{i\in I} X_i$ sending
$(p_i\colon Y_i\to X_i)_{i\in I}$ to
$\coprod_{i\in I} p_i\colon \coprod_i Y_i \to \coprod_i X_i$ is an
equivalence of categories. A \emph{countably extensive} category
satisfies the analogous property for countable coproducts. In
extensive categories, coproduct injections $\inl,\inr$ are monic, and
coproducts of monomorphisms are monic; generally, coproducts behave
like disjoint unions of sets.

\begin{expl}\label{ex:categories}
Examples of countably extensive categories include the category $\Set$ of sets and functions; the category $\Set^{\C}$ of presheaves on a small category $\C$ and natural transformations; and the categories of posets and monotone maps, nominal sets and equivariant maps, and metric spaces and non-expansive maps, respectively.
\end{expl}

\noindent\emph{Algebras.}
Given an endofunctor $F$ on a category $\C$, an \emph{$F$-algebra} is
a pair $(A,a)$ consisting of an object~$A$ and a morphism
$a\colon FA\to A$ (the \emph{structure}). A \emph{morphism} from
$(A,a)$ to an $F$-algebra $(B,b)$ is a morphism $h\colon A\to B$
of~$\C$ such that $h\comp a = b\comp Fh$. Algebras for~$F$ and their
morphisms form a category $\Alg(F)$, and an \emph{initial} $F$-algebra
is simply an initial object in that category.  We denote the initial
$F$-algebra by $\mu F$ if it exists, and its structure by
$\ini\colon F(\mu F) \to \mu F$. Initial algebras admit the
\emph{structural induction principle}: the algebra $\mu F$ has no
proper subalgebras, that is, every $F$-algebra monomorphism
$m\colon (A,a)\monoto (\mu F,\ini)$ is an isomorphism.

More generally, a \emph{free $F$-algebra} on an object $X$ of $\C$ is an
$F$-algebra $(F^{\star}X,\iota_X)$ together with a morphism
$\eta_X\c X\to F^{\star}X$ of~$\C$ such that for every algebra $(A,a)$
and every $h\colon X\to A$ in $\C$, there exists a unique
$F$-algebra morphism $h^\klstar\colon (F^{\star}X,\iota_X)\to (A,a)$
such that $h=h^\klstar\comp \eta_X$. If free algebras
exist on every object, their formation induces a monad
$F^{\star}\colon \C\to \C$, the \emph{free monad} generated by~$F$. Every $F$-algebra $(A,a)$ yields an
Eilenberg-Moore algebra $\wh{a} \colon F^{\star} A \to A$ as the free
extension of $\id_A\c A\to A$.

The most familiar example of functor algebras are algebras for a
signature. Given a set $S$ of \emph{sorts}, an \emph{$S$-sorted algebraic signature} consists of a set~$\Sigma$
of operation symbols together with a map $\ar\colon \Sigma\to S^{\star}\times S$
associating to every $\f\in \Sigma$ its \emph{arity}. We write $\f\colon s_1\times\cdots\times s_n\to s$ if $\ar(\f)=(s_1,\ldots,s_n,s)$, and $\f\colon s$ if $n=0$ (in which case $\f$ is called a \emph{constant}). Every
signature~$\Sigma$ induces a polynomial functor on the category $\Set^S$ of $S$-sorted sets, denoted by the
same letter $\Sigma$, given by $(\Sigma X)_s = \coprod_{\f\colon s_1\cdots s_n\to s} \prod_{i=1}^n X_{s_i}$ for $X\in \Set^S$ and $s\in S$. An algebra for the functor $\Sigma$ is
precisely an algebra for the signature $\Sigma$, viz.~an $S$-sorted set $A=(A_s)_{s\in S}$ in $\Set^S$ equipped with an operation $\f^A\colon \prod_{i=1}^n A_{s_i}\to A_s$ for every $\f\colon s_1\cdots s_n\to s$ in $\Sigma$. Morphisms of $\Sigma$-algebras are $S$-sorted
maps respecting the algebraic structure. Given an $S$-sorted set $X$ of
variables, the free algebra $\Sigmas X$ is the $\Sigma$-algebra of
$\Sigma$-terms with variables from~$X$; more precisely, $(\Sigmas X)_s$ is inductively defined by $X_s\seq (\Sigmas X)_s$ and $\f(t_1,\ldots,t_n)\in (\Sigmas X)_s$ for all $\f\colon s_1\cdots s_n\to s$ and $t_i\in (\Sigmas X)_{s_i}$. In particular, the free
algebra on the empty set is the initial algebra $\mu \Sigma$; it is
formed by all \emph{closed terms} of the signature. For every
$\Sigma$-algebra $(A,a)$, the induced Eilenberg-Moore algebra
$\wh{a}\colon \Sigmas A \to A$ is given by the map that evaluates terms
over~$A$ in the algebra~$A$.

\medskip\noindent\emph{Coalgebras.}
Dual to the notion of algebra, a \emph{coalgebra} for an
endofunctor $F$ on $\C$ is a pair $(C,c)$ of an object $C$ (the
\emph{state space}) and a morphism $c\colon C\to FC$ (the
\emph{structure}).

\subsection{Higher-Order Abstract GSOS}

We summarize the framework of higher-order abstract
GSOS~\cite{gmstu23}, which extends the original, first-order
counterpart introduced by Turi and
Plotkin~\cite{DBLP:conf/lics/TuriP97}. In higher-order abstract GSOS,
the operational semantics of a higher-order language is presented in the form of a
\emph{higher-order GSOS law}, a categorical structure parametric in
\begin{enumerate}
\item a category $\C$ with finite products and coproducts;
\item an object $\Pt \in \C$ of \emph{variables};
\item an endofunctor $\Sigma \c \C \to \C$, where $\Sigma = \Pt + \Sigma'$
  for some endofunctor~$\Sigma'$, such that free $\Sigma$-algebras exist on every object (hence $\Sigma$ generates a free monad $\Sigmas$);
\item a mixed-variance bifunctor $B\colon
  \C^\opp\times \C\to \C$.
\end{enumerate}
The functors $\Sigma$ and $B$ represent the \emph{syntax} and the \emph{behaviour} of a higher-order language. The motivation behind $B$ having two arguments is
that transitions have labels, which behave contravariantly, and
poststates, which behave covariantly; in term models the
objects of labels and states will coincide. The presence of
an object $V$ of variables is a technical requirement for the
modelling of languages with variable
binding~\cite{DBLP:conf/lics/FiorePT99,DBLP:conf/lics/FioreT01}, such
as the $\lambda$-calculus. An object of
$V/\C$, the coslice category of \emph{$\Pt$-pointed objects}, is
thought of as a set $X$ of programs with an embedding $p_X\c V\to X$
of the variables. In point-free calculi, e.g.\ \stsc
as introduced below, we put $V=0$ (the initial object).
\begin{defn}\label{def:ho-gsos-law}
  A \emph{($\Pt$-pointed) higher-order GSOS law} of $\Sigma$ over $B$
  is a family of morphisms~\eqref{eq:ho-gsos-law} that is dinatural in $(X,p_X) \in \Pt/\C$ and natural in $Y\in \C$:
  \begin{align}\label{eq:ho-gsos-law}
    \rho_{(X,p_X),Y} \c \Sigma (X \times B(X,Y))\to B(X, \Sigma^\star (X+Y))
  \end{align}
\end{defn}

\begin{notn}\label{not:rho}
\begin{enumerate}[label=(\roman*)]
\item In~\eqref{eq:ho-gsos-law}, we have implicitly applied the forgetful
  functor $\Pt/\C \to C$ at $(X,p_X)$. In addition, we write $\rho_{X,Y}$ for
  $\rho_{(X,p_X),Y}$ if the point $p_X$ is clear from the context.
\item For $(A,a)\in \Alg(\Sigma)$, we view $A$ as {$\Pt$-pointed} by
  $p_A = \bigl(\Pt\xra{\inl} \Pt+\Sigma' A = \Sigma  A \xra{a} A\bigr)$.
\end{enumerate}
\end{notn}
\noindent Informally,~$\rho_{X,Y}$ assigns to an operation of
the language with formal arguments from~$X$ having specified
next-step behaviours in~$B(X,Y)$ (i.e.\ with labels in~$X$ and formal
poststates in~$Y$) a next-step behaviour in
$B(X, \Sigma^\star (X+Y))$, i.e.\ with the same labels, and with
poststates being program terms mentioning variables from both~$X$
and~$Y$. Every higher-order GSOS law~\eqref{eq:ho-gsos-law} induces a canonical \emph{operational
model} $\gamma
\c \mS \to B(\mS,\mS)$, viz.\ a $B(\mS,-)$-coalgebra on the initial algebra $\mS$, defined by \emph{primitive
  recursion}~\cite[Prop. 2.4.7]{DBLP:books/cu/J2016} as the unique morphism~$\gamma$ making the following diagram commute:
\begin{equation*}
  \begin{tikzcd}[column sep=9ex, row sep=normal]
    \Sigma(\mS)
    \dar[dashed, "\Sigma \brks{\id,\,\gamma}"']
    \ar[rr,"\iota"]
    &[-2ex] &[2ex]    \mS
    \dar[dashed, "\gamma"]
    \\
    \Sigma (\mS\times {B(\mS,\mS)})
    \rar["{\rho_{\mS,\mS}}"]
    &
    B(\mS,\Sigma^\star(\mS+\mS))
    \rar["{B(\mS,\hat\iota\comp\Sigmas\nabla)}"] &
    B(\mS,\mS)
  \end{tikzcd}
\end{equation*}
Here, we regard the initial algebra $(\mS,\ini)$ as $V$-pointed as explained in \Cref{not:rho}.

\paragraph{Simply Typed SKI Calculus.}
We illustrate the ideas behind higher-order abstract GSOS with 
an extended version of the simply typed SKI calculus~\cite{hindley2008lambda}, a typed combinatory logic 
which we call \stsc. It is expressively equivalent to the simply typed $\lambda$-calculus but does not use variables; hence it avoids the complexities associated to variable binding and
substitution in the $\lambda$-calculus, which we treat in \Cref{sec:lambda-laws}. The set~$\Ty$ of \emph{types} is inductively defined as
\begin{equation}\label{eq:type-grammar}
  \Ty \Coloneqq \utype \mid \arty{\Ty}{\Ty}.
\end{equation}
The constructor $\arty{}{}$ is right-associative, i.e.\
$\arty{\tau_{1}}{\arty{\tau_{2}}{\tau_{3}}}$ is parsed as
$\arty{\tau_{1}}{({\arty{\tau_{2}}{\tau_{3}}})}$. The terms of $\stsc$ are formed over the $\Ty$-sorted signature $\Sigma$ whose operation symbols are listed below, with
$\tau,\tau_1,\tau_2,\tau_3$ ranging over all types in $\Ty$:
\begin{align*}
& \mathsf{e} \c \utype &&  \mathsf{app}_{\tau_1,\tau_2}\colon (\arty{\tau_{1}}{\tau_{2}})\times \tau_1\to \tau_2  \\
& S_{\tau_{1},\tau_{2},\tau_{3}} \c \arty{(\arty{\tau_{1}}{\arty{\tau_{2}}{\tau_{3}}})}{\arty{(\arty{\tau_{1}}{\tau_{2}})}{\arty{\tau_{1}}{\tau_{3}}}} && K_{\tau_{1},\tau_{2}} \c \arty{\tau_{1}}{\arty{\tau_{2}}{\tau_{1}}} \\
& S'_{\tau_1,\tau_2,\tau_3}\c (\arty{\tau_{1}}{\arty{\tau_{2}}{\tau_{3}}})\to (\arty{{(\arty{\tau_{1}}{\tau_{2})}}}{\arty{\tau_{1}}{\tau_{3}}})  && K'_{\tau_1,\tau_2}\colon \tau_1\to (\arty{\tau_{2}}{\tau_{1}}) \\
& S''_{\tau_1,\tau_2,\tau_3}\c (\arty{\tau_{1}}{\arty{\tau_{2}}{\tau_{3}}})\times (\arty{\tau_{1}}{\tau_{2}})\to (\arty{\tau_{1}}{\tau_{3}}) &&  I_{\tau} \c \arty{\tau}{\tau}
\end{align*}
We let $\Tr=\mS$ denote the $\Ty$-sorted set of closed $\Sigma$-terms. Informally, $\mathsf{app}$ represents function application (we write $s\, t$ for $\mathsf{app}(s,t)$), 
and the constants $I_\tau$, $K_{\tau_1,\tau_2}$, $S_{\tau_1,\tau_2,\tau_3}$ represent the $\lambda$-terms $\lambda t.\,t$, $\lambda t.\,\lambda s.\, t$ and $\lambda t.\,\lambda s.\,\lambda u.\, (s\, u)\, (t\, u)$, respectively.
\begin{figure*}[t]
  \begin{gather*}
    \inference{}{\mathsf{e} \xto{\checkmark}}
    \qquad
    \inference{}{S_{\tau_{1},\tau_{2},\tau_{3}}\xto{t}S'_{\tau_{1},\tau_{2},\tau_{3}}(t)}
    \qquad
    \inference{}{S'_{\tau_{1},\tau_{2},\tau_{3}}(p)\xto{t}S''_{\tau_{1},\tau_{2},\tau_{3}}(p,t)}     \\[2ex]
    \inference{}{S''_{\tau_{1},\tau_{2},\tau_{3}}(p,q)\xto{t}(p\app{\tau_{1}}{\arty{\tau_{2}}{\tau_{3}}} t)\app{\tau_{2}}{\tau_{3}} (q\app{\tau_{1}}{\tau_{2}} t)}
    \qquad
    \inference{}{K_{\tau_{1},\tau_{2}}\xto{t}K'_{\tau_{1},\tau_{2}}(t)}
    \qquad
    \inference{}{K'_{\tau_{1},\tau_{2}}(p)\xto{t}p}
    \\[1ex]
    \qquad
    \inference{}{I_{\tau}\xto{t}t}
    \qquad
    \inference{p\to p'}{p \app{\tau_{1}}{\tau_{2}} q\to p' \app{\tau_{1}}{\tau_{2}} q}
    \qquad
    \inference{p\xto{q} p'}{p \app{\tau_{1}}{\tau_{2}} q\to p'}
  \end{gather*}
  \caption{(Call-by-name) operational semantics of \stsc.}
  \label{fig:skirules}
\end{figure*}
The operational semantics of \stsc involves three kinds of transitions:
$\xto{\checkmark}$, $\xto{t}$ and $\xto{}$. It is presented in
\Cref{fig:skirules}; here, $p,p',q,t$ range over terms in $\Tr$ of appropriate type. Intuitively, $s\xto{\checkmark}$ identifies~$s$ 
as an explicitly irreducible term; $s\xto{t} r$ states that $s$ acts as a function mapping $t$ to~$r$; and $s\to t$ indicates that~$s$ reduces to
$t$. Our use of labelled transitions in higher-order operational semantics is inspired by work on bisimilarity in the $\lambda$-calculus~\cite{Abramsky:lazylambda,DBLP:journals/tcs/Gordon99}.
The use of $K'$, $S'$ and $S''$ does not
impact the behaviour of programs, except for possibly adding more unlabelled
transitions. For example, the standard rule r $S t{} s{} e\to
(te)(se)$ for the $S$-combinator is rendered as the chain of transitions $ S t{} s{} e\to S'(t)\, s e\to S''(t,s)\, e\to (te)(se).$
The transition system for \stsc is deterministic: for every term $s$, either $s \xto{\checkmark}$, or there exists a unique $t$ such that $s\to t$, or for each appropriately typed $t$ there exists a unique $s_t$ such that $s\xto{t} s_t$. Therefore, given
\begin{align}
  \label{eq:beh}
  B_\tau(X,Y) &= Y_\tau + D_\tau(X,Y),\\
  \label{eq:val} \qquad
  D_{\utype}(X,Y) &= \term = \{\ast\}\qquad \text{and} \qquad
                     D_{\arty{\tau_{1}}{\tau_{2}}}(X,Y) = Y_{\tau_{2}}^{X_{\tau_{1}}},
\end{align}
the operational rules in \Cref{fig:skirules} determine a $\Set^\Ty$-morphism 
$\gamma \c \Tr \to B(\Tr,\Tr)$:
\begin{flalign}
\hspace{4em}&& \gamma_\utype(s)=&\;\inr(*) && \text{if $s \xto{\checkmark}$ where $s\c \utype$,} \notag\\
&&\gamma_\tau(s)=&\;\inl(t) && \text{if $s \xto{} t$ where $s,t\c\tau$,}\label{exa:gamma}\\
&&\gamma_{\arty{\tau_{1}}{\tau_{2}}}(s) =&\; \inr(\lambda t.\,s_t) &&\text{if $s \xto{t} s_t$ for $s \c \arty{\tau_{1}}{\tau_{2}}$ and $t\c \tau_1$.}&\hspace{3em}\notag
\end{flalign}
\begin{prop}
  \label{prop:bif}
  The object assignments \eqref{eq:beh} and \eqref{eq:val} extend to mixed-variance bifunctors
  \begin{equation}
    \label{eq:bif}
    B,D \c (\Set^{\Ty})^{\opp} \times \Set^{\Ty} \to \Set^{\Ty}.
  \end{equation}
\end{prop}
The semantics of \stsc{} in \Cref{fig:skirules} corresponds to a
($0$-pointed) higher-order GSOS law of the syntax functor
$\Sigma$ over the behaviour bifunctor $B$, i.e.\ to a family of maps~\eqref{eq:ho-gsos-law}
dinatural in $X\in \Set^\Ty$ and natural in $Y\in \Set^\Ty$. The maps $\rho_{X,Y}$ are cotuples defined by distinguishing cases on the constructors $\mathsf{e}, S, S', S'', K, K', I,
\mathsf{app}$ of \stsc{}, and each component of $\rho$ is determined by the 
rules that apply to the corresponding constructor. We provide a few illustrative cases; see~\Cref{app:omitted-proofs}, p.~\pageref{page:rho-stsc}, for a complete definition.
\begin{flalign}
  \rho_{X,Y} \c \Sigma(X \times B(X,Y)) 												& \to B(X, \Sigma^\star (X+Y)) && \\[1ex]
  \hspace{3em} \rho_{X,Y}~   (S''_{\tau_{1},\tau_{2},\tau_{3}}((p,f),(q,g))) & =   \lambda t.\,(p \app{}{} t)\app{}{}(q \app{}{} t)&&\\
  \rho_{X,Y}~   ((p,f)\app{}{}(q,g)) 														& =   f(q)&&\hspace{-2em}\text{if $f \c Y_{\tau_{2}}^{X_{\tau_{1}}}$}\label{eq:def-rho-cbn} \\ 
  \rho_{X,Y}~   ((p,f)\app{}{}(q,g)) 														& =   f q &&\hspace{-2em}\text{if $f \c Y_{\arty{\tau_{1}}{\tau_{2}}}$}&&\hspace{3em}
\end{flalign}
The
operational model $\gamma \c \Tr \to B(\Tr,\Tr)$ of $\rho$ coincides with the coalgebra  \eqref{exa:gamma}.

\begin{rem}\label{rem:cbv-nondet}
The rules for application in \Cref{fig:skirules} implement the call-by-name
evaluation strategy. Other strategies can be captured by varying the rules and
consequently the corresponding higher-order GSOS law. For the call-by-value
strategy, one replaces the last rule with	 \eqref{eq:rule-mod-1} and
\eqref{eq:rule-mod-2} below and modifies clause \eqref{eq:def-rho-cbn} in the
definition of $\rho$ accordingly. One can also model the traditional view of 
combinatory logic as a rewrite system~\cite{hindley2008lambda} where any redex
can be reduced, no matter how deeply. This amounts to specifying a maximally
nondeterministic strategy by adding the rule \eqref{eq:rule-mod-3} below to
\Cref{fig:skirules}.  Notably, this makes the operational model
nondeterministic, and hence the corresponding  higher-order GSOS law relies on
the behaviour functor $\Pow B$ instead of the original $B$ given
by~\eqref{eq:beh}, where $\Pow$ is the powerset functor.\\
\begin{minipage}{.33\textwidth}
\begin{equation}\label{eq:rule-mod-1}
\inference{p\xto{t} p'\quad q\to q'}{p \app{\tau_{1}}{\tau_{2}} q\to p \app{}{} q'}
\end{equation}
\end{minipage}
~
\begin{minipage}{.35\textwidth}
\begin{equation}\label{eq:rule-mod-2}
\inference{p\xto{q} p'~\quad q\xto{t}q'}{p \app{\tau_{1}}{\tau_{2}} q\to p'}
\end{equation}
\end{minipage}
~
\begin{minipage}{.27\textwidth}
\begin{equation}\label{eq:rule-mod-3}
\inference{q\to q'}{p \app{\tau_{1}}{\tau_{2}} q\to p \app{}{} q'}
\end{equation}
\end{minipage}
\end{rem}

\section{Coalgebraic Logical Predicates}
\label{sec:logpred}

\subsection{Predicate Lifting}
\label{subsec:predlift}

Predicates and relations
on coalgebras are often most conveniently modelled through \emph{predicate} and \emph{relation
  liftings}~\cite{DBLP:journals/jlp/KurzV16} of the underlying type functors. In the following we introduce a framework of
predicate liftings for mixed-variance bifunctors, adapting existing notions of relation
lifting~\cite{UrbatTsampasEtAl23}, which enables reasoning about
``higher-order'' coalgebras, such as operational models of higher-order GSOS laws. The following global
assumptions ensure that predicates and relations behave in an expected manner:
\begin{assumptions}
  \label{assumptions}
  From now on, we fix $\C$ to be a complete, well-powered and
  extensive category in which, additionally, strong epimorphisms are stable
  under pullbacks.
\end{assumptions}
\noindent The categories of \Cref{ex:categories} satisfy these assumptions. \noindent Since $\C$ is complete and well-powered, every morphism $f$ admits a
(strong epi, mono)-factorization $f=m\comp e$~\cite[Prop. 4.4.3]{borceux94}; we call $m$ the \emph{image} of $f$. The category $\Pred{\C}$ of
\emph{predicates} over $\C$ has as objects all monics (predicates) $P \pred{} X$ from
$\C$, and as morphisms $(p \c P \pred{} X) \to (q \c Q \pred{} Y)$
all pairs $(f \c X \to Y, f|_P \c P \to Q)$ such that
$q\comp f|_P = f\comp p$ (so $f|_P$ is uniquely determined by~$f$).
(Co)products in $\Pred{\C}$ are lifted from
$\C$. The \emph{fiber
  $\Pred[X]{\C}$} is the subcategory of all monics $P\pred{} X$
for fixed~$X$ and morphisms $(\id_X,-)$. It is is preordered by $p\leq q$ if~$p$ factors
through $q$; identifying $p,q$ if $p\leq q$ and $q\leq p$, we regard $\Pred[X]{\C}$ as a poset. Since $\C$ is complete and well-powered, $\Pred[X]{\C}$ is a complete lattice; we write 
$\bigwedge$ for meets (i.e.\ pullbacks) and $\bigvee$ for joins. We will also write $\iimg{f}{P}$ 
for the \emph{inverse image} of a predicate $p\c P\pred{} X$ under $f\c Y\to X$, i.e.\ 
the pullback of $p$ along $f$. The \emph{direct image} $\fimg{f}{Q}$ of $q\c Q\pred{} Y$ 
under $f\c Y\to X$ is the image of the composite $f\comp p\c Q\to X$.
This yields an adjunction between $\Pred[X]{\C}$ and $\Pred[Y]{\C}$, i.e.\ $Q\leq\iimg{f}{P}$
iff $\fimg{f}{Q}\leq P$.

A \emph{predicate lifting} of an endofunctor ${\Sigma \c \C \to \C}$  is an endofunctor   $\ol{\Sigma}\c
\Pred{\C}\to \Pred{\C}$ making the left-hand diagram below commute; similarly, a \emph{predicate lifting} of a 
mixed-variance bifunctor
${B\c \C^\opp\times \C\to \C}$ is a bifunctor $\ol{B} \c
\Pred{\C}^\opp\times \Pred{\C} \to \Pred{\C}$ making the right-hand diagram below commute. Here $\under{-}$ is the forgetful functor sending $p\c P\pred{} X$ to~$X$.
\begin{equation}
  \label{eq:liftingd}
  \begin{tikzcd}
    \Pred{\C} \ar{d}[swap]{\under{-}}  \ar{r}{\ol{\Sigma}} & \Pred{\C} \ar{d}{\under{-}}  \\
    \C \ar{r}{\Sigma}  & \C
  \end{tikzcd}
  \qquad
  \begin{tikzcd}
    \Pred{\C}^\opp\times \Pred{\C}
    \ar{d}[swap]{\under{-}^\opp\times \under{-}} \ar{r}{\ol{B}} & \Pred{\C} \ar{d}{\under{-}}  \\
    \C^\opp \times \C \ar{r}{B}  & \C
  \end{tikzcd}
\end{equation}
		We denote by $\overline{\Sigma}$ both the action on predicates
		and on the corresponding objects in $\C$, i.e.\ $\overline{\Sigma}(p\c P\pred{}
		X)\c \overline{\Sigma} P\pred{} \Sigma X$.
Every endofunctor $\Sigma$ on $\C$ admits a canonical predicate lifting $\ol{\Sigma}$ mapping $p \c P \pred{} X$ to the image $\overline{\Sigma}p \c {\overline{\Sigma}P \pred{} \Sigma X}$ of $\Sigma p\c \Sigma P\to \Sigma X$~\cite{DBLP:books/cu/J2016}. Note that $\ol{\Sigma}p=\Sigma p$ if $\Sigma$ preserves monos. In the remainder we will only consider canonical liftings of endofunctors.

\begin{prop}\label{prop:free-monad-lift}
If $\Sigma$ preserves strong epis, then
\( \ol{\Sigma}^\star = \ol{\Sigmas}. \) 
\end{prop}

The canonical predicate liftings for mixed-variance bifunctors are slightly more complex. Similarly to the case of relation liftings of such functors developed in recent work~\cite{UrbatTsampasEtAl23}, their construction involves suitable pullbacks.
\begin{prop}
  \label{prop:liftingb}
Every bifunctor
  $B\c \C^\opp\times \C \to \C$ admits a canonical predicate lifting $\ol{B}\colon \Pred{\C}^\op \times \Pred{\C}\to \Pred{\C}$ sending $(p\colon P\monoto X,\, q\colon Q\monoto Y)$ to the predicate $m_{P,Q}\colon \ol{B}(P,Q)\monoto B(X,Y)$, the image of the morphism $r_{P,Q}$ given by the pullback below:
  \begin{equation}
    \label{eq:liftingpb}
    \begin{tikzcd}[column sep=tiny, row sep=.3ex]
      &[1em]& {T_{P,Q}}
    \pullbackangle{-45}
    &&&& {B(P,Q)} \\
    \\
    {\overline{B}(P,Q)} \\
    \\
    && {B(X,Y)} &&&& {B(P,Y)}
    \arrow["{e_{P,Q}}"', two heads, from=1-3, to=3-1]
    \arrow["{m_{P,Q}}"', tail, pos=.8, from=3-1, to=5-3]
    \arrow["s_{P,Q}", from=1-3, to=1-7]
    \arrow["r_{P,Q}"', from=1-3, to=5-3]
    \arrow["{B(\id,q)}", from=1-7, to=5-7]
    \arrow["{B(p,\id)}", from=5-3, to=5-7]
  \end{tikzcd}
\end{equation}
\end{prop}
\noindent If ${B}$ preserves monos in the covariant argument, then $B(\id,q)$ is monic and, since monos are pullback-stable, $\ol{B}(P,Q)$ is simply the predicate $r_{P,Q}\c T_{P,Q}\monoto B(X,Y)$.

\begin{expl}
  \label{ex:lifting}
   The bifunctors $B$ and $D$ of \eqref{eq:beh} and \eqref{eq:val} have canonical predicate liftings 
\begin{align}
     \ol{B}_\tau(P,Q) & = Q_\tau + \overline{D}_\tau(P,Q)\quad\text{where} \\
\overline{D}_{\utype}(P,Q)
      &=\term, \quad
      \overline{D}_{\arty{\tau_{1}}{\tau_{2}}}(P,Q)
      = \{f\c {X_{\tau_{1}}}\to {Y_{\tau_{2}}}\mid\forall x
         \in P_{\tau_{1}}.\, f(x)\in Q_{\tau_{2}} \}\subseteq Y_{\tau_{2}}^{X_{\tau_{1}}}. \label{eq:val-lift}
\end{align}
\end{expl}

\noindent Predicate liftings allow us to generalize \emph{coalgebraic
  invariants}~\cite[\textsection 6.2]{DBLP:books/cu/J2016}, viz.\
predicates on the state space of a coalgebra that are closed under the
coalgebra structure in a suitable sense, from endofunctors to
mixed-variance bifunctors:

\begin{notation}
For the remainder of the paper, we fix a mixed-variance bifunctor $B\colon \C^\op \times \C\to \C$ and a predicate lifting $\ol{B}\colon \Pred{\C}^\op\times \Pred{\C}\to \Pred{\C}$.
\end{notation}
\begin{defn}[Coalgebraic invariant]
  \label{def:coalginv}
 Let
  $c \c Y \to B(X,Y)$ be a $B(X,-)$-coalgebra. Given predicates $S \pred{} X$,
  $P \pred{} Y$, we say that $P$ is an \emph{$S$-relative
    ($\ol{B}$-)invariant (for~$c$)} if $P \leq \iimg{c}{\ol{B}(S,P)}$,
  equivalently, $\fimg{c}{P} \leq \ol{B}(S,P)$. (Mention of~$\ol{B}$
  is usually omitted.)
\end{defn}
\noindent Coalgebraic invariants will feature centrally in our notion
of logical predicate.

\subsection{Logical Predicates via Lifted Bifunctors}
\label{sec:logpred-lifted}

As a reasoning device, the method of logical predicates (which are
unary logical relations) typically applies to the
following scenario: One has an operational semantics on an
inductively defined set~$\mS$ of $\synt$-terms and a target predicate
$P \pred{} \mS$ to be proved, in the sense that one wants to show
$P = \mS$.  Logical predicates come into play when a direct proof
of $P = \mS$ by structural induction is not possible. The
classical example of such a predicate is \emph{strong
  normalization}~\cite{girard1989proofs,DBLP:journals/jsyml/Tait67}. The
idea is to strengthen~$P$, obtaining a predicate featuring a certain
``logical'' structure that does allow for a proof by
induction. We now develop this scenario in our abstract bifunctorial setting.

\begin{defn}[Coalgebraic logical predicate]
  \label{def:logpred} Suppose that
   $c \c X \to B(X,X)$ is a $B(X,-)$ coalgebra with state space $X$. A predicate
  $P \pred{} X$ is \emph{logical (for $c$)} if it is a $P$-relative
  $\ol{B}$-invariant (as per Def. \ref{def:coalginv}), i.e.\
  $P \leq \iimg{c}{\ol{B}(P,P)}$, equivalently,
  $\fimg{c}{P} \leq \ol{B}(P,P)$.
\end{defn}
\noindent In applications, $c$ is the operational model $\gamma\colon \mS\to B(\mS,\mS)$ of a higher-order language, or some coalgebra derived from it.
\noindent The self-referential nature of logical predicates (as
relative to themselves) is meant to cater for the property
that ``inputs in $P$ are mapped to outputs in $P$''. The following
example from \stsc illustrates this:
\begin{expl}
For $B$ given by \eqref{eq:beh} and its canonical lifting $\ol{B}$, a predicate $P \pred{} \Tr$ is logical for the operational model $\gamma \c \Tr \to B(\Tr,\Tr)$ from \eqref{exa:gamma} if $\fimg{\gamma}{P} \leq
  {\ol{B}(P,P)}$, that is,
  \begin{equation*}
    \begin{aligned}
        \fimg{(\gamma_{\utype})}{P_{\utype}} \leq&\;
       P_{\utype} + 1, \\
       \forall \tau_{1},\tau_{2}.\, \fimg{(\gamma_{\arty{\tau_{1}}{\tau_{2}}})}{P_{\arty{\tau_{1}}{\tau_{2}}}} \leq&\;
        P_{\arty{\tau_{1}}{\tau_{2}}} + \{f\c {\Tr_{\tau_{1}}}\to {\Tr_{\tau_{2}}}\mid\forall s \in P_{\tau_{1}}.\, f (s)\in P_{\tau_{2}} \},
    \end{aligned}
  \end{equation*}
  using the description of $\ol{B}$ from~\Cref{ex:lifting}. More explicitly, this means that
  \begin{itemize}
   \item if $s\in P_{\tau}$ and $s\to t$ then $t\in P_{\tau}$;
   \item if $s\in P_{\arty{\tau_{1}}{\tau_{2}}}$ and $s\xto{t} u$, then $t\in P_{\tau_1}$ implies $u\in P_{\tau_2}$.
  \end{itemize}
  As we can see in the second clause, function terms that satisfy $P$
  produce outputs that satisfy $P$ on all inputs that satisfy $P$. This is the key property of any logical predicate.
\end{expl}

Defining a suitable logical predicate (or relation) is the centerpiece of various
sophisticated arguments in higher-order settings. One standard application of logical predicates are proofs of strong normalization, which we now illustrate in the case of $\stsc$. For the operational model $\gamma\colon \Tr\to B(\Tr,\Tr)$ and terms $r,s,t$ of compatible type, put
\begin{itemize}
\item $s\To t$ if $s=s_0\to s_1\to \cdots \to s_n=t$ for some $n\geq 0$ and terms $s_0,\ldots,s_n$;
\item $s\xTo{t} r$ if $s\To s'$ and $s'\xto{t} r$ for some (unique) $s'$;
\item ${\Downarrow}(s)$ if $s\To s'$ and $\gamma(s')\in D(\Tr,\Tr)$ for some (unique) $s'$.
\end{itemize} 
Coalgebraically, this associates a \emph{weak operational model} $\wave{\gamma} \c \Tr \to
\Pow B(\Tr,\Tr)$ to $\gamma$, where $\wave{\gamma}(t)=\{t'\mid t\To t'\}\cup \{\gamma(t')\mid t\To t',\gamma(t')\in D(\Tr,\Tr)\}$.

\emph{Strong normalization} of $\stsc$ asserts that ${\Downarrow}=\Tr$: every term eventually reduces to a function or explicitly terminates. We now devise three different logical predicates on
$\Tr$, each of which provides a proof of that property. The idea is to refine the target predicate ${\Downarrow}\monoto \Tr$ to a logical predicate, for 
which showing that it is totally true will be facilitated by its invariance
w.r.t.\ a corresponding coalgebra structure. Our first example will be based on the following notion of refinement:
\begin{defn}[Locally maximal logical refinement]
  \label{def:refinement} Let $c\colon X\to B(X,X)$ be a coalgebra and let $P\monoto X$ be a predicate. A predicate $\logp{P}\monoto X$ is a \emph{locally maximal logical refinement of $P$} if (i) $\logp{P} \leq P$, (ii) $\logp{P}$ is logical (i.e.\ a $\logp{P}$-relative $\ol{B}$-invariant), and (iii) for every predicate $Q\leq P$ that is a $\logp{P}$-relative $\ol{B}$-invariant, one has $Q\leq \logp{P}$. 
\end{defn}

\begin{expl}%
  \label{ex:strongnorm1}
We define the predicate $\sqdown\monoto \Tr$, i.e.\ a family of subsets $\sqdown_{\tau}\seq \Tr_\tau$ ($\tau\in \Ty$), by induction on the structure of the type $\tau$: we put $\logp {\Downarrow}_{\utype}={\Downarrow}_\utype$, and we 
take $\logp {\Downarrow}_{\arty{\tau_{1}}{\tau_{2}}}$ to be the greatest subset of $\Tr_{\arty{\tau_{1}}{\tau_{2}}}$ satisfying
  \[
    \begin{aligned}
      \logp {\Downarrow}_{\arty{\tau_{1}}{\tau_{2}}}(t) \implies &\;
        {\Downarrow}_{\arty{\tau_{1}}{\tau_{2}}}(t) \land
        \begin{cases}
          \logp {\Downarrow}_{\arty{\tau_{1}}{\tau_{2}}}(t')
          & \text{if} \quad t \to t'  \\
          \logp {\Downarrow}_{\tau_{1}}(s) \implies \logp {\Downarrow}_{\tau_{2}}(t')
          & \text{if} \quad t \xto{s} t' 
        \end{cases}
    \end{aligned}
  \]
From this definition it is not difficult to verify by induction on the type that
\begin{equation}\label{eq:sqdown-greatest} \text{$\sqdown$ is a locally maximal logical refinement of $\Downarrow$.} \end{equation}
  Our goal is to
  show that $\sqdown$ is a subalgebra of $\mS$, equivalently $\ol{\Sigma} (\sqdown) \leq
  \iimg{\iota}{\sqdown}$, which then implies $\sqdown = \Tr$ and hence ${\Downarrow}=\Tr$ by structural
  induction.
  Taking the partition $\Sigma=\Xi+\Delta$ where $\Xi$ is the part of the 
  signature for application and $\Delta$ is the part of the signature for the remaining
  term constructors, we separately prove $\ol{\Xi} (\sqdown) \leq
  \iimg{\iota}{\sqdown}$ and $\ol{\Delta} (\sqdown) \leq
  \iimg{\iota}{\sqdown}$. It suffices to come up with 
  $\sqdown$-relative invariants $A,C\seq {\Downarrow}$ such that $\ol{\Xi}(\sqdown)\leq \iimg{\iota}{A}$
  and $\ol{\Delta}(\sqdown) \leq \iimg{\iota}{C}$. Then by \eqref{eq:sqdown-greatest} we can conclude $A,C\seq \sqdown$, so
  \[
  \ol{\Xi} (\sqdown) \leq \iimg{\iota}{A}\leq \iimg{\iota}{\sqdown}\qqand
  \ol{\Delta} (\sqdown) \leq \iimg{\iota}{C}\leq \iimg{\iota}{\sqdown}.
\]
  Let us record for further reference what it means for $Q\pred{} \Tr$ to be 
  a $\sqdown$-relative invariant contained in $\Downarrow$. Given $t\in Q_\tau$, 
  the following must hold:
\[ \text{(1) ${\Downarrow}_\tau\, t$, (2) if $t \to t'$ then $Q_{\tau}(t')$, (3) if $t \c \arty{\tau_{1}}{\tau_{2}}$ and $t \xto{s} t'$ and
    	$\sqdown_{\tau_1} s$ then $Q_{\tau_{2}}(t')$}.  \]
We first put
$A = \sqdown \lor \fimg{(\iota\comp\inl)}{\ol{\Xi}\sqdown}$, and prove (1)--(3) for $Q=A$. So let $t\in A_\tau$; we distinguish cases on the disjunction defining~$A$.
If $\sqdown_\tau\, t$, then (1)--(3) follow easily by
definition.  Otherwise, we have $t=p\app{}{} q$ such
that $\sqdown_{\arty{\tau_{1}}{\tau_{2}}}p$ and $\sqdown_{\tau_1}q$.
    \begin{enumerate}
    \item 
      By definition, $\sqdown_{\arty{\tau_{1}}{\tau_{2}}}p$ and
      $\sqdown_{\tau_1}q$ entail that $p\xTo{q} p'$ for a (unique)
      term~$p'$, and that $\sqdown_{\tau_2}p'$, hence
      $\Downarrow_{\tau_2} p'$. Since $p \app{}{}q\To p'$, it follows that 
      $\Downarrow_{\tau_2}p\app{}{} q$.
    \item We distinguish cases over the semantic rules for application:
      \begin{enumerate}
      \item $p \app{}{}q \to p' \app{}{}q$ where $p \to p'$. Then
        $\sqdown_{\arty{\tau_{1}}{\tau_{2}}}p'$, hence
        $A_{\tau_{2}}(p' \app{}{}q)$.
      \item $p \app{}{}q \to p'$ where $p \xto{q} p'$. Since
        $\sqdown_{\arty{\tau_{1}}{\tau_{2}}} p$ and
        $\sqdown_{\tau_1}q$, we have $\sqdown_{\tau_2}p'$, so
		$A_{\tau_2}(p')$.
      \end{enumerate}
    \item $t$ does not have labelled transitions, hence this case is void.
    \end{enumerate}
  
    Next, we show that
    $C = \sqdown \lor \fimg{(\iota\comp\inr)}{\ol{\Delta}(\sqdown)}$ is
    a $\sqdown$-relative invariant. We consider two representative
    cases; the remaining cases are handled similarly.
      \begin{itemize}
  \item Case $I_{\tau} \c \arty{\tau}{\tau}$. Since $I$ terminates immediately, property (1) holds by definition of $\Downarrow$ and (2) holds vacuously. For (3), if $I\xto{s}t'$ and $\sqdown_\tau s$, then $t'=s\in \sqdown_\tau\seq C_\tau$.
  \item Case
    $S''_{\tau_{1},\tau_{2},\tau_{3}}(t,s) \c
    \arty{\tau_{1}}{\tau_{3}}$ with
    $\sqdown_{\arty{\tau_{1}}{\arty{\tau_{2}}{\tau_{3}}}} t$ and
    $\sqdown_{\arty{\tau_{1}}{\tau_{2}}} s$. Again,~(1) holds because
    $S''(t,s)$ terminates immediately, and~(2) holds
    vacuously. For~(3), suppose that $\sqdown_{\tau_1}r$; we have to
    show $(t\app{}{}r)\app{}{}(s\app{}{}r)\in C_{\tau_{3}}$. This
    follows from the inequality
    $\ol{\Xi}(\sqdown) \leq \fimg{\iota}{\sqdown}$ shown above, because
    $\sqdown_{\arty{\tau_{2}}{\tau_{3}}}(t \app{}{}r)$,
    $\sqdown_{\tau_2}(s \app{}{}r)$ by definition of $\sqdown$.
  \end{itemize}
\end{expl}
Note that the definition of $\sqdown$ uses both induction (over the structure of types) and coinduction (by taking at every type the greatest predicate satisfying some property).
\begin{expl}\label{ex:logical-plotkin}
  \label{ex:log1}
  We give an alternative logical predicate defined purely inductively. It resembles
  Plotkin's original concept of logical relation~\cite{plotkin1973lambda}. We define $\DDar\monoto \Tr$ by
\begin{equation}
  \label{eq:log1}
  \begin{aligned}
     \DDar_{\utype}\, (t) \iff &\; {\Downarrow}_\utype\, (t), \\
     \DDar_{\arty{\tau_{1}}{\tau_{2}}}\, (t)
      \iff &\; {\Downarrow}_{\arty{\tau_{1}}{\tau_{2}}}\,t
      \land (\forall s\c\tau_1.\, t \xTo{s} t' \land \DDar_{\tau_{1}}\, (s) \implies
      \DDar_{\tau_{2}}\, (t')).
  \end{aligned}
\end{equation}
It is evidently logical for the restriction $\tsup{\gamma} \c \Tr \to \Pow D(\Tr,\Tr)$ 
of the weak operational model to labelled transitions, given by $\tsup{\gamma}(t) := \{\gamma(t')\}$ 
if $t\To t'$ and $\gamma(t')\in D(\Tr,\Tr)$, and $\tsup{\gamma}(t) :=\emptyset$ otherwise. A proof of strong normalization using $\DDar$ is given in \Cref{app:omitted-proofs}
\end{expl}

\begin{expl}\label{ex:logical-tait}
 A more popular
  (cf.~\cite{DBLP:journals/corr/abs-1907-11133,STATMAN198585}) and subtly
  different variant of $\DDar$ for proving strong normalization goes back to
Tait~\cite{DBLP:journals/jsyml/Tait67}. We define $\mathrm{SN}\monoto \Tr$ by
\begin{equation}
  \label{ex:sn}
  \begin{aligned}
    \mathrm{SN}_\utype\, (t) \iff
    &\;
      {\Downarrow}_\utype\, (t) \\
    \mathrm{SN}_{\arty{\tau_{1}}{\tau_{2}}}\,(t)\iff
    &\;
      {\Downarrow}_{\arty{\tau_{1}}{\tau_{2}}}\,(t)
      \land (\forall s \c \tau_{1}.\, \mathrm{SN}_{\tau_{1}}(s) \implies
      \mathrm{SN}_{\tau_{2}}(t \app{}{}s))
  \end{aligned}
\end{equation}
  Unlike $\DDar$, it is not immediate that $\mathrm{SN}$
  is logical for $\tsup{\gamma}$ (see \Cref{app:omitted-proofs}). For a proof of strong
  normalization based on $\mathrm{SN}$ in the context of the $\lambda$-calculus,
  see \cite[Sec.~2]{DBLP:journals/corr/abs-1907-11133}.
\end{expl}

While all three logical predicates $\sqdown$, $\DDar$, $\mathrm{SN}$ are eligible for proving strong normalization, with proofs of similar length and complexity, the predicate $\sqdown$  arguably has the most generic flavour, as it  depends neither on a system-specific notion of weak transition (which appears in the definition of $\DDar$) nor on the syntax of the language (such as the application operator appearing in the definition of $\mathrm{SN}$). Thus, our abstract categorical approach to logical predicates will be based on a generalization of $\sqdown$.

\subsection{Constructing Logical Predicates}\label{sec:constructing-logical-predicates}
Our abstract coalgebraic notion of logical predicate (\Cref{def:logpred}) is
parametric in the bifunctor $B$ and its lifting $\ol{B}$ and decoupled from any
specific syntax. Next, we develop a systematic construction that promotes a predicate~$P$
to a logical predicate, specifically to a locally maximal refinement of $P$, generalizing $\sqdown$ in \Cref{ex:strongnorm1}. 
The construction proceeds in two stages. First, we fix the contravariant argument of the lifted bifunctor $\ol{B}$ and construct a greatest coalgebraic invariant w.r.t.\ the resulting endofunctor~\cite[\textsection
6.3]{DBLP:books/cu/J2016}:
\begin{defn}[Relative henceforth]
  \label{def:henceforth}
  Let $c \c Y \to B(X,Y)$ and let $S \pred{} X$ be a predicate. The
  \emph{($S$-)relative henceforth modality} sends $P \pred{} Y$ to $\invp^{\ol{B},c} (S, P)
  \pred{} Y$, which is the supremum in $\Pred[Y]{\C}$ of all $S$-relative invariants contained in $P$:
  \begin{equation}\label{eq:hence}
    \invp^{\ol{B},c} (S, P) = \bigvee\{ Q \leq P \mid Q
    \text{ is an $S$-relative $\ol{B}$-invariant for $c$} \}.
  \end{equation}
We will omit the superscripts $\ol B,c$ when they are irrelevant or clear from the context.
\end{defn}
\begin{prop}
  \label{prop:invarprop}
  The predicate $\invp (S, P)$ is the greatest $S$-relative $\ol{B}$-invariant contained in $P$. Moreover, the map $(S,P)\mapsto \invp (S, P)$ is antitone in $S$ and monotone in $P$.
\end{prop}
\begin{proof}
The first statement follows from the Knaster-Tarski theorem since $\invp (S, P)$ is the greatest fixed point
$
    \invp(S, P) = \nu G.\ P \land \iimg{c}{\overline{B}(S,G)}
$
in the complete lattice $\Pred[Y]{\C}$.
The second statement holds due to the mixed variance of the predicate lifting $\ol{B}$.
\end{proof}
The relative henceforth modality only yields relative invariants. To obtain a
logical predicate, i.e.\ an invariant relative to itself, we move to the second stage
of our construction, which is based on ultrametric semantics, see~e.g.~\cite{235f3d3a0dee4f339a8f59beff18cae3}. Let us briefly recall some terminology.
A metric space $(X,\, d\colon X\times X\to \mathbb{R})$ is \emph{$1$-bounded} if $d(x,y)\leq 1$ for all $x,y$, an \emph{ultrametric space} if $d(x,y)\leq\max\{d(x,z),d(z,y)\}$ for all $x,y,z$, and \emph{complete} if every Cauchy sequence converges. A map $f\colon (X,d)\to (X',d')$ between metric spaces
is \emph{nonexpansive} if $d'(f(x),f(y))\leq d(x,y)$ for all $x,y$, and \emph{contractive}
if there exists $c\in [0,1)$, called a \emph{contraction factor}, such that $d'(f(x),f(y))\leq c\cdot d(x,y)$
for all $x,y$. A family of maps $(f_i\c X\to X')_{i\in I}$ is \emph{uniformly contractive}
if there exists $c\in [0,1)$ such that each $f_i$ is contractive with factor $c$. By Banach's fixed point theorem, every contractive endomap $f\colon X\to X$ on a non-empty complete metric space has a unique fixed point.

\begin{defn}\label{ass:contr}
The category $\C$ is \emph{predicate-contractive} if
\begin{enumerate}
  \item every $\Pred[X]{\C}$ carries the structure of a complete $1$-bounded ultrametric space;
  \item for every $f\c X\to Y$ in $\C$, the map $\iimg{f}{\argument}\c\Pred[Y]{\C}\to\Pred[X]{\C}$
  is non-expansive;
  \item for any two co-well-ordered families $(P^i\pred{} X)_{i\in I}$ and $(Q^i\pred{} X)_{i\in I}$ of predicates,
  \begin{displaymath}\textstyle
    d\bigl(\bigand_{i\in I} P^i,\bigand_{i\in I} Q^i\bigr)\leq \sup_{i\in I} d(P^i,Q^i).
  \end{displaymath} 
  Here $(P^i\pred{} X)_{i\in I}$ is \emph{co-well-ordered} if each nonempty subfamily has a greatest element.
\end{enumerate}
\end{defn}
\begin{expl}
  The category $\C=\Set^\Ty$ is predicate-contractive when equipped with the
  ultrametric on $\Pred[X]{\C}$ given by $d(P,Q)=2^{-n}$ for
  $P,Q\monoto X$, where $n=\inf\{\sharp\tau\mid P_\tau\neq Q_\tau\}$
  and $\sharp\tau$ is the size of~$\tau$, defined by $\sharp\utype=1$
  and $\sharp(\arty{\tau_1}{\tau_2})=\sharp\tau_1+\sharp\tau_2$. By
  convention, $\inf\,\emptyset=\infty$ and $2^{-\infty}=0$. To see
  predicate-contractivity, first note that a function
  $\mathcal{F}\c\Pred[Y]{\C}\to\Pred[X]{\C}$ is non-expansive iff
\begin{displaymath}
\inf\{\sharp\tau \mid (\mathcal{F} P)_\tau \neq (\mathcal{F} Q)_\tau\}
\geq \inf\{\sharp\tau \mid P_\tau \neq Q_\tau\} \qquad\text{for all $P, Q\monoto Y$},
\end{displaymath}
and contractive (necessarily with factor at most $1/2$) iff that inequality holds strictly.

This immediately implies clause~(2) of~\Cref{ass:contr}: inverse images in $\Set^\Ty$ are computed 
pointwise, and $\iimg{f_\tau}{P_\tau} \neq\iimg{f_\tau}{Q_\tau}$
implies $P_\tau\neq Q_\tau$ for $f\colon X\to Y$ and $P,Q\monoto Y$. Similarly, since intersections are computed pointwise, clause~(3) amounts to
\begin{align*}
\inf\Bigl\{\sharp\tau\mid \bigcap_{i\in I}P^i_\tau\neq \bigcap_{i\in I}Q^i_\tau\Bigr\}
\geq&\;\inf\{\sharp\tau\mid \exists i\in I: P^i_\tau\neq Q_\tau^i\},
\end{align*}
which is clearly true, for if $\bigcap_{i\in I}P^i_\tau\neq \bigcap_{i\in I}Q^i_\tau$ then 
$P^i_\tau\neq Q^i_\tau$ for some $i\in I$.
\end{expl}
\begin{defn}[Contractive lifting]
Suppose that $\C$ is predicate-contractive. The predicate lifting $\ol{B}\colon \Pred{\C}^\op \times \Pred{\C}\to \Pred{\C}$ is
\emph{contractive} if for every $S\pred{} X$ the map $\ol{B}(S,\argument)$
  is non-expansive, and the family $(\ol{B}(\argument,P))_{P\pred{} X}$ is uniformly contractive.
\end{defn}

\begin{prop}\label{pro:square}
Let $\ol{B}$ be contractive and $c\colon X\to B(X,X)$. For every $S\pred{} X$, the map $\invp^{\ol{B},c}(S,\argument)$ 
  is non-expansive, and the family $(\invp^{\ol{B},c}(\argument,P))_{P\pred{} X}$
  is uniformly contractive. 
\end{prop}
\noindent Contractive liftings allow us to augment every predicate $P$ to a logical
predicate:
\begin{defn}[Henceforth]
  \label{def:square}
  Let $\ol{B}$ be contractive and $c\colon X\to B(X,X)$. For each
  predicate $P\monoto X$ we define $\logp^{\ol{B},c} P \pred{} X$
  (where we usually omit the superscripts) to be the unique fixed point of
  the contractive endomap
  \begin{equation}\label{eq:fix}
      S\mapsto \invp^{\ol{B},c} (S, P) \quad \text{on}\quad \Pred[X]{\C}.
  \end{equation}
\end{defn}

\begin{theorem}\label{thm:square-refinement}
The predicate $\logp{P}$ is the unique locally maximal logical refinement of $P$.
\end{theorem}

\begin{proof}
By \eqref{eq:fix}, $\logp{P}$ is the unique predicate satisfying
$\logp{P}=\invp(\logp{P},P)$. By \eqref{eq:hence}, this equality says
that $\logp{P}$ is the greatest $\logp{P}$-relative invariant contained in $P$,
as needed.
\end{proof}
\begin{expl}%
  \label{ex:basicsquare}
  Let~$B$ be the behaviour bifunctor on $\Set^\Ty$ given by \eqref{eq:beh}. Its canonical lifting~$\ol{B}$ (\Cref{ex:lifting}) is contractive because $\ol{B}_{\arty{\tau_1}{\tau_2}}(P,Q)$ depends only on $P_{\tau_1}$, $Q_{\tau_2}$, $Q_{\arty{\tau_1}{\tau_2}}$; in other words, $\ol{B}$ decreases the size of types in the contravariant argument and does not increase it in the covariant argument. Given a coalgebra $c \c X \to B(X,X)$ and %
  ${P \pred{} X}$, the fixed point $\logp^{\ol{B},c} P$ is given by the
  $\Ty$-indexed family of  greatest fixed points
  \begin{equation*}
    \begin{aligned}
      \logp P_{\utype} =
      &\; \nu G.\, P_{\utype} \land \iimg{c_{\utype}}{G + \term}, \\
      \logp P_{\arty{\tau_{1}}{\tau_{2}}} =
      &\; \nu G.\,
        P_{\arty{\tau_{1}}{\tau_{2}}} \land \iimg{c_{\arty{\tau_{1}}{\tau_{2}}}}{G +
        \{f\c {\Tr_{\tau_{1}}}\to {\Tr_{\tau_{2}}}\mid\forall s \in \logp{P}_{\tau_{1}}.\, f (s)\in \logp{P}_{\tau_{2}} \}}.
    \end{aligned}
  \end{equation*}
This follows from \Cref{thm:square-refinement} since the above predicate is clearly a locally maximal refinement of $P$. By instantiating $c$ to the operational model $\gamma\colon \mS\to B(\mS,\mS)$ of $\stsc$ and taking $P={\Downarrow}$, we recover the definition of $\sqdown$ in \Cref{ex:strongnorm1}. 
\end{expl}
\begin{expl}%
  The logical predicate $\DDar \pred{} \Tr$ of \Cref{ex:logical-plotkin} is precisely
  $\logp {\Downarrow}$ for $\Pow D$ w.r.t.\ its canonical lifting
  and the coalgebra $\tsup{\gamma} \c \Tr \to
\Pow D(\Tr,\Tr)$.
More generally, for a coalgebra $c \c X \to \Pow D(X,X)$, the predicate $\logp P$ is inductively defined as
follows:
  \begin{equation*}
    \begin{aligned}
       \logp P_{\utype} =&\; P_{\utype}, \\
       \logp P_{\arty{\tau_{1}}{\tau_{2}}} =&\;
        P_{\arty{\tau_{1}}{\tau_{2}}} \land \iimg{c_{\arty{\tau_{1}}{\tau_{2}}}}{\{F\subseteq X_{\tau_{2}}^{X_{\tau_{1}}}\mid\forall f\in F.\, s\in \square P_{\tau_{1}}\implies f (s)\in \square P_{\tau_{2}} \}}.
    \end{aligned}
  \end{equation*}
\end{expl}

\begin{rem}
The construction of logical predicates for typed languages is enabled by the ``type-decreasing'' nature of the associated behaviour bifunctors. In untyped settings, e.g.\ for $B(X,Y) = Y+Y^{X}$ on $\Set$ modelling untyped combinatory logic~\cite{gmstu23}, the canonical lifting $\ol{B}$ is not contractive, hence the fixed point $\logp{P}$ in general fails to exist.
\end{rem}

\begin{rem}\label{rem:fibrations}
The forgetful functor $\under{-}\colon \Pred{\C}\to \C$ forms a complete lattice
fibration~\cite{jacobs99}, equivalently a topological functor~\cite{AdamekEA90},
and all notions and results of the present subsection extend to that level of
generality. We leave the details for future work, as our reasoning techniques
found in the upcoming sections are tailored to logical predicates.
\end{rem}

We are now in a position to state precisely what a {proof via logical
predicates} is in our framework. Given the operational model $\gamma \c \mS
\to B(\mS,\mS)$ of a higher-order language, a predicate lifting
$\ol{B}$, and a target predicate $P \pred{}
\mS$, a \emph{proof of $P$ via logical predicates} is a proof that $\logp P$ forms a subalgebra of the initial algebra $\mS$, which means
\begin{equation}
  \label{eq:logrelproof}
  \overline{\Sigma} (\logp P) \leq \iimg{\iota}{\logp P}, \text{\quad equivalently\quad}
    \fimg{\iota}{\overline{\Sigma} (\logp P)} \leq \logp P.
\end{equation}
Then $\logp{P}=\mS$ by structural induction, whence $P=\mS$ because $\logp{P}\leq P$.

Up to this point, we have streamlined and formalized coalgebraic logical predicates as a
certain abstract construction on predicates (\Cref{def:square}) and
presented proofs by coalgebraic logical predicates as standard structural
induction on said
construction. This presentation is indeed that of an abstract method: the
various parts of the problem setting, namely the syntax, the behaviour and its
predicate lifting, as well as the operational semantics, are all parameters. In the
next section, we exploit the parametric and generic nature of this method in two
main ways. First, we present up-to techniques that simplify the proof goal \eqref{eq:logrelproof}
as much as possible. Second, we look to instantiate our method to problems on
\emph{classes of higher-order languages}, as opposed to reasoning about operational models
of individual languages such as \stsc or the $\lambda$-calculus.

\section{Logical Predicates and Higher-Order Abstract GSOS}
\label{sec:hogsos}

As indicated before, substantial parts of the proof of strong normalization in \Cref{ex:strongnorm1} look generic. Specifically, the properties (2) and (3) established for $Q=A$ and $Q=C$  are independent of the choice of 
predicate $P = {\Downarrow}$ in $\logp P$. Moreover, these steps are either
obvious or follow immediately from the operational rules of $\stsc$: the predicates~$A$ and~$C$
being invariants can be attributed to the fact that except for terms of the form $S''(-,-)$, all terms evolve either to a variable or to some flat term such as $p' \app{}{}q$. 
The core of the proof, which is tailored to the choice of~$P$, 
lies in proving property~(1).

As it turns out, for a class of higher-order GSOS laws that we call
\emph{relatively flat higher-order GSOS laws}, conditions (2) and (3) are automatic. This insight leads us to a powerful up-to technique that simplifies proofs via logical predicates.

\subsection{Relatively Flat Higher-Order GSOS Laws}\label{sec:relatively-flat-laws}
The following definition abstracts the restricted nature of the rules of $\stsc$ to the level of higher-order GSOS laws. For simplicity, we confine ourselves to $0$-pointed laws, however all the  results of this subsection easily extend to the $V$-pointed case.
\begin{defn}
  \label{def:relativelyflat}
  Let $\Sigma \c \C \to \C$ be a syntax functor of the form $\Sigma = \coprod_{j \in J}
  \Sigma_{j}$, where $(J,\prec)$ is a non-empty well-founded strict partial order, and put
  $\Sigma_{\prec k} = \coprod_{j \prec k}\Sigma_{j}$. A \emph{relatively flat}
  \emph{($0$-pointed) higher-order GSOS law} of~$\Sigma$ over $B$ is a
  $J$-indexed family of morphisms
  \begin{align}\label{eq:rflat}
    \rho^{j}_{X,Y} \c \Sigma_{j} (X \times B(X,Y))\to
    B(X, \Sigma^{\star}_{\prec j} (X+Y) + \Sigma_{j}\Sigma^{\star}_{\prec j} (X+Y))
  \end{align}
  dinatural in $X\in \C$ and natural in $Y\in \C$.
\end{defn}
We put $e_{j,X} = [\inj^{\klstar}_{\prec j},
\iota\comp \inj_j \comp \Sigma_{j}(\inj^{\klstar}_{\prec j})]\c 
\Sigma^{\star}_{\prec j} X + \Sigma_{j}\Sigma^{\star}_{\prec j} X\to \Sigma^{\star} X$ where $\inj_{\prec j}\c \Sigma_{\prec j} \to \Sigma$ and $\inj_{j}\c \Sigma_{j} \to \Sigma$ are the coproduct injections, with free extensions  $\inj^{\klstar}_{\prec j} \c \Sigma^{\star}_{\prec j}
\to \Sigma^{\star}$ and $\inj_{j}^\klstar\c \Sigma_{j}^\star \to \Sigma^\star$. Every relatively flat higher-order GSOS law \eqref{eq:rflat} determines an ordinary higher-order
GSOS law of $\Sigma$ over $B$ with components
\begin{align*}
  \rho_{X,Y} = \coprod_{j \in J} \Sigma_{j} (X \times B(X,Y)) \xrightarrow{\coprod_{j \in
      J}\rho^{j}_{X,Y}}\,&
  \coprod_{j \in J} B(X,
  \Sigma^{\star}_{\prec j} (X+Y) + \Sigma_{j}\Sigma^{\star}_{\prec j} (X+Y))\\*
  \xrightarrow{[B(X, e_{j,X+Y})]_{j\in J}}\,&
  B(X,\Sigma^{\star}(X + Y)).
\end{align*}
When we interpret a higher-order GSOS law as a set of operational rules, relative flatness means that the operations of the language can be ranked in a way that every term $\f(-,\cdots,-)$ with $\f$ of rank $j$ evolves into a term that uses only operations of strictly lower rank, except possibly its head symbol which may have the same rank $j$.
\begin{expl}
  \label{ex:tclisessflat}
  \stsc{} is relatively flat: put $J = \{0\prec 1\}$, let $\Sigma_0$ contain application, and let $\Sigma_1$ contain all other operation symbols. This is immediate from the rules in \Cref{fig:skirules}.
\end{expl}

\begin{defn}
Suppose that each $\Sigma_j$ preserves strong epimorphisms. A \emph{predicate lifting} of \eqref{eq:rflat} is a relatively flat $0$-pointed higher-order GSOS law $(\ol{\rho}^j)_{j\in J}$ of $\ol{\Sigma}=\coprod_j \ol{\Sigma_j}$ over~$\ol{B}$ where for every $P\monoto X$ and $Q\monoto Y$ the $\Pred{\C}$-morphism $\ol{\rho}^{j}_{P,Q}$ is carried by $\rho^j_{X,Y}$.
\end{defn}

\begin{rem}
\begin{enumerate}
\item The condition on $\Sigma_j$ ensures $\ol{\Sigma_j}^\star=\ol{\Sigma_j^\star}$ (\Cref{prop:free-monad-lift}), so that the first component of $\ol{\rho}^{j}_{P,Q}$ has type $\Sigma_{j} (X \times B(X,Y))\to
    B(X, \Sigma^{\star}_{\prec j} (X+Y) + \Sigma_{j}\Sigma^{\star}_{\prec j} (X+Y))$.
\item Liftings are unique if they exist: since $\ol{\rho}^{j}_{P,Q}$ is a $\Pred{\C}$-morphism, it is determined by its first component $\rho^j_{X,Y}$. Moreover, the (di)naturality of $\ol{\rho}^{j}$ follows from that of $\rho^j$.
\item For the canonical lifting $\ol{B}$, a lifting $(\ol{\rho}^{j})_{j\in J}$ of $(\rho^j)_{j\in J}$ always exists (\Cref{appsec:liftings}).
\end{enumerate}
\end{rem}
The following theorem establishes a sound up-to technique for logical predicates. It states that for operational models of relatively flat laws, the proof goal \eqref{eq:logrelproof} can be established by checking a substantially relaxed property.

\begin{theorem}[Induction up to $\square$]
  \label{th:main2}
  Let $\gamma \c \mS \to B(\mS,\mS)$ be the operational model of a relatively flat $0$-pointed higher-order GSOS law that admits a predicate lifting. Then for every predicate $P \pred{} \mS$ and every locally maximal logical refinement $\logp^{\gamma,\ol{B}}P$,
  \[
    \ol{\Sigma} (\logp^{\gamma,\ol{B}} P) \leq \iimg{\iota}{P} \qquad \text{implies} \quad
    \ol{\Sigma} (\logp^{\gamma,\ol{B}} P) \leq \iimg{\iota}{\logp^{\gamma,\ol{B}} P} \quad\text{(hence ${P}=\mS$)}.
  \]
\end{theorem}
\noindent We stress that the theorem applies to any refinement $\logp^{\gamma,\ol{B}}P$ and does not assume a specific construction (e.g.\ that of \Cref{sec:constructing-logical-predicates}). The up-to technique facilitates proofs via logical predicates quite dramatically. For illustration, we revisit strong normalization:

\begin{expl}
  \label{ex:tclnorm}
 We give an alternative proof of strong normalization of \stsc (cf.\ \Cref{ex:strongnorm1}) via induction up to $\square$. Hence it suffices to prove
  \[
    {\ol{\Sigma} (\logp {\Downarrow})} \leq \iimg{\iota}{\Downarrow},
  \]
  which states that a term is terminating if all of its
  subterms are in the logical predicate $\logp {\Downarrow}$. This is clear for terms that are not applications, since they immediately terminate (cf.\ \Cref{fig:skirules}). Now consider an application $p \app{}{} q$ such that $
  \logp_{\arty{\tau_{1}}{\tau_{2}}} {\Downarrow}(p)$ and $\logp_{\tau_{1}}
  {\Downarrow}(q)$. Since $\sqdown$ is a logical predicate contained in $\Downarrow$, this entails that $p\xTo{q} p'$ for a (unique)
      term~$p'$, and that $\sqdown_{\tau_2}p'$, hence
      $\Downarrow_{\tau_2} p'$. Since $p \app{}{}q\To p'$, it follows that 
      $\Downarrow_{\tau_2}p\app{}{} q$.

Analogous reasoning shows that \stsc is strongly normalizing under the call-by-value
and the maximally nondeterministic evaluation strategy (\Cref{rem:cbv-nondet}). In
the latter case, strong normalization means that every term must eventually
terminate, independently of the order of evaluation.
\end{expl}
The reader should compare the above compact argument to the laborious original proof given in \Cref{ex:strongnorm1}. Our up-to technique can be seen to precisely isolate the non-trivial core of the proof, while providing its generic parts for free. For a further application -- type safety of the simply typed $\lambda$-calculus -- see \Cref{sec:lambda-laws}.

\subsection{\texorpdfstring{$\lambda$-Laws}
  {Lambda-Laws}}\label{sec:lambda-laws}
\label{sec:lambdalawsmain}

We proceed to explain how our theory of logical predicates applies to languages with variables and binders. We highlight the core ideas and technical challenges in the case of the $\lambda$-calculus, and briefly sketch their categorical generalization; a full exposition can be found in~\Cref{app:lambdashenanigans}. Let \stlc be the simply typed call-by-name $\lambda$-calculus with the set $\Ty$ of types given by \eqref{eq:type-grammar} and operational rules
\begin{equation}\label{eq:lambda-rules}
  \inference{t \xto{} t'}{t \app{}{}s \to t' \app{}{}s} \qquad
  \inference{}{(\lambda x \c \tau_{1}.\,t) \app{}{}s \to t[s/x]}
\end{equation}
where $s,t,t'$ range over $\lambda$-terms of appropriate type, and $[-/-]$ denotes capture-avoiding substitution. To model \stlc in higher-order abstract GSOS, we follow ideas by Fiore~\cite{DBLP:journals/mscs/Fiore22}. Our base category $\C$ is the presheaf category $(\Set^{\fset/{\Tyl}})^{\Tyl}$ where $\fset$ denotes the category of finite cardinals and functions, and the set $\Tyl$ is regarded as a discrete category. An object $\Gamma\colon n\to \Ty$ of $\fset/\Tyl$ is a \emph{typed context}, associating to each variable $x\in n$ a type; we put $|{\Gamma}|:=n$ . A presheaf $X\in (\Set^{\fset/{\Tyl}})^{\Tyl}$ associates to each context $\Gamma$ and each type $\tau$ a set $X_{\tau}(\Gamma)$ whose elements we think of as terms of type $\tau$ in context $\Gamma$.

The syntax of $\stlc$ is captured by the functor $\Sigma\colon (\Set^{\fset/{\Tyl}})^{\Tyl}\to (\Set^{\fset/{\Tyl}})^{\Tyl}$ where
\begin{equation}
  \label{eq:lambdasigma}
  \begin{aligned}
    \Sigma_{\utype}X =
    &\; V_{\utype} + K_{1}
      + \coprod_{\tau \in \Tyl}X_{\arty{\tau}{\utype}}
      \times X_{\tau},\\
    \Sigma_{\arty{\tau_{1}}{\tau_{2}}} X=
    &\; V_{\arty{\tau_{1}}{\tau_{2}}} + \delta^{\tau_{1}}_{\tau_{2}} X
      + \coprod_{\tau \in \Tyl}X_{\arty{\tau}{\arty{\tau_{1}}{\tau_{2}}}}
                                           \times X_{\tau}.
  \end{aligned}
\end{equation} 
Here $K_1\in \Set^{\fset/{\Tyl}}$ is the constant presheaf on $1$, $V$ is given by $V_{\tau}(\Gamma)=\{x\in |\Gamma|\mid \Gamma(x)=\tau\}$, and $\delta$ by
 $(\delta^{\tau_1}_{\tau_2}X)(\Gamma) = X_{\tau_2}(\Gamma + \check\tau_1)$ with $(-) + \check\tau_1$ denoting context extension by a variable of type $\tau_1$. Informally, $K_1$, $V$ and $\delta$ represent the constant $\mathsf{e}\colon \utype$, variables, and $\lambda$-abstraction, respectively. The initial algebra for $\Sigma$ is the presheaf $\Lambda$ of $\lambda$-terms, i.e.\ $\Lambda_\tau(\Gamma)$ is the set of $\lambda$-terms (modulo $\alpha$-equivalence) of type $\tau$ in context $\Gamma$~\cite{DBLP:journals/mscs/Fiore22}.

The behaviour bifunctor $B^{\lambda} \colon ((\Set^{\fset/{\Tyl}})^{\Tyl})^\op\times (\Set^{\fset/{\Tyl}})^{\Tyl}\to (\Set^{\fset/{\Tyl}})^{\Tyl}$ for $\stlc$ has two separate components: it is given by a product
\begin{flalign}\label{eq:lambda-beh} 
&&B^{\lambda}(X,Y)=\;&\llangle X,Y\rrangle \times B(X,Y)&&\\[1ex] 
\text{where}&& 
 \llangle X,Y \rrangle_{\tau}(\Gamma) =\;& \Set^{\fset/{\Tyl}}\Bigl(\prod_{x \in
    |\Gamma|}X_{\Gamma(x)}, Y_{\tau}\Bigr),\notag\\
 && B(X,Y) =\;& (K_{1} + Y + D(X,Y)),\notag\\
 && D_{\utype}(X,Y) =\;& K_{1} \qquad \text{and}
  \qquad D_{\arty{\tau_{1}}{\tau_{2}}}(X,Y) = Y_{\tau_{2}}^{X_{\tau_{1}}},\notag
\end{flalign}
and $Y_{\tau_{2}}^{X_{\tau_{1}}}$ is an exponential object in $\Set^{\fset/\Ty}$.
The bifunctor $\llangle -,-\rrangle$ models an abstract substitution structure;
for instance, every $\lambda$-term $t\in \Lambda_\tau(\Gamma)$ induces a natural
transformation $\prod_{x\in|\Gamma|} \Lambda_{\Gamma(x)}\to \Lambda_\tau$ in
$\llangle \Lambda,\Lambda\rrangle_{\tau}(\Gamma)$ mapping a tuple
$(t_1,\ldots,t_{|\Gamma|})$  to the term obtained by simultaneous substitution
of the terms $t_i$ for the variables of $t$. The summands of the bifunctor
$B$ abstract from the possible operational behaviour of $\lambda$-terms:
a term may explicitly terminate, reduce, get stuck (e.g.\ if it is a variable),
or act as a function.

The operational rules \eqref{eq:lambda-rules} of \stlc can be encoded into a
$V$-pointed higher-order GSOS law of $\Sigma$ over $B^{\lambda}$, similar to the untyped
$\lambda$-calculus treated in earlier work~\cite{gmstu23}. The operational model
$\langle \phi,\gamma\rangle\colon \Lambda\to \llangle \Lambda,\Lambda
\rrangle\times B(\Lambda,\Lambda)$ is the coalgebra whose components
$\phi,\gamma$ describe the substitution structure and the operational behaviour
of $\lambda$-terms.

At this point, a key technical issue can be observed: the canonical predicate
lifting $\ol{\llangle -,-\rrangle}$ is not contractive. Indeed, given $P\monoto
X$, $Q\monoto Y$, the predicate $\ol{\llangle P,Q\rrangle}_\tau$ consists of all
natural transformations  $\prod_{x \in
    |\Gamma|}X_{\Gamma(x)}\to Y_{\tau}$ that restrict to $\prod_{x \in
    |\Gamma|}P_{\Gamma(x)}\to Q_{\tau}$, and this expression depends on
  $P_{\Gamma(x)}$ where the type $\Gamma(x)$ may be of higher complexity than~$\tau$.
  In particular, we conclude that $\ol{B^{\lambda}}$ is not contractive. In
  contrast, the canonical lifting $\ol{B}$ is contractive and hence
  $\square^{\gamma,\ol{B}} P$ exists for every $P\monoto \Lambda$ (\Cref{def:square}).
 However, it is well-known that logical
  predicates do not do the trick for inductive proofs in
the $\lambda$-calculus, see e.g.~\cite[p. 9]{DBLP:journals/corr/abs-1907-11133}
and
\cite[p. 150]{DBLP:books/daglib/0005958}; rather, one needs to prove the
\emph{open extension} of the logical predicate, which is the larger predicate \[\blacksquare^{\gamma,\ol{B}} P = \iimg{\phi}{\ol{\llangle \square^{\gamma,\ol{B}}{P},\square^{\gamma,\ol{B}}{P}\rrangle}}.\]
The standard proof method is then to show $\blacksquare^{\gamma,\ol{B}} P=\Lambda$ directly by structural induction. However, this can be greatly simplified by the following up-to-principle, which works with the original predicate $\square^{\gamma,\ol{B}}{P}$ and forms a counterpart of  \autoref{th:main2} for the $\lambda$-calculus:
\begin{theorem}[Induction up to $\blacksquare$]\label{thm:ind-up-to-blackquare}
Let $P\monoto \Lambda$ be a predicate. Then
\[ \ol{\Sigma}(\square^{\gamma,\ol{B}})\leq \iimg{\ini}{P} \qquad\text{implies}\qquad  \ol{\Sigma}(\blacksquare^{\gamma,\ol{B}}P)\leq \iimg{\ini}{\blacksquare^{\gamma,\ol{B}}P} \quad\text{(hence $P=\Lambda$)}. \]
\end{theorem}
\begin{rem}\label{rem:up-to-blacksquare-concrete}
Concretely, 
the theorem states that to prove $P=\Lambda$, it suffices to prove that
(1) variables satisfy $P$, (2) the unit expression $\mathsf{e} \c \utype$ satisfies $P$, (3) for all application terms $p \app{}{} q$ such that $\logp_{\arty{\tau_{1}}{\tau_{2}}}
    P(\Gamma)(p)$ and $\logp_{\tau_{1}} P(\Gamma)(q)$, we have $P_{\tau_{2}}(\Gamma)(p
    \app{}{} q)$, and (4) for all $\lambda$-abstractions $\lambda x \c \tau_{1}.\,t$ such that $t \in \logp_{\tau_{2}}P(\Gamma,x)$, we have
    $P_{\arty{\tau_{1}}{\tau_{2}}}(\Gamma)(\lambda x \c \tau_{1}.\,t)$.
\end{rem}

\begin{expl}
  We prove type safety for \stlc{} via induction up to $\blacksquare$. Thus consider
 the predicate $\mathrm{Safe}\monoto \Lambda$ that is constantly true on open terms and given by 
  \[
    t\in \mathrm{Safe}_{\tau}(\varnothing) \iff \big(\forall {e}.\,t \To e \implies
    (\text{$e$ is not an application}) \lor \exists r.\, e \to r\big),
  \]
   on closed terms. We only need to check the conditions (1)--(4) of \Cref{rem:up-to-blacksquare-concrete}. Conditions (1), (2), (4) are clear since variables are open terms and the term $\mathsf{e}\colon \utype$ and $\lambda$-abstractions do not reduce.
The only interesting clause is (3) for the empty context. Thus let $p \app{}{} q$ be a closed application term with $p \in \logp
      {\mathrm{Safe}}_{\arty{\tau_{1}}{\tau_{2}}}(\varnothing)$ and $q \in \logp
      {\mathrm{Safe}}_{\tau_{1}}(\varnothing)$; we need to show $p \app{}{} q \in
      {\mathrm{Safe}}_{\tau_{2}}(\varnothing)$. We proceed by case distinction on $p \app{}{} q
  \To e$:
  \begin{enumerate}[label=(\alph*)]
  \item $p \To p'$ and $e = p' \app{}{} q$. Then $p'\in \logp
      {\mathrm{Safe}}_{\arty{\tau_{1}}{\tau_{2}}}(\varnothing)$ by invariance, in particular $p'$ is safe, so $p'$ is either not an application or reduces. In the former case, $p'$ is necessarily a $\lambda$-abstraction since it is closed and not of type $\utype$. Thus, in both cases, $e$ reduces.
  \item $p \To \lambda x.p'$ and $p'[q/x] \To e$. Since $\logp
    \mathrm{Safe}$ is a logical predicate, from $p \in \logp
      {\mathrm{Safe}}_{\arty{\tau_{1}}{\tau_{2}}}(\varnothing)$ and $q \in \logp_{\tau_{1}}\mathrm{Safe}(\varnothing)$ we can deduce 
    $p'[q/x] \in \logp_{\tau_{2}}\mathrm{Safe}(\varnothing)$, whence $e \in \logp_{\tau_{2}}\mathrm{Safe}(\varnothing)$. In particular, $e$ is safe, which implies that $e$ is either not an application or reduces.
  \end{enumerate}
\end{expl}
As an exercise, we invite the reader to prove strong normalization of $\stlc$ via induction up to $\blacksquare$. The reader should compare these short and simple proofs with more traditional ones, see e.g.~\cite{DBLP:journals/corr/abs-1907-11133}.

All the above results and observations for \stlc can be generalized and developed at the level of general higher-order abstract GSOS laws. To this end, we first abstract the behaviour functor \eqref{eq:lambda-beh} to a functor of the form $B(X,Y)=(X\monto Y)\times B'(X,Y)$, where $(-)\monto (-)$ is the internal hom-functor of a suitable closed monoidal structure on the base category $\C$. In the case of $\stlc$, this structure is given by Fiore's \emph{substitution tensor}~\cite{DBLP:journals/mscs/Fiore22}. Second, we observe that the higher-order GSOS law of $\stlc$ is an instance of a special kind of law that we coin \emph{relatively flat $\lambda$-laws}. The induction-up-to-$\blacksquare$ technique of \Cref{thm:ind-up-to-blackquare} then can be shown to hold for operational models of relatively flat $\lambda$-laws. More details can be found in \Cref{app:lambdashenanigans}.

\takeout{
It is well-known that logical predicates do not suffice for inductive proofs in
the $\lambda$-calculus, see e.g.~\cite[p. 9]{DBLP:journals/corr/abs-1907-11133} and
\cite[p. 150]{DBLP:books/daglib/0005958}; rather, the trick is to prove the
\emph{open extension} of a logical predicate, as it allows for a sufficiently
strong inductive hypothesis for the proof to go through. This technical problem
is present in our abstract setting in a more fundamental way, and it leads to
the introduction of \emph{$\lambda$-laws}.
Due to space restrictions, we do not expand on the full theory of
$\lambda$-laws; instead; we present an \emph{instantiation} of our main result
on the simply typed $\lambda$-calculus, which requires minimal groundwork, and
we motivate the need for $\lambda$-laws by explaining the
underlying technical problem.

Consider a simply typed, intrinsically typed, call-by-name $\lambda$-calculus a
base type of $\utype$, which we call \stlc{}. Typing contexts are maps
$\Gamma\c X \to \Tyl$ from finite
sets $X$ regarded as sets of variables to their corresponding types.
We use the
common notation $\Gamma \vdash t \c \tau$ for terms-in-context. The operational
semantics of \stlc{} are given by the standard rules 
\begin{equation*}
  \inference{t \xto{} t'}{t \app{}{}s \to t' \app{}{}s} \qquad
  \inference{}{(\lambda x \c \tau_{1}.\,t) \app{}{}s \to t[s/x]}
\end{equation*}
Instantiated to \stlc{}, our main result (\Cref{th:lambdamain2}) reads as
follows:
\begin{theorem}
  \label{th:instant}
  Let $P$ be a predicate on typed $\lambda$-terms. In order to prove that $\logp
  P$ holds for all terms, it suffices to show that
  \begin{enumerate}
  \item Variables satisfy $P$ (e.g. all variables terminate).
  \item The unit expression $\Gamma \vdash \mathsf{e} \c \utype$ satisfies
    $P_{\utype}(\Gamma)$.
  \item For all terms of the form $p \app{}{} q$, with $\Gamma \vdash p
    \c {\arty{\tau_{1}}{\tau_{2}}}$ and $\Gamma \vdash q \c
    \tau_{1}$ such that $\logp_{\arty{\tau_{1}}{\tau_{2}}}
    P(\Gamma)(p)$ and $\logp_{\tau_{1}} P(\Gamma)(q)$, we have $P_{\tau_{2}}(\Gamma)(p
    \app{}{} q)$.
  \item For all terms of the form $\lambda x \c \tau_{1}.\,t$ with $\Gamma,x \vdash t
    \c \tau_{2}$ such that $t \in \logp_{\tau_{2}}P(\Gamma,x)$, then
    $P_{\arty{\tau_{1}}{\tau_{2}}}(\Gamma,x)(\lambda x \c \tau_{1}.\,t)$.
  \end{enumerate}
\end{theorem}
Contrast the above with standard proofs via logical predicates, where the
goal is to inductively prove that the \emph{open
  extension} of $\logp P$ holds,
e.g. one would have to show that for a term $\lambda x \c \tau.\,t$ where $t$
satisfies the open extension of $\logp P$ and a substitution $\sigma$ that
satisfies $\logp P$, $(\lambda x \c \tau.\,t)[\sigma]$ satisfies $\logp P$, a
task decidedly more laborious.

\begin{remark}
  \Cref{th:instant} is a corollary of all $\lambda$-terms satisfying the \emph{open
    extension} of $\logp P$ (the \emph{fundamental property} of the logical
  predicate). This is because if variables satisfy $P$ then
  they necessarily satisfy $\logp P$, which suffices to show that $\logp P$ holds
  for all terms.
\end{remark}
We briefly explain why \Cref{th:main2} does not apply to the
$\lambda$-calculus. The syntax of \stlc{} is given in terms of a typed binding signature which,
following~\cite{DBLP:conf/fossacs/Fiore05,DBLP:conf/csl/FioreH10,
DBLP:journals/mscs/Fiore22}, corresponds to an endofunctor $\Sigma \c
(\Set^{\fset/{\Tyl}})^{\Tyl} \to (\Set^{\fset/{\Tyl}})^{\Tyl}$ in the category
of \emph{(covariant) presheaves on typed variable contexts}. In
particular~\cite[p. 1038]{DBLP:journals/mscs/Fiore22},
\begin{equation}
  \label{eq:lambdasigma}
  \begin{aligned}
    \Sigma_{\utype}X =
    &\; V_{\utype} + K_{1}
      + \coprod_{\tau \in \Tyl}X_{\arty{\tau}{\utype}}
      \times X_{\tau}\\
    \Sigma_{\arty{\tau_{1}}{\tau_{2}}} X=
    &\; V_{\arty{\tau_{1}}{\tau_{2}}} + \delta^{\tau_{1}}_{\tau_{2}} X
      + \coprod_{\tau \in \Tyl}X_{\arty{\tau}{\arty{\tau_{1}}{\tau_{2}}}}
                                           \times X_{\tau}
  \end{aligned}
\end{equation}
where $\delta^{\tau}$ is the ``variable binding'' endofunctor on type $\tau$ and
$V$ is the presheaf of \emph{variables}, which is defined in terms of the Yoneda
embedding: $V_{\tau}(\Gamma) \cong \{x \in |\Gamma| \mid \Gamma(x) = \tau\}$.
Category $(\Set^{\fset/{\Tyl}})^{\Tyl}$ is closed monoidal w.r.t. the so-called
\emph{substitution tensor}, whose unit is the presheaf $V$ and whose right
adjoint $\llangle \argument, Y \rrangle$, modelling \emph{simultaneous
  substitution}, is defined as follows~\cite[Eq. (4)]{DBLP:conf/csl/FioreH10}:
\begin{equation}
  \llangle X,Y \rrangle_{\tau}(\Gamma) = \Set^{\fset/{\Tyl}}\Bigl(\prod_{i \in
    |\Gamma|}X_{\Gamma(i)}, Y_{\tau}\Bigr).
\end{equation}
More details on $(\Set^{\fset/{\Tyl}})^{\Tyl}$ are provided in
\Cref{app:lambdacat}.
The higher-order abstract GSOS approach to the untyped
$\lambda$-calculus~\cite[\textsection 5]{gmstu23} carries over to the typed
case: there is a $V$-pointed higher-order GSOS law
\begin{equation}
  \label{eq:lambdarho}
  \rho_{X,Y} \c \Sigma (X \times \llangle X,Y \rrangle\times B(X,Y))\to \llangle X,\Sigma^\star (X+Y) \rrangle\times B(X, \Sigma^\star (X+Y)),
\end{equation}
defined in a way analogous to~\cite[Def. 5.8]{gmstu23}, of
$\Sigma$~\eqref{eq:lambdasigma} over $\llangle \argument, \argument \rrangle \times B(\argument,\argument)$ 
where $B(X,Y) = K_{1} + Y + D^{\lambda}(X,Y)$,
where $D^{\lambda} \c
\bigl((\Set^{\fset/{\Tyl}})^{\Tyl}\bigr)^{\opp} \times
(\Set^{\fset/{\Tyl}})^{\Tyl} \to (\Set^{\fset/{\Tyl}})^{\Tyl}$ is as follows:
  $D^{\lambda}_{\utype}(X,Y) = K_{1}$, $D^{\lambda}_{\arty{\tau_{1}}{\tau_{2}}}(X,Y) = X_{\tau_{2}}^{Y_{\tau_{1}}}$.
Ideally, we would look to apply \Cref{th:main2} on the $V$-pointed higher-order
GSOS law of \stlc{}. However, the fundamental problem we face here is that only the 
bifunctor $B$ is contractive, while $\llangle \argument, \argument \rrangle \times B(\argument,\argument)$ 
is not. 
Intuitively, should we try to prove strong normalization like in
\Cref{ex:strongnorm1}, we would face a familiar situation (see e.g. ~\cite[p.
9]{DBLP:journals/corr/abs-1907-11133} and \cite[p.
150]{DBLP:books/daglib/0005958}): the strengthening of
the inductive hypothesis that $\logp P$ provides does not cover $\lambda$-abstractions,
whose semantics requires the $\llangle \argument, \argument \rrangle$ part of the 
behaviour functor. Instead, we have to build~$\logp P$ for $B$ and subsequently
apply the standard trick of extending $\logp P$ to open terms.

Our notion of a $\lambda$-law identifies the general patterns of variable
binding and substitution, as well as the inner working of
$\lambda$-abstractions, and allows for our earlier result on induction up to
$\logp$ to carry over to languages with variable binding. The full theory of
$\lambda$-laws is presented in \Cref{app:monclosed}.
}

\section{Strong Normalization for Deterministic Systems, Abstractly}\label{sec:det}
The high level of generality in which the theory of logical predicates is
developed above enables reasoning uniformly about whole families of languages and behaviours. In this section, we narrow our focus to deterministic
systems and establish a general strong normalization criterion, which can be checked in concrete instances 
by mere inspection of the operational rules corresponding to higher-order abstract GSOS laws. 

Throughout this section, we fix a $0$-pointed higher-order GSOS law $\rho$ of a signature endofunctor $\Sigma\c \C\to \C$ over a behaviour bifunctor $B\c \C^\op\times \C\to \C$, where
\[ B(X,Y)=Y+D(X,Y) \quad\text{for some}\quad D\colon \C^\op\times \C\to \C. \]
For instance, the type functor \eqref{eq:beh} for $\stsc$ is of that form.
The operational model $\gamma\c \mS\to \mS+D(\mS,\mS)$ has an $n$-step extension $\gamma^{(n)}\colon \mS\to \mS+D(\mS,\mS)$, for each $n\in \Nat$, where $\gamma^{(0)}$ is the left coproduct injection and $\gamma^{(n+1)}$ is the composite
\[
\mS\xrightarrow{\gamma}\mS+D(\mS,\mS) \xrightarrow{\gamma^{(n)}+\id} \mS+D(\mS,\mS)+D(\mS,\mS) \xrightarrow{\id+\nabla} \mS+D(\mS,\mS).\]

We regard $D(\mS,\mS)$ as a predicate on $B(\mS,\mS)$ via the right coproduct injection, which is monic by extensivity of $\C$, and define the following predicates on $\mS$:
\[  
\Downarrow_n \,=\, (\gamma^{(n)})^\star[D(\mS,\mS)]\qqand \Downarrow\,=\,\bigor_{n} \Downarrow_n.
\]
In \stsc, these are the predicates of strong normalization or strong normalization after at most $n$ steps, resp. Accordingly, we define strong normalization abstractly as follows:

\begin{defn} The higher-order GSOS law $\rho$ is \emph{strongly normalizing} if $\Downarrow\,=\mS$.  
\end{defn}

We next identify two natural conditions on the law $\rho$ that together ensure strong normalization. The first roughly asserts that for a term $t=\f(x_1,\ldots,x_n)$ whose variables~$x_i$ are non-progressing, the term $t$ is either non-progressing or it progresses to a variable.

\begin{defn}
The higher-order GSOS law $\rho$ is \emph{simple} if its components $\rho_{X,Y}$ restrict to morphisms $\rho_{X,Y}^0$ as in the diagram below, where $\eta$ is the unit of the free monad $\Sigmas$:
\[
\begin{tikzcd}
\Sigma(X\times D(X,Y)) \ar[dashed]{r}{{\rho}_{X,Y}^0} \ar{d}[swap]{\Sigma(\id\times \inr)} & X+Y+D(X,\Sigmas(X+Y)) \ar{d}{\eta_{X+Y}+\id} \\
\Sigma(X\times (Y+D(X,Y)) \ar{r}{\rho_{X,Y}} & \Sigmas(X+Y)+D(X,\Sigmas(X+Y))
\end{tikzcd}
\]  
\end{defn}
The second condition asserts that the rules represented by the higher-order GSOS law remain sound when strong transitions are replaced by weak ones. In the following, the \emph{graph} of a morphism $f\colon A\to B$ is the image $\gra(f)\monoto A\times B$ of $\langle \id,f\rangle\colon A\to A\times B$.

\begin{defn}
The higher-order GSOS law $\rho$ \emph{respects weak transitions} if for every $n\in \Nat$, the graph of the composite below is contained in $\bigvee_k\, \gra(\gamma^{(k)}\comp \ini)$.
\[ \Sigma(\mS)\xrightarrow{\Sigma\langle \id,\gamma^{(n)}\rangle} \Sigma(\mS\times B(\mS,\mS)) \xrightarrow{\rho_{\mS,\mS}} B(\mS,\Sigmas(\mS+\mS)) \xrightarrow{B(\id, \hat{\ini}\comp \Sigmas\nabla)} B(\mS,\mS)
\]
\end{defn}
\noindent Note that the higher-order GSOS law for $\stsc$ is simple and respects weak transitions. Thus, strong normalization of $\stsc$ is an instance of the following strong normalization theorem for higher-order abstract GSOS. Concerning its conditions, an \emph{$\omega$-directed union} is a colimit of an $\omega$-chain $X_0\monoto X_1\monoto X_2\monoto\cdots$ of monics. We say that monos in $\C$ are \emph{$\omega$-smooth} if any such colimit has monic injections, and moreover for every compatible cocone of monos, the mediating morphism is monic. This property holds in every locally finitely presentable category~\cite[Prop. 1.62]{adamek_rosicky_1994}, e.g.\ sets, posets, or presheaves.

\begin{theorem}[Strong normalization]\label{thm:strong-normalization}
Suppose that the following conditions hold:
\begin{enumerate}
\item On top of \Cref{assumptions}, $\C$ is countably extensive, and monos are $\omega$-smooth.
\item $\Sigma$ preserves $\omega$-directed unions, and $D$ preserves monos in the second component.
\item $\rho$ is relatively flat, simple, and respects weak transitions.
\item $\Downarrow$ has a locally maximal logical refinement w.r.t.\ $\gamma$ and the canonical lifting $\ol{B}$.  
\end{enumerate}
Then the higher-order GSOS law $\rho$ is strongly normalizing.
\end{theorem}
Recall that condition (4) holds if $\ol{B}$ is contractive (\Cref{thm:square-refinement}). The proof uses the induction-up-to-$\square$ technique and a careful categorical abstraction of \Cref{ex:tclnorm}.

\section{Conclusion and Future Work}
\label{sec:conclusion}

Our work presents the initial steps towards a unifying,
efficient theory of logical relations for higher-order languages based on
higher-order abstract GSOS. This theory can be broadened in various directions.
One obvious direction would be to extend our theory from
predicates to relations. Binary logical relations are often utilized as sound
(and sometimes complete) relations w.r.t. \emph{contextual equivalence}.
Additional generalizations are suggested by the large amount of
existing work on logical relations. One important direction is to
generalize the type system to cover, e.g., recursive types, parametric
polymorphism, or dependent types. Supporting recursive types will
presumably require an adaptation of the method of
step-indexing~\cite{5230591} to our abstract setting. Another point of
interest is to apply and extend our framework to effectful (e.g.\ probabilistic) 
settings~\cite{DBLP:journals/pacmpl/LagoG22,Pitts:1999:ORF:309656.309671}, including e.g.\ an effectful version of the criterion of \Cref{sec:det}.

As indicated in \Cref{rem:fibrations}, large parts of our development in \Cref{sec:logpred}
can be reformulated in fibrational terms. This has the potential merit
of enabling abstract reasoning about higher-order programs in metric
and differential settings as done in previous work on fine-grain
call-by-value~\cite{DBLP:conf/fscd/DagninoG22,DBLP:journals/corr/abs-2303-03271}. In
future work, we aim to develop such a generalization, and to explore
the connection between our weak transition semantics and the general
evaluation semantics used in \emph{op.\ cit.}

\clearpage

\renewcommand{\c}{\cedilla}

\bibliography{../mainBiblio}
\bibliographystyle{splncs04}

\clearpage
\appendix

\begin{adjustwidth}{-3em}{-3em}
\section{Omitted Details and Proofs}
\label{app:omitted-proofs}

\section*{Proof of \Cref{prop:bif}}
The category $\Set^\Ty$ has as objects $\Ty$-indexed families of sets, and as morphisms $\Ty$-indexed families of functions between them. The identities and composition are given componentwise by those of $\Set$. If we write $\Set(-,-) \colon \Set^\opp \times \Set \to \Set$ for the hom-functor of $\Set$,
the action of $D$ on morphisms $f = (f_\tau \colon X'_\tau \to X_\tau)_{\tau \in \Ty}$ and $g = (g_\tau \colon Y_\tau \to Y'_\tau)_{\tau \in \Ty}$ is given by
\[
  D_\utype(f,g) = \id_\term \qquad D_{\arty{\tau_{1}}{\tau_{2}}}(f,g) = \Set(f_{\tau_1},g_{\tau_2}).
\]
Functoriality of $D$ follows directly from that of the constant functor into $\term$ and of $\Set(-,-)$
The action of $B$ on morphisms $f$ and $g$ as above is simply given by
\[
  B_\tau(f,g) = g_\tau + D_\tau(f,g).
\]
Functoriality of $B$ follows from that of $+$, $D$, and $\Id_{\Set}$.

\section*{Full Definition of $\rho$ for \stsc{}}
\label{page:rho-stsc}

\begin{align*}
  \rho_{X,Y} \c \quad
  & \Sigma(X \times B(X,Y))
  & \to \quad
  & B(X, \Sigma^\star (X+Y)) \\
  \rho_{X,Y}(tr) = \quad
  & \texttt{case}~tr~\texttt{of} \\
  & \mathsf{e} & \mapsto \quad & * & \\
  & S_{\tau_{1},\tau_{2},\tau_{3}} & \mapsto \quad & \lambda t.S'(t) & \\
  & S'_{\tau_{1},\tau_{2},\tau_{3}}(p,f) & \mapsto \quad & \lambda t.S''(p,t) & \\
  & S''_{\tau_{1},\tau_{2},\tau_{3}}((p,f),(q,g))
  & \mapsto \quad
  & \lambda t.(p \app{}{} t)\app{}{}(q \app{}{} t) & \\
  & K_{\tau_{1},\tau_{2}} & \mapsto \quad & \lambda t.K'(t) & \\
  & K'_{\tau_{1},\tau_{2}}(p,f) & \mapsto \quad & \lambda t.p & \\
  & I_{\tau} & \mapsto \quad & \lambda t.t & \\
  & (p,f)\app{}{}(q,g) & \mapsto \quad
  &
    \begin{cases}
      f(q) \text{\quad if $f \c Y_{\tau_{2}}^{X_{\tau_{1}}}$} \\
      f \app{}{} q \text{\quad if $f \c Y_{\arty{\tau_{1}}{\tau_{2}}}$}
    \end{cases} &
\end{align*}

\section*{Proof of \Cref{prop:free-monad-lift}}

This proposition is best understood in the more general context of lifting free monads to arrow categories, which is discussed in \Cref{sec:free-monad-lift}.

\section*{Proof of \Cref{prop:liftingb}}
We defined $\overline B$ on objects, we now need to define it on morphisms in
such a way that the right-hand diagram in~\eqref{eq:liftingd} commutes. To this
end, consider objects $p \c P \pred{} X$, ${q \c Q \pred{} Y}$, $p' \c P' \pred{} X'$,
$q' \c Q' \pred{} Y$ of $\Pred\C$, and morphisms $f \colon p \to p'$, $g \colon q
\to q'$.
We prove that~$B(f,g)$ canonically defines an arrow in $\Pred \C$ from
$m_{P',Q}\colon \ol{B}(P',Q)\monoto B(X',Y)$ to $m_{P,Q'}\colon \overline B(P,Q')\monoto B(X,Y')$, by showing that there exists
a morphism ${B(f,g)}\mid_{\overline B(P',Q)} \colon \overline B(P',Q) \to
\overline B(P,Q')$ in $\C$ that makes the following square commutative:
\begin{equation*}
  \begin{tikzcd}[column sep=6em, row sep=5ex]
    \overline B(P',Q)
    \ar[r,dashed,"{B(f,g)}\mid_{\overline B(P',Q)}"]
    \ar[d,"{m_{P',Q}}"',tail]
    &
    \overline B(P,Q')
    \ar[d,"{m_{P,Q'}}",tail]
    \\
    B(X',Y)
    \rar["{B(f,g)}"] &
    B(X,Y')
  \end{tikzcd}
\end{equation*}
Consider the diagram below, where $T_{P',Q}$ is the pullback of $B(p',Y)$ against $B(P',q)$:
\begin{equation*}
  \begin{tikzcd}[column sep=2.5em, row sep=5ex]
    {T_{P',Q}}\ar[dr,"q", dashed]\ar[rr,"s_{P',Q}"]\ar[dd,"r_{P',Q}"'] &
    & {B(P',Q)}\dar["{B(f\mid_P,g\mid_Q)}"] \ar[dddr,bend left=40,"{B(\id,q)}"]
    \\[-2ex]
    &
    {T_{P,Q'}} \pullbackangle{-45}\dar["r_{P,Q'}"']\rar["s_{P,Q'}"]
    & {B(P,Q')}\dar["{B(\id,q')}"]
    \\
    B(X',Y)\rar["{B(f,g)}"]\ar[drrr,bend right=10, "{B(p',\id)}"] & {B(X,Y')}\rar["{B(p,\id)}"] & B(P,\id)\\[-4ex]
    &&&[2ex] B(P',Y)\ar[lu, "{B(f\mid_P,g)}"']
  \end{tikzcd}
\end{equation*}
The bottom and rightmost diagrams commute by functoriality of $B$, hence there is a unique $q$ that makes the top and leftmost diagrams commutative.
If we factorize $r_{P',Q}$ and $r_{P,Q'}$ then, we get that the outer diagram below commutes. Via diagonal fill-in we obtain the unique arrow in the middle that makes both inner square commutative, in particular the bottom one is what we desired.
\begin{equation*}
  \begin{tikzcd}[column sep=6em, row sep=5ex]
    T_{P',Q}
    \rar["q"]
    \dar["{e_{P',Q}}"', two heads]
    \arrow["r_{P',Q}"',rounded corners,
    to path={[pos=0.75]
			-| ([xshift=-.5cm]\tikztotarget.west)\tikztonodes
			-- (\tikztotarget)
    }]{dd}
    &
    T_{P,Q'}
    \dar["{e_{P,Q'}}", two heads]
    \arrow["r_{P,Q'}",rounded corners,
    to path={[pos=0.75]
      -| ([xshift=.5cm]\tikztotarget.east)\tikztonodes
      -- (\tikztotarget)
    }]{dd}
    \\
    \overline B(P',Q)
    \rar["{B(f,g)}\mid_{\overline B(P',Q)}",dashed]
    \dar["{m_{P',Q}}"', tail]
    &
    \overline B(P,Q')
    \dar["{m_{P,Q'}}", tail]
    \\
    B(X',Y)
    \rar["{B(f,g)}"] &
    B(X,Y')
  \end{tikzcd}
\end{equation*}
Functoriality of $\overline B$ so defined is a straightforward consequence of the uniqueness of the restriction-corestriction of $B(f,g)$ for arbitrary $f$ and $g$.

\section*{Details for \Cref{ex:logical-plotkin}}
  We now summarize a proof of $\DDar_{\tau}t$ for all $t \c \tau$ by
  structural induction on~$\Tr$. The logical predicate $\DDar$ was
  previously defined in \Cref{ex:log1}. Coinduction is not necessary for this
  proof, but extra work needs to be done in order to show that $\DDar$ is
  preserved and reflected by $\beta$-reduction (that is, if
  $\DDar_{\tau}t$ and $t \To t'$ then $\DDar_{\tau}(t')$ and
  vice-versa), as neither is
  immediate. This is in line with standard proofs of strong normalization that
  use this kind of logical predicates~\cite[p.8]{DBLP:journals/corr/abs-1907-11133}. 
  Here, we postulate the above fact and show some select cases.
  \begin{itemize}
  \item Case $I_{\tau} \c \arty{\tau}{\tau}$.
    \begin{enumerate}
    \item $I$ indeed terminates, as $I \xto{s} s$.
    \item It suffices to show $\DDar_{\tau}s \implies
      \DDar_{\tau}s$, which is a tautology.
    \end{enumerate}
  \item Case $t \app{}{}s$ with
    $\DDar_{\arty{\tau_{1}}{\tau_{2}}}t$ and
    $\DDar_{\tau_{1}}s$.
    \begin{enumerate}
    \item By the induction hypothesis, we know that $t$ eventually terminates,
      thus $t \xTo{s} t'$ and $t \app{}{}s \To t'$. Moreover, by the definition of
      $\DDar$, we conclude that $\DDar_{\tau_{2}} t'$. This means
      that $t'$ terminates, and thus so does $t \app{}{}s$.
    \item We just proved that $\DDar_{\tau_{2}} t'$ and $t \app{}{}s \To t'$.
      Logical predicate $\DDar$ is reflected by $\beta$-reductions, hence~$\DDar_{\tau_{2}} t\app{}{}s$.
    \end{enumerate}
  \item Case $S''_{\tau_{1},\tau_{2},\tau_{3}}(t,s) \c \arty{\tau_{1}}{\tau_{3}}$
    with $\DDar_{\arty{\tau_{1}}{\arty{\tau_{2}}{\tau_{3}}}}t$
    and $\DDar_{\arty{\tau_{1}}{\tau_{2}}} s$.
    \begin{enumerate}
    \item $S''(t,s)$ terminates, as $S''(t,s) \xto{r}
      (t \app{}{}r)\app{}{}(s\app{}{}r)$.
    \item Given $\DDar_{\tau_{1}}r$, we have to show
      $\DDar_{\tau_{3}}(t\app{}{}r)\app{}{}(s\app{}{}r)$. For this, all we have to
      do is repeat the steps in the case of application, first for
      for $t \app{}{}r$ and $s \app{}{}
      r$, and finally for $(t\app{}{}r)\app{}{}(s\app{}{}r)$.
    \end{enumerate}
  \item Case $S'_{\tau_{1},\tau_{2},\tau_{3}}(t)$ with
    $\DDar_{\arty{\tau_{1}}{\arty{\tau_{2}}{\tau_{3}}}}t$.
    \begin{enumerate}
    \item $S'(t)$ terminates, as $S'(t) \xto{s} S''(t,s)$.
    \item Given $\DDar_{\arty{\tau_{1}}{\tau_{2}}}s$, one needs to prove
      that $\DDar_{\arty{\tau_{1}}{\tau_{3}}}S''(t,s)$. We in turn repeat
      the steps in the case of $S''$.
    \end{enumerate}
  \end{itemize}

\section*{Details for \Cref{ex:logical-tait}}
\label{app:log}

  In order to show that
  $\mathrm{SN}$ is logical for $\tsup{\gamma}$, it suffices to show that
  \[
    \forall t \c \arty{\tau_{1}}{\tau_{2}}. \,
    \mathrm{SN}_{\arty{\tau_{1}}{\tau_{2}}}(t) \implies
    \forall s \c \tau_{1}.\,
    \mathrm{SN}_{\tau_{1}}(s) \land (t \xTo{s} t' \implies
    \mathrm{SN}_{\tau_{2}}(t')).
  \]
  By definition of $\mathrm{SN}$, we know that $\mathrm{SN}_{\tau_{2}}(t \app{}{} s)$. In addition, as $t$ terminates, it is clear that $t \app{}{}s \To t'$. As
  such, it suffices to show that $\mathrm{SN}$ is closed under finite sequences
  of $\beta$-reductions. By definition of $\widetilde{\gamma}$ and
  ignoring the obvious reflexive case, we need to show that
  \[
    \forall t,p,t'\c \tau.\, \mathrm{SN}_{\tau}(t) \land t \to p \land p \To t'
    \implies \mathrm{SN}_{\tau}(t').
  \]
  Should $\mathrm{SN}_{\tau}(p)$, then induction on $\To$ would complete the
  proof. As such, it suffices to show that~$\mathrm{SN}$ is closed under (small-step)
  $\beta$-reductions, which is done by structural induction on $\Ty$ and it
  amounts to proving that
  \[
    \forall t,t'\c  \tau.\, \mathrm{SN}_{\tau}(t) \land t \to t' \implies {\Downarrow}_\tau\, t'
  \]
  and
  \[
    \forall t,t'\c \arty{\tau_{1}}{\tau_{2}}.\,
    \mathrm{SN}_{\arty{\tau_{1}}{\tau_{2}}}(t)  \land t \to t' \implies ( \forall
    s.\,\mathrm{SN}_{\tau_{1}}(s)
    \implies \mathrm{SN}_{\tau_{2}}(t' \app{}{}s)).
  \]
  The former is true as ${\Downarrow}_\tau\, t$ and the semantics is
  deterministic. For the latter, note that by~$\mathrm{SN}_{\arty{\tau_{1}}{\tau_{2}}}(t)$ we have that
  $\mathrm{SN}_{\tau_{2}}(t \app{}{}s)$ and $t \app{}{}s \to t' \app{}{}s$ by the
  definition of $\beta$-reduction. Type
  $\tau_{2}$ is structurally smaller than $\arty{\tau_{1}}{\tau_{2}}$ and thus
  induction on $\Ty$ gives $\mathrm{SN}_{\tau_{2}}(t' \app{}{}s)$. %

\section*{Proof of \Cref{pro:square}}

\begin{lem}\label{lem:int_cont}
Given $P,Q,S\pred{} X$, $d(P\land S,Q\land S)\leq d(P,Q)$. 
\end{lem}
\begin{proof}
Indeed, since $P\land S = \iimg{s}{P}$ and $Q\land S = \iimg{s}{Q}$ where $s\c S\pred{} X$,
we are done by~\Cref{ass:contr}~(2).
\end{proof}

By assumption, there is $u\in [0,1)$, such that for every $p\c P\pred{} X$ 
\begin{align*}
d(\ol{B}(S,P),\ol{B}(S',P))\leq u\cdot d(S,S').
\end{align*}
We will show that for every ordinal $\alpha$, 
\begin{align}\label{eq:ord_cont}
  d(\invp^\alpha(S, P),\invp^\alpha(S', P))\leq u\cdot d(S,S'),
\end{align} 
where $\invp^{\alpha+1}(S, P) = P \land \iimg{c}{\overline{B}(S,\invp^\alpha(S, P))}$,
and $\invp^\alpha(S, P) = \bigand_{\beta<\alpha} \invp^\beta(S, P)$ for every limit ordinal~$\alpha$,
in particular, ${\invp^0(S, P) = \top}$. Using the fact that $\invp(S, P) = \bigand_{\alpha}\invp^\alpha(S, P)$
where the meet is taken over all ordinals $\alpha$, we will be able to conclude 
that $\invp$ is contractive in the first argument as follows:
\begin{flalign*}
&& d(\invp(S, P),\invp(S', P)) &\;= d\Bigl(\bigand_{\alpha}\invp^\alpha(S, P),\bigand_{\alpha}\invp^\alpha(S', P)\Bigr)&\\
&&  &\;\leq \sup_{\alpha} d(\invp^\alpha(S, P),\invp^\alpha(S', P))&\by{\Cref{ass:contr}~(3)}\\
&&  &\;\leq u\cdot d(S,S').&\by{\eqref{eq:ord_cont}}
\end{flalign*}
We proceed to show~\eqref{eq:ord_cont} by transfinite induction over $\alpha$. 
\begin{itemize}
  \item \emph{successor ordinals:} Note that
  \begin{flalign*}
    &&d(\invp^{\alpha+1}(S, P),&\invp^{\alpha+1}(S', P)) \\*
    &&=&\; d(P \land \iimg{c}{\overline{B}(S,\invp^\alpha(S, P))}, P \land \iimg{c}{\overline{B}(S',\invp^\alpha(S', P))})\\
    &&\leq&\; d(\iimg{c}{\overline{B}(S,\invp^\alpha(S, P))},\iimg{c}{\overline{B}(S',\invp^\alpha(S', P))})&\by{\Cref{lem:int_cont}}\\
    &&\leq&\; d({\overline{B}(S,\invp^\alpha(S, P))},{\overline{B}(S',\invp^\alpha(S', P))})&\by{\Cref{ass:contr}~(2)}\\
    &&\leq&\; \max\{d({\overline{B}(S,\invp^\alpha(S, P))},{\overline{B}(S',\invp^\alpha(S, P))}),&\\
    &&&\hspace{2.6em}
                  d({\overline{B}(S',\invp^\alpha(S, P))},{\overline{B}(S',\invp^\alpha(S', P))})\},&\by{ultrametric inequality}
  \end{flalign*}
and we are left to show that 
\begin{align*}
d({\overline{B}(S,\invp^\alpha(S, P))},{\overline{B}(S',\invp^\alpha(S, P))})\leq&\; u\cdot d(S,S'),\\
d({\overline{B}(S',\invp^\alpha(S, P))},{\overline{B}(S',\invp^\alpha(S', P))})\leq&\; u\cdot d(S,S').
\end{align*}
The first inequality is due to the uniform contractivity assumption for $\ol{B}.$
The second inequality is by non-expansiveness of $\overline{B}(S',\argument)$, combined 
with the induction hypothesis.
  \item \emph{limit ordinals:}    
  \begin{flalign*}
    &&  d(\invp^{\alpha}(S, P),\invp^{\alpha}(S', P)) 
    &\;= d\Bigl(\bigand_{\beta<\alpha} \invp^\beta(S, P),\bigand_{\beta<\alpha} \invp^\beta(S', P)\Bigr)\\
    &&  &\;\leq \sup_{\beta<\alpha} d(\invp^\beta(S, P),\invp^\beta(S', P))&\by{\Cref{ass:contr}~(3)}\\
    &&  &\;\leq u\cdot d(S,S')&\by{induction hypothesis}
  \end{flalign*}
\end{itemize}
It is shown analogously that $\invp^{\ol{B},c}$ is non-expansive in the second 
argument.

\section*{Proof of \Cref{th:main2}}

Again, we omit the superscripts $\gamma,\ol{B}$ to ease the notation.
Assume that $\fimg{\iota}{\ol{\Sigma} \logp P} \leq P $. 
Since~$\logp P$ is a greatest invariant relative to itself that implies $P$, 
to show $\logp P =\top$, it suffices to show that
$\fimg{\iota}{\ol{\Sigma} \logp P} \leq \square P$, for then $\square P =\top$.
To show $\fimg{\iota}{\ol{\Sigma} \logp P} \leq \square P$, we establish
\[
  \fimg{(\iota \comp \inj_j)}{\ol{\Sigma}_{j}\logp P } \leq \logp P.
\]
for all $j \in J$, by well-founded induction. Specifically, we proceed to prove
that $\logp P \lor \fimg{(\iota \comp \inj_j)}{\ol{\Sigma}_{j} \logp P} \pred{} \mS$ is a
$\logp P$-relative invariant, which is sufficient. By preservation of unions under direct images,
it suffices to show
\begin{align*}
  \fimg{\gamma}{\logp P} \leq \overline{B}(\logp P,\logp P \lor
  \fimg{(\iota \comp \inj_j)}{\ol{\Sigma}_{j} \logp P}), \text{and}\\*
  \fimg{(\gamma \comp \iota \comp \inj_j)}{\ol{\Sigma}_{j} \logp P}
  \leq \overline{B}(\logp P,\logp P \lor
  \fimg{(\iota \comp \inj_j)}{\ol{\Sigma}_{j} \logp P}).
\end{align*}
The first inequality follows from the definition of $\logp P$. The second one
states that $\gamma \comp \iota \comp \inj_j$ extends to a morphism of predicates,
i.e.\ that the diagram
\[
  \begin{tikzcd}
    \ol{\Sigma}_{j} \logp P
    \ar[dashed]{rr}{}
    \ar[tail]{d}{}
    &[3em]
    & \overline{B}(\logp P,\logp P \lor
    \fimg{(\iota \comp \inj_j)}{\ol{\Sigma}_{j} \logp P})
    \dar[tail]{}
    \\
    \Sigma_{j}(\mS)
    \ar{r}{\iota \comp \inj_j}
    & \mS
    \rar{\gamma}
    & B(\mS,\mS)
  \end{tikzcd}
\]
commutes for a suitably chosen top horizontal morphism.
The requisite morphism is entailed by commutativity of the following diagram:
\[
  \adjustbox{scale=0.75,center}{
    \begin{tikzcd}[sep=scriptsize]
      {\ol\Sigma_{j}\logp P} &[-2em]& {\Sigma_{j} \mS} &[2em] \Sigma\mS &[-3em] \mS \\
      \\
      {\ol\Sigma_{j}(\logp P \times\ol{B}(\logp P,\logp P))} && {\Sigma_{j}(\mS \times B(\mS,\mS))} & {\Sigma(\mS \times B(\mS,\mS))} \\
      \\
      {\ol{B}(\logp P, \ol\Sigma^{\star}_{\prec j}(\logp P + \logp P) + \ol\Sigma_{j}\ol\Sigma^{\star}_{\prec j}(\logp P + \logp P))} && {B(\mS,\Sigma^{\star}_{\prec j}(\mS + \mS) + \Sigma_{j}\Sigma^{\star}_{\prec j}(\mS + \mS))} & {B(\mS,\Sigma^{\star}(\mS + \mS))} \\
      \\
      {\ol{B}(\logp P, \ol\Sigma^{\star}_{\prec j}\logp P + \ol\Sigma_{j}\ol\Sigma^{\star}_{\prec j}\logp P))} && {B(\mS,\Sigma^{\star}_{\prec j}\mS + \Sigma_{j}\Sigma^{\star}_{\prec j}\mS)} & {B(\mS,\Sigma^{\star}\mS)} & {} \\
      \\
      {\ol{B}(\logp P, \logp P + \ol\Sigma_{j} \logp P)} && {B(\mS,\mS + \Sigma_{j}\mS)}\\
      \\
      {\overline{B}(\logp P,\logp P \lor
        \fimg{(\iota \comp \inj_j)}{\ol{\Sigma}_{j} \logp P})} &&&& {B(\mS,\mS)}
      \arrow[tail, from=1-1, to=1-3]
      \arrow[dashed, from=1-1, to=3-1]
      \arrow[dashed, from=3-1, to=5-1]
      \arrow[dashed, from=5-1, to=7-1]
      \arrow[dashed, from=9-1, to=11-1]
      \arrow["{\Sigma_{j}\langle\mathrm{id},\gamma\rangle}"', from=1-3, to=3-3]
      \arrow["{\rho^{j}_{\mS,\mS}}"', from=3-3, to=5-3]
      \arrow["{B(\mS,\Sigma^{\star}_{\prec{j}}\nabla + \Sigma_{j}\Sigma^{\star}_{\prec{j}}\nabla)}"', from=5-3, to=7-3]
      \arrow[tail, from=7-1, to=7-3]
      \arrow[tail, from=3-1, to=3-3]
      \arrow[tail, from=5-1, to=5-3]
      \arrow["{(1)}"{description}, draw=none, from=1-1, to=3-3]
      \arrow["{(2)}"{description}, draw=none, from=3-1, to=5-3]
      \arrow["{(3)}"{description}, draw=none, from=5-1, to=7-3]
      \arrow["{\inj_j}", from=1-3, to=1-4]
      \arrow["\iota", from=1-4, to=1-5]
      \arrow["{\Sigma\langle\id,\gamma\rangle}", from=1-4, to=3-4]
      \arrow["{\inj_{j}}", from=3-3, to=3-4]
      \arrow["{(5)}"{description}, draw=none, from=1-3, to=3-4]
      \arrow[dashed, from=7-1, to=9-1]
      \arrow[tail, from=9-1, to=9-3]
      \arrow["\gamma", from=1-5, to=11-5]
      \arrow["{\qquad (9)}"{description}, draw=none, from=1-4, to=11-5, pos=.3]
      \arrow["{B(\mS, \hat \iota)}", from=7-4, to=11-5]
      \arrow["{\rho_{\mS,\mS}}"', from=3-4, to=5-4]
      \arrow["{B(\mS,\Sigma^{\star}\nabla)}", from=5-4, to=7-4]
      \arrow["{(6)}"{description}, draw=none, from=3-3, to=5-4]
      \arrow["{B(\mS,[\id,\;\iota\comp \inj_j])}", from=9-3, to=11-5]
      \arrow["{B(\mS,e_{j})}", from=5-3, to=5-4]
      \arrow["{B(\mS, e_{j})}", from=7-3, to=7-4]
      \arrow["{(7)}"{description}, draw=none, from=5-3, to=7-4]
      \arrow["{(8)\hspace{3em}}"{pos=0.4}, shift left=1, draw=none, from=7-3, to=11-5]
      \arrow["{(4)\hspace{2em}}"{description}, draw=none, from=7-1, to=9-3]
      \arrow["{B(\mS,\;\hat\iota\comp\inj^{\klstar}_{\prec j}+\Sigma_{j}(\hat\iota\comp\inj^{\klstar}_{\prec j}))}"', from=7-3, to=9-3]
      \arrow[from=11-1,to=11-5]
      \arrow["{(10)\hspace{5em}}"{description}, draw=none, from=9-1, to=11-5]
    \end{tikzcd}}
\]
Here, (1) commutes by definition of $\square P$ and
functoriality of $\ol\Sigma_j$, (2) commutes by the assumption that $\rho^j$ lifts
to $\Pred[X]{\C}$, (3) by functoriality of the involved functors and the simple 
fact that~$\nabla$ lifts, and (4) commutes by functoriality and by noting that the 
induction hypothesis entails $\fimg{(\hat\iota \comp \inj_{\prec j})}{\ol{\Sigma}_{\prec j}^\klstar\logp P } \leq \logp P$.
Furthermore, (5) commutes by naturality, (6) by definition of $\rho$, (7) by naturality of $e_{i}$ and functoriality of $\ol B$, (8) by the definition of $e_{i}$, and (9) by the definition of $\gamma$.
Finally, (10) commutes, because coproducts lift to $\Pred[X]{\C}$.

\section*{Proof of \Cref{thm:strong-normalization}.}

\begin{rem}
We shall need some general properties of predicates in (countably) extensive categories that are collected in \Cref{sec:predicates}.
\end{rem}

\begin{rem}\label{rem:c-sigma-props}
Under the conditions of \Cref{thm:strong-normalization} the canonical predicate lifting $\ol{B}$ is given by
\[ \ol{B}(P,Q) = Q+\ol{D}(P,Q) \]
for predicates $p\c P\monoto X$ and $q\c Q\monoto Y$, where $\ol{D}$ is the canonical lifting of $D$. Indeed, since $D$ preserves monomorphisms in the second component, the object $\ol{D}(P,Q)$ is simply the pullback of $D(p,\id)$ and $D(\id,q)$. Letting $r_{P,Q}$ and $s_{P,Q}$ denote the projections, and using that pullbacks commute with coproducts in extensive categories (\Cref{rem:extensive}\ref{rem:extensive:pullbacks-commute-with-coproducts}), the following is a pullback:
\[
\begin{tikzcd}
Q+\ol{D}(P,Q) \pullbackangle{-45} \ar{r}{\id+s_{P,Q}} \ar[tail]{d}[swap]{q+r_{P,Q}} & Q+D(P,Q) \ar[tail]{d}{q+D(\id,q)} \\
Y+D(X,Y) \ar{r}{\id+D(p,\id)} & Y+D(P,Y) 
\end{tikzcd}
\]
Note that the functor $B$ preserves monomorphisms in the second component because $D$ does and monomorphisms are stable under coproducts by extensivity of $\C$ (\Cref{rem:extensive}\ref{rem:extensive:pullbacks-commute-with-coproducts}).
 Hence the above diagram is precisely the pullback defining $\ol{B}(P,Q)$.
\end{rem}

We make use of the following lemma:

\begin{lem}\label{lem:gamma-n-comp}
  The following diagram commutes for every $m,n\in \Nat$:
  \begin{equation}\label{eq:gamma-m-n}
    \begin{tikzcd}
      \mS \ar{r}{\gamma^{(m+n)}} \ar{d}[swap]{\gamma^{(m)}} & \mS+D(\mS,\mS) \\
      \mS+D(\mS,\mS) \ar{r}{\gamma^{(n)}} & \mS+D(\mS,\mS)+D(\mS,\mS) \ar{u}[swap]{\id+\nabla}
    \end{tikzcd}
  \end{equation}
\end{lem}

\begin{proof}
By induction on $m$. The case $m=0$ is clear because $\gamma^{(0)}=\inl$. The induction step $m\to m+1$ follows from the commutative diagram below. Its upper part commutes by induction. To see that the lower part commutes, it suffices to consider its precomposition with the three coproduct injections of $\mS+D(\mS,\mS)+D(\mS,\mS)$. The composite with the first injection yields $\gamma^{(n)}\c \mS\to\mS+D(\mS,\mS)$, and the composite with both the second and third injection yields $\inr\colon D(\mS,\mS)\to \mS+D(\mS,\mS)$.
\[
\begin{tikzcd}[scale cd =.6, column sep=2em]
\mS \ar{r}{\gamma} \ar{d}[swap]{\gamma} & \mS+D(\mS,\mS) \ar{r}{\gamma^{(m+n)}+\id} & \mS+D(\mS,\mS)+D(\mS,\mS) \ar{r}{\id+\nabla} & \mS+D(\mS,\mS) \\
\mS+D(\mS,\mS) \ar[equals]{ur} \ar{d}[swap]{\gamma^{(m)}+\id} & & \mS+D(\mS,\mS)+D(\mS,\mS)+D(\mS,\mS) \ar{u}[swap]{\id+\nabla+\id} & \\
\mS+D(\mS,\mS)+D(\mS,\mS) \ar{d}[swap]{\id+\nabla} \ar{urr}{\gamma^{(n)}+\id+\id} & & & \\
\mS+D(\mS,\mS) \ar{rrr}{\gamma^{(n)}+\id} & & & \mS+D(\mS,\mS)+D(\mS,\mS) \ar{uuu}[swap]{\id+\nabla}
\end{tikzcd}
\]
\end{proof}

\begin{notation}
\begin{enumerate}
\item We denote the pullbacks defining the predicates $\Downarrow_n$ and $\sqdown\wedge{\Downarrow_n}$ as follows:
\[
\begin{tikzcd}
\Downarrow_n \pullbackangle{-45} \ar{r}{\ol{\gamma}^{(n)}} \ar[tail]{d}[swap]{m_n} & D(\mS,\mS) \ar[tail]{d}{\inr} \\
\mS \ar{r}{\gamma^{(n)}} & \mS+D(\mS,\mS)
\end{tikzcd}
\qquad
\begin{tikzcd}
\sqdown\,\wedge\Downarrow_n \pullbackangle{-45} \ar[tail]{r}{\ol{m}_n} \ar[tail]{d}[swap]{\wt{m}_n} & \sqdown \ar[tail]{d}{m} \\
\Downarrow_n \ar[tail]{r}{m_n} & \mS
\end{tikzcd}
\qquad
\]
\item Let $\tilde{\gamma}\colon \sqdown\to \sqdown + \ol{D}(\sqdown,\sqdown)$ witness that $\sqdown$ is a logical predicate. For $n\in\Nat$ we define
\[ \tilde{\gamma}^{(n)}\colon \sqdown\to \sqdown+\ol{D}(\sqdown,\sqdown)  \]
recursively by
\begin{align*}
\tilde{\gamma}^{(0)} &= (\,\sqdown \xto{\inl} \sqdown+\ol{D}(\sqdown,\sqdown\,),\\
 \tilde{\gamma}^{(n+1)}& =(\, \sqdown\xto{\tilde{\gamma}}\sqdown+\ol{D}(\sqdown,\sqdown) \xto{\tilde{\gamma}^{(n)}+\id} \sqdown+\ol{D}(\sqdown,\sqdown)+\ol{D}(\sqdown,\sqdown) \xto{\id+\nabla} \sqdown+\ol{D}(\sqdown,\sqdown)  \,).
\end{align*} 
One easily verifies by induction on $n$ that $\wt{\gamma}^{(n)}$ satisfies the following commutative diagram:
\begin{equation}\label{eq:gamma-n-diag}
\begin{tikzcd}
\sqdown \ar{r}{\wt{\gamma}^{(n)}} \ar[tail]{d}[swap]{m} & \sqdown+\ol{D}(\sqdown,\sqdown) \ar[tail]{d}{m+r_{\sqdown,\sqdown}} \\
\mS \ar{r}{\gamma^{(n)}} & \mS+D(\mS,\mS)
\end{tikzcd}
\end{equation}
\item Finally, by extensivity of $\C$ the lower rectangle in the diagram on the left below is a pullback, and commutativity of the outside follows from the diagram on the right. Hence we obtain a unique $\dot{\gamma}^{(n)}$ making the diagram on the left commute. 
\[
\begin{tikzcd}[column sep=2em, scale cd=.9]
\sqdown\wedge{\Downarrow_n} \ar[dashed]{dr}{\dot{\gamma}^{(n)}} \ar{dd}[swap]{\wt{m}_n} \ar{rr}{\ol{m}_n} & & \sqdown \ar{d}{\wt{\gamma}^{(n)}} \\
& \ol{D}(\sqdown,\sqdown) \pullbackangle{-45} \ar{r}{\inr} \ar{d}[swap]{r_{\sqdown,\sqdown}} & \sqdown+\ol{D}(\sqdown,\sqdown) \ar{d}{m+r_{\sqdown,\sqdown}} \\
\Downarrow_n \ar{r}{\ol{\gamma}^{(n)}} & D(\mS,\mS) \ar{r}{\inr} & \mS+D(\mS,\mS) 
\end{tikzcd}
\quad
\begin{tikzcd}[column sep=2em, scale cd=.9]
\sqdown\wedge{\Downarrow_n} \ar{dd}[swap]{\wt{m}_n} \ar{rr}{\ol{m}_n} & & \sqdown \ar{d}{\wt{\gamma}^{(n)}} \ar{dl}[swap]{m} \\
&  \mS \ar{dr}[swap]{\gamma^{(n)}} \ar[phantom]{r}[description]{\text{\eqref{eq:gamma-n-diag}}} & \sqdown+\ol{D}(\sqdown,\sqdown) \ar{d}{m+r_{\sqdown,\sqdown}} \\
\Downarrow_n \ar{ur}{m_n} \ar{r}{\ol{\gamma}^{(n)}} & D(\mS,\mS) \ar{r}{\inr} & \mS+D(\mS,\mS) 
\end{tikzcd}
\]
\end{enumerate}
\end{notation}

\begin{proof}[Proof of \Cref{thm:strong-normalization}]
By \Cref{th:main2}, it suffices to prove ${\ol{\Sigma}(\sqdown)}\leq \, \iimg{\ini}{\Downarrow}$. The left-hand side of this inequality can be expressed as
\[ {\ol{\Sigma}(\sqdown)} = {\ol{\Sigma}(\sqdown\wedge {\Downarrow})}
= {\ol{\Sigma}(\sqdown\wedge (\vee_n \Downarrow_n))}
= {\ol{\Sigma}(\vee_n(\sqdown\wedge{\Downarrow_n}))} 
= {\vee_n \ol{\Sigma}(\sqdown\wedge{\Downarrow_n})}.
\]
We use \Cref{lem:preim-pres-countable-joins} in the third step, and \Cref{rem:smooth} and that $\ol{\Sigma}$ preserves directed unions in the fourth step. Consequently, we are left to prove
\begin{equation}\label{eq:proof-goal}{\ol{\Sigma}(\sqdown\wedge{\Downarrow_n})}\leq \, \iimg{\ini}{\Downarrow} \qquad \text{for every $n\in \Nat$}.
\end{equation} This requires two steps.

\begin{enumerate}
\item Let $g_n\c \Sigma(\sqdown\wedge{\Downarrow_n}) \to \mS+D(\mS,\mS)$ denote the following composite:
\[  
\begin{tikzcd}[scale cd=.75, row sep=20]
\Sigma(\sqdown\wedge{\Downarrow_n}) \ar{r}{\Sigma\wt{m}_n} & \Sigma(\Downarrow_n) \ar{r}{\Sigma{m_n}} & \Sigma(\mS) \ar{r}{\langle \id,\, \gamma^{(n)}\rangle} & \Sigma(\mS\times D(\mS,\mS)) \ar{dlll}[swap]{\rho_{\mS,\mS}} \\
\Sigmas(\mS+\mS) + D(\mS,\Sigmas(\mS+\mS)) \ar{rr}{\Sigmas\nabla+D(\id,\Sigmas\nabla)} && \Sigmas(\mS) + D(\mS,\Sigmas(\mS)) \ar{r}{\hatini+D(\id,\hatini)} & \mS+D(\mS,\mS).
\end{tikzcd}
\]
We claim that
\begin{equation}\label{eq:gn-fact}
\text{$g_n$ factorizes through $\Downarrow+\,D(\mS,\mS)\monoto\mS+D(\mS,\mS)$}.
\end{equation} 
To see this, consider the diagram in \Cref{fig:normalization}.
\begin{figure*}
\begin{equation}\label{diag:normalization}
\rotatebox{-90}{
\begin{tikzcd}[column sep=3em, ampersand
      replacement=\&, scale cd=.6, row sep=4em]
\Sigma(\sqdown \wedge{\Downarrow_n}) \ar{r}{\Sigma\wt{m}_n} \ar{dr}{\Sigma\langle \ol{m}_n,\,\wt{m}_n\rangle} \ar{ddd}[swap]{\Sigma\langle \ol{m}_n,\, \dot{\gamma}^{(n)} \rangle }  \& \Sigma(\Downarrow_n) \ar{r}{\Sigma m_n} \& \Sigma(\mS) \ar{rr}{\Sigma \langle \id,\,\gamma^{(n)}\rangle}  \& \& \Sigma(\mS\times (\mS+D(\mS,\mS))) \ar{d}{\rho_{\mS,\mS}}  \\
 \& \Sigma(\sqdown\times \Downarrow_n) \ar{u}[swap]{\Sigma\outr} \ar{r}{\Sigma(\id\times m_n)} \ar{d}{\Sigma(\id\times \ol{\gamma}^{(n)})} \& \Sigma(\sqdown\times \mS) \ar{u}[swap]{\Sigma\outr} \ar{d}{\Sigma(\id\times \gamma^{(n)})} \ar[phantom]{urr}[description]{(\ast)} \& \& \Sigmas(\mS+\mS)+D(\mS,\Sigmas(\mS+\mS)) \ar{d}{\Sigmas\nabla+D(\id,\Sigmas\nabla)} \ar{dddl}[swap]{\id+D(m,\id)} \\
\& \Sigma(\sqdown\times D(\mS,\mS))  \ar{r}{\Sigma(\id\times \inr)} \ar{d}{\Sigma(\id\times D(m,\id))} \& \Sigma(\sqdown\times(\mS+D(\mS,\mS))) \ar{uurr}{\Sigma(m\times\id)} \ar{d}{\Sigma(\id\times(\id+D(m,\id)))} \& \& \Sigmas(\mS)+D(\mS,\Sigmas(\mS)) \ar{d}{\hatini+D(\id,\hatini)} \\
\Sigma(\sqdown\times \ol{D}(\sqdown,\sqdown))  \ar{ur}{\Sigma(\id\times r_{\sqdown,\sqdown})} \ar{d}[swap]{\Sigma(\id\times s_{\sqdown,\sqdown})} \& \Sigma(\sqdown\times D(\sqdown,\mS))  \ar{r}{\Sigma(\id\times \inr)} \& \Sigma(\sqdown\times(\mS+D(\sqdown,\mS))) \ar{d}{\rho_{\sqdown,\mS}} \& \& \mS+D(\mS,\mS) \ar{ddd}{\id+D(m,\id)} \\
\Sigma(\sqdown\times D(\sqdown,\sqdown)) \ar{r}{\Sigma(\id\times \inr)} \ar[swap]{dd}{\rho^0_{\sqdown,\sqdown}} \ar{ur}{\Sigma(\id\times D(\id,m))} \& \Sigma(\sqdown\times (\sqdown+D(\sqdown,\sqdown)))  \ar{d}{\rho_{\sqdown,\sqdown}} \ar{ur}{\Sigma(\id\times(m+D(\id,m)))} \& \Sigmas(\sqdown+\mS) + D(\sqdown,\Sigmas(\sqdown+\mS))  \ar{r}[yshift=1em]{\Sigmas(m+\id)+D(\id,\Sigmas(m+\id))} \& \Sigmas(\mS+\mS)+D(\sqdown,\Sigmas(\mS+\mS)) \ar{d}{\Sigmas\nabla+D(\id,\Sigmas\nabla)} \& \\
\& \Sigmas(\sqdown+\sqdown)+D(\sqdown,\Sigmas(\sqdown+\sqdown)) \ar{r}{\Sigmas\nabla+D(\id,\Sigmas\nabla)} \ar{ur}[description]{\Sigmas(\id+m)+D(\id,\Sigmas(\id+m))} \& \Sigmas(\sqdown)+D(\sqdown,\Sigmas(\sqdown)) \ar{r}{\Sigmas m + D(\id,\Sigmas m)} \& \Sigmas(\mS)+D(\sqdown,\Sigmas(\mS)) \ar{dr}{\hatini+D(\id,\hatini)}  \& \\
\sqdown+\sqdown+D(\sqdown,\Sigmas(\sqdown+\sqdown)) \ar{ur}{\eta_{\sqdown+\sqdown}+\id} \ar{rr}{\nabla+D(\id,\Sigmas\nabla)} \& \& \sqdown+D(\sqdown,\Sigmas(\sqdown)) \ar{u}{\eta_{\sqdown}+\id} \ar{r}{m+D(\id,\Sigmas m)} \& \mS+D(\sqdown,\Sigmas(\mS)) \ar{u}{\eta_{\mS}+\id} \ar{r}{\id+D(m,\hatini)} \& \mS+D(\sqdown,\mS)  
\end{tikzcd}
}
\end{equation}
\caption{Diagram for proof of \Cref{thm:strong-normalization}}\label{fig:normalization}
\end{figure*}
All inner cells except ($\ast$) commute by (di)naturality or by definition, and ($\ast$) commutes when precomposed with $\Sigma(\id\times m_n)\comp \Sigma\langle \ol{m}_n, \wt{m}_n\rangle$. It follows that the outside of the diagram commutes, which shows that the composite 
\[\Sigma(\sqdown\wedge{\Downarrow_n}) \xto{g_n} \mS+D(\mS,\mS)\xto{\id+D(m,\id)} \mS+D(\sqdown,\mS)\] factorizes through 
\[\sqdown+D(\sqdown,\Sigmas(\sqdown)) \xto{m+D(\id,\ini\comp \Sigmas m)} \mS+D(\sqdown,\mS).\]
Let $(P,p_1,p_2)$ be the pullback of $D(m,\id)$ and $D(\id,\ini\comp \Sigmas m)$. Then, since pullbacks and coproducts in $\C$ commute (\Cref{rem:extensive}\ref{rem:extensive:pullbacks-commute-with-coproducts}), the following is a pullback:
\[
\begin{tikzcd}[column sep=50]
\sqdown+P \pullbackangle{-45} \ar{r}{m+p_2} \ar[tail]{d}[swap]{\id+p_1} & \mS+D(\mS,\mS) \ar[tail]{d}{\id+D(m,\id)} \\
\sqdown+D(\sqdown,\Sigmas(\sqdown)) \ar{r}{m+D(\id,\ini\comp \Sigmas m)} & \mS+D(\sqdown,\mS) 
\end{tikzcd}
\]
It follows that $g_n$ factorizes through $m+p_2\colon \sqdown+P\to \mS+D(\mS,\mS)$. In particular, since $\sqdown\leq\, \Downarrow$ in $\Pred[\mS]{\C}$, we can conclude that \eqref{eq:gn-fact} holds.
\item The graph of $g_n$, regarded as a relation between $\Sigma(\mS)$ and $\mS+D(\mS,\mS)$, satisfies
\[
\gra(g_n) \leq \Sigma(\mS)\times (\Downarrow+\,D(\mS,\mS)) \qqand
\gra(g_n) \leq \vee_k\, \gra(\gamma^{(k)}\comp \ini).
\]
Indeed, the first inequality follows from \eqref{eq:gn-fact}, and the second one from our assumption that $\rho$ respects weak transitions. Consequently,
\begin{equation}\label{eq:gra-gn-ineq}
\gra(g_n) \leq \Sigma(\mS)\times (\Downarrow+\,D(\mS,\mS)) \,\wedge \, \vee_k\, \gra(\gamma^{(k)}\comp \ini).
\end{equation}
The relation $\gra(g_n)$ has domain $\Sigma(\sqdown\wedge{\Downarrow_n})$. Therefore, to establish \eqref{eq:proof-goal} it suffices to show that the domain of the relation on the right-hand side of \eqref{eq:gra-gn-ineq} is contained in $\iimg{\ini}{\Downarrow}$. We first note that
\begin{align*}
& \Sigma(\mS)\times (\Downarrow+\, D(\mS,\mS)) \,\wedge \, \vee_k\, \gra(\gamma^{(k)}\comp \ini) \\
=\; & \Sigma(\mS)\times (\vee_l\Downarrow_l+\, D(\mS,\mS)) \,\wedge \, \vee_k\, \gra(\gamma^{(k)}\comp \ini) \\
=\; & \vee_l (\Sigma(\mS)\times (\Downarrow_l+\, D(\mS,\mS)))  \,\wedge \, \vee_k\, \gra(\gamma^{(k)}\comp \ini) \\
=\; & \vee_{k,l} \big(\Sigma(\mS)\times (\Downarrow_l+\, D(\mS,\mS)) \,\wedge \, \gra(\gamma^{(k)}\comp \ini)\big).
\end{align*}
We use \Cref{lem:join-vs-prod-coprod} in the third step and \Cref{lem:preim-pres-countable-joins} in the fourth step. Hence, it suffices to prove
\begin{equation}\label{eq:proof-goal-2}
{\dom(\Sigma(\mS)\times (\Downarrow_l+\, D(\mS,\mS))  \,\wedge \, \gra(\gamma^{(k)}\comp \ini))} \leq\,\iimg{\ini}{\Downarrow_{k+l}} \qquad\text{for every $k,l\in \Nat$}.
\end{equation}
To this end, consider first the following diagram:
\[
\begin{tikzcd}[scale cd=.7]
\Sigma(\mS)\times (\Downarrow_l+\, D(\mS,\mS))  \,\wedge \, \gra(\gamma^{(k)}\comp \ini)  \ar[tail]{r} \ar[tail]{d} \ar{dr}{\outr} & \gra(\gamma^{(k)}\comp \ini) \ar{rr}{\outl} \ar{d}{\outr} \ar[phantom]{drr}[description]{\text{\eqref{eq:graph-outl-outr}}} & & \Sigma(\mS) \ar{d}{\ini} \\
\Sigma(\mS)\times (\Downarrow_l+\,D(\mS,\mS)) \ar{dd}[swap]{\outr} \ar{r}{\outr} & \mS+D(\mS,\mS)  \ar{dr}{\gamma^{(l)}+\id} \ar[bend right=2em]{ddrr}{[\gamma^{(l)},\inr]}  & & \mS \ar[phantom]{dl}[description, yshift=-1em]{\text{\eqref{eq:gamma-m-n}}} \ar{ll}[swap]{\gamma^{(k)}} \ar{dd}{\gamma^{(k+l)}} \\
& & \mS+D(\mS,\mS)+D(\mS,\mS) \ar{dr}{\id+\nabla}  & \\
\Downarrow_l + D(\mS,\mS) \ar{uur}{m_l+\id} \ar{r}{[\ol{\gamma}^{(l)},\id]} & D(\mS,\mS) \ar{rr}{\inr} & & \mS+D(\mS,\mS)
\end{tikzcd}
\]
Except for the cell using \eqref{eq:graph-outl-outr} and \eqref{eq:gamma-m-n}, all cells commute either trivially or by definition.
 Hence, the outside of the diagram below commutes, and so the universal property of the pullback $\Downarrow_{k+l}$ yields the unique dashed morphism making the inner cells in the diagram below commute. The upper cell witnesses that \eqref{eq:proof-goal-2} holds.
\[
\begin{tikzcd}[scale cd=.7]
\Sigma(\mS)\times (\Downarrow_l+D(\mS,\mS))  \,\wedge \, \gra(\gamma^{(k)}\comp \ini) \ar[tail]{r} \ar[tail]{d} \ar[dashed]{dr} & \gra(\gamma^{(k)}\comp \ini) \ar{rr}{\outl} & & \Sigma(\mS) \ar{d}{\ini} \\
\Sigma(\mS)\times (\Downarrow_l+D(\mS,\mS)) \ar{d}[swap]{\outr} & \Downarrow_{k+l} \ar[tail]{rr}{m_{k+l}} \ar{d}[swap]{\ol{\gamma}^{(k+l)}} \pullbackangle{-45} & & \ar{d}{\gamma^{(k+l)}} \mS \\
\Downarrow_l + D(\mS,\mS) \ar{r}{[\ol{\gamma}^{(l)},\id]} & D(\mS,\mS) \ar{rr}{\inr} & & \mS+D(\mS,\mS)
\end{tikzcd}
\]
\end{enumerate}
\end{proof}

\section*{Proof of \Cref{thm:ind-up-to-blackquare}}
We prove this theorem in the more general context of $\lambda$-laws, see \Cref{app:lambdashenanigans}.

\section{Properties of Predicates}\label{sec:predicates}
We collect some useful facts about predicates over (countably) extensive categories that we use in various proofs. In this section let $\C$ be complete, cocomplete, and well-powered.

\begin{rem}\label{rem:extensive}
We shall use the following facts about extensive categories:
\begin{enumerate}
\item\label{rem:extensive:pullbacks-commute-with-coproducts} Pullbacks commute with coproducts: given pullback squares as shown on the left below, where $i=1,2$, the square on the right is a pullback~\cite[Prop.~2.6]{amv10}.
\[
\begin{tikzcd}
P_i \pullbackangle{-45} \ar{r}{f_i} \ar{d}[swap]{g_i}  & A_i \ar{d}{h_i} \\
B_i \ar{r}{k_i} & C_i 
\end{tikzcd}
\qquad
\begin{tikzcd}
P_1+P_2 \pullbackangle{-45} \ar{r}{f_1+f_2} \ar{d}[swap]{g_1+g_2}  & A_1+A_2 \ar{d}{h_1+
h_2} \\
B_1+B_2 \ar{r}{k_1+k_2} & C_1+C_2 
\end{tikzcd}
\]
In particular, coproducts of monomorphisms are monomorphisms.
\item\label{rem:extensive:pullbacks-stable-under-copairing} Pullbacks are stable under copairing: given pullback squares as shown on the left below, where $i=1,2$, the square on the right is a pullback~\cite[Lem.~2.7]{amv10}.
\[
\begin{tikzcd}
P_i \pullbackangle{-45} \ar{r}{f_i} \ar{d}[swap]{g_i}  & A_i \ar{d}{h_i} \\
B \ar{r}{k} & C 
\end{tikzcd}
\qquad
\begin{tikzcd}
P_1+P_2 \pullbackangle{-45} \ar{r}{f_1+f_2} \ar{d}[swap]{[g_1,g_2]}  & A_1+A_2 \ar{d}{[h_1,
h_2]} \\
B \ar{r}{k} & C 
\end{tikzcd}
\]
\item\label{rem:extensive:distributive} Every extensive category with finite products is distributive, that is, the functor $X\times (-)\colon \C\to \C$ preserves finite coproducts~\cite[Prop.~4.5]{cbl93}.
\item\label{rem:extensive:injections-monic} In every distributive category, in particular in every extensive category with finite products, coproduct injections are monic~\cite[Prop.~3.3]{cbl93}
\end{enumerate}
Analogous statements hold for countably extensive categories and countable coproducts.
\end{rem}

\begin{rem}\label{rem:smooth}
We say that monomorphisms in $\C$ are \emph{smooth} if for every directed diagram whose connecting morphisms are monic, the colimit injections are monic, and moreover for every cocone over that diagram the induced mediating morphism is monic. This property holds, e.g.,\ in locally finitely presentable categories, such as sets, posets, or presheaves~\cite[Prop. 1.62]{adamek_rosicky_1994}. Smoothness of monomorphisms implies that a directed join of predicates in $\Pred[X]{\C}$, $X\in \C$, is formed as the directed colimit of the underlying diagram in $\C$.
\end{rem}

\begin{lem}\label{lem:preim-pres-countable-joins}
Suppose that $\C$ is countably extensive and strong epimorphisms are pullback-stable.
\begin{enumerate}
\item The map $f^\star\colon \Pred[Y]{\C}\to \Pred[X]{\C}$ preserves countable joins for every $f\colon X\to Y$.
\item Finite meets distributive over countable joins of predicates:
\[ P\wedge (\vee_n P_n) = \vee_n (P\wedge P_n) \qquad \text{for all predicates $P\monoto Y$ and $P_n\monoto Y$ ($n\in \Nat$)}.
\]
\end{enumerate}
\end{lem}

\begin{proof}
\begin{enumerate}
\item Let $f\colon X\to Y$ be a morphism and let $p_n\colon P_n\monoto Y$ ($n\in \Nat$) be predicates. Form the preimage of $p_n$ under $f$ as shown in the diagram on the left below. Then, since $\C$ is countably extensive, the diagram on the right is a pullback (\Cref{rem:extensive}\ref{rem:extensive:pullbacks-stable-under-copairing}).
\begin{equation}\label{eq:pullbacks}
\begin{tikzcd}
f^\star P_n \ar{r}{\ol{f}_n} \ar[tail]{d}[swap]{\ol{p}_n} \pullbackangle{-45} & P_n \ar[tail]{d}{p_n} \\
X \ar{r}{f} & Y
\end{tikzcd}
\qquad
\begin{tikzcd}
\coprod_{n} f^\star P_n \ar{r}{\coprod_n \ol{f}_n} \ar[tail]{d}[swap]{[\ol{p}_n]_n} \pullbackangle{-45} & \coprod_n P_n \ar[tail]{d}{[p_n]_n} \\
X \ar{r}{f} & Y
\end{tikzcd}
\end{equation}
The predicates $\vee_n P_n$ and $\vee_n f^\star P_n$ are obtained via the following (strong epi, mono)-factorizations:
\[
\begin{tikzcd} \coprod_n P_n \ar[shiftarr={yshift=1.5em}]{rr}{[p_n]_n} \ar[two heads]{r}{e} & \vee_n P_n \ar[tail]{r}{m} & Y \end{tikzcd}
\qquad
\begin{tikzcd} \coprod_n f^\star P_n \ar[shiftarr={yshift=1.5em}]{rr}{[\ol{p}_n]_n} \ar[two heads]{r}{\ol{e}} & \vee_n f^\star P_n \ar[tail]{r}{\ol{m}} & X 
\end{tikzcd}
\]
Diagonal fill-in yields a unique $\ol{f}$ making the two rectangular cells in the diagram below commute.
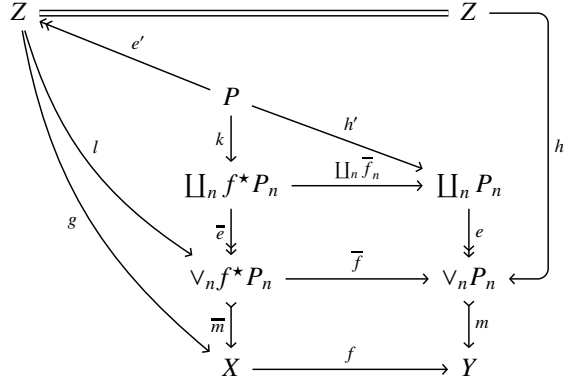
\[
\begin{tikzcd}[column sep=5em]
Z \ar[bend right=2em]{ddddr}[swap]{g} \ar[bend right=2em]{dddr}{l} \ar[equals]{rr} & & Z \ar[shiftarr = {xshift=3em}]{ddd}{h} \\
& P \ar[two heads]{ul}[swap]{e'} \ar{dr}{h'} \ar{d}[swap]{k} & \\
& \coprod_n f^\star P_n \ar{r}{\coprod_n \ol{f}_n} \ar[two heads]{d}[swap]{\ol{e}} & \coprod_n P_n \ar[two heads]{d}{e} \\
& \vee_n f^\star P_n \ar[tail]{d}[swap]{\ol{m}}  \ar{r}{\ol{f}} & \vee_n P_n \ar[tail]{d}{m} \\
& X \ar{r}{f} & Y 
\end{tikzcd}
\]
To prove $f^{\star}[\vee_n P_n] = \vee_n f^\star P_n$, we show that the lower rectangle is a pullback. Thus suppose that $g$ and $h$ are morphisms such that $f\comp g = m\comp h$; our task is to construct a morphism $l$ such that 
\begin{equation}\label{eq:l-props}
g=\ol{m}\comp l \qand h=\ol{f}\comp l.
\end{equation}
(Note that $l$ is unique because $\ol{m}$ is monic.) To this end, form the pullback $(P,e',h')$ of $e$ and $h$. The morphism $e'$ is strongly epic because $e$ is strongly epic and strong epimorphisms are pullback-stable. The universal property of the pullback on the right in \eqref{eq:pullbacks} yields a unique $k$ such that 
\[ g\comp e' = \ol{m}\comp \ol{e}\comp k \qand h'=(\coprod_n \ol{f}_n)\comp k. \]
Diagonal fill-in yields a unique $l$ such that
\[ g=\ol{m}\comp l \qand \ol{e}\comp k = l\comp e'. \]
In particular, $l$ satisfies the first equality of \eqref{eq:l-props}. The second equality holds because, by the above commutative diagram, it holds when precomposed with the epimorphism $e$.
\item That finite meets distributive over countable joins means precisely that for every predicate $p\colon P\monoto Y$ the map $p^\star$ preserves countable joins, which follows from part (1). \qedhere
\end{enumerate}
\end{proof}

\begin{rem}\label{rem:strong-epis-stable-under-products}
If strong epimorphisms in $\C$ are pullback-stable then they also stable under products, that is, for any two strong epimorphisms $e,e'$ their product $e\times e'$ is a strong epimorphism. Since $e\times e' = (e\times \id)\comp (\id\times e')$, we only need to show that $e\times \id$ and $\id\times e'$ are strong epimorphisms, and by symmetry it suffices to consider the former. To this end, just note that the following is a pullback:
\[
\begin{tikzcd}
A\times C \pullbackangle{-45} \ar{r}{e\times \id} \ar{d}[swap]{\outl} & B\times C \ar{d}{\outl} \\
A \ar{r}{e} & B
\end{tikzcd}
\]
\end{rem}

\begin{lem}\label{lem:join-vs-prod-coprod}
Suppose that $\C$ is countably extensive and that strong epimorphisms are pullback-stable. Then for every $X,Y\in \C$ and predicates $P_n\monoto Y$ ($n\in\Nat$),
\begin{enumerate}
\item\label{lem:join-vs-prod-coprod:prod} $X\times (\vee_n P_n) = \vee_n (X\times P_n)$ in $\Pred[X\times Y]{\C}$;
\item\label{lem:join-vs-prod-coprod:coprod} $X+(\vee_n P_n) = \vee_n (X+P_n)_n$ in $\Pred[X+Y]{\C}$.  
\end{enumerate}
\end{lem}

\begin{proof} Let $p_n\colon P_n\monoto Y$ ($n\in\Nat$) be predicates.
\begin{enumerate}
\item Consider the commutative diagram below, where $e$ and $m$ are the (strong epi, mono)-factori\-za\-tion of $[p_n]_n$, and $e'$ and $m'$ are the (strong epi-mono)-factorization of  $[\id\times p_n]_n$, and the isomorphism witnesses that finite products distribute over countable coproducts (\Cref{rem:extensive}\ref{rem:extensive:distributive}),
\[
\begin{tikzcd}
\coprod_n (X\times P_n) \ar{rr}{\cong} \ar[two heads]{d}[swap]{e'} & & X\times (\coprod_n P_n) \ar[two heads]{d}{\id\times e} \\
\vee_n (X\times P_n) \ar[tail]{dr}[swap]{m'}  & & X\times(\vee_n P_n) \ar[tail]{dl}{\id\times m} \\
& X\times Y & 
\end{tikzcd}
\]
Note that $\id\times e$ is a strong epimorphism (\Cref{rem:strong-epis-stable-under-products}). Hence the uniqueness of (strong epi, mono)-factorizations shows that $X\times (\vee_n P_n) = \vee_n (X\times P_n)$ as subobjects of $X\times Y$.
\item Consider the commutative diagram below, where $e$ and $m$ are (strong epi, mono)-factorization of $[p_n]_n$, and the isomorphism witnesses commutativity of coproducts.
\[
\begin{tikzcd}
\coprod_n (X+P_n) \ar{r}{\cong} \ar{d}[swap]{[\id+p_n]_n} & \coprod_n X + \coprod_n P_n \ar[two heads]{d}{[\id]_n+e} \\
X+Y & X+(\vee_n P_n) \ar[tail]{l}[swap]{\id+m} 
\end{tikzcd}
\]
\end{enumerate}
Since $\id+m$ is monic (\Cref{rem:extensive}\ref{rem:extensive:pullbacks-commute-with-coproducts}), we see that $[\id]_n+e$ and $\id+m$ form the (strong epi, mono)-factorization of $[\id+p_n]_n$, whence $X+(\vee_n P_n) = \vee_n (X+P_n)_n$ as subobjects of $X+Y$.\qedhere
\end{proof}

\begin{defn}
A \emph{relation} between objects $A$ and $B$ of $\C$ is a predicate over $A\times B$, that is, a subobject $r\colon R\monoto A\times B$. We write $r=\langle \outl_R,\outr_R\rangle$ and usually omit the subscript $R$. The \emph{domain} of $R$ is the image of the morphism $\outl$:
\[ 
\begin{tikzcd}
	R \ar[shiftarr = {yshift=1.5em}]{rr}{\outl_R} \ar[two heads]{r}{e_{\dom(R)}} & \dom(R) \ar[tail]{r}{m_{\dom(R)}} & A
\end{tikzcd}
\]
The \emph{graph} of a morphism $f\colon A\to B$ is the relation between $A$ and $B$ given by the image of $\langle \id,f\rangle$:
\[
\begin{tikzcd}
A \ar[shiftarr = {yshift=1.5em}]{rr}{\langle \id,f\rangle} \ar[two heads]{r}{e_{\gra(f)}} & \gra(f) \ar[tail]{r}{m_{\gra(f)}} & A\times B
\end{tikzcd}
\]
\end{defn}

\begin{rem}\label{rem:graph-of-morphism}
Note that $\gra(\f)$ has domain $A$, and that we have the commutative triangle
\begin{equation}\label{eq:graph-outl-outr}
\begin{tikzcd}
& \gra(f) \ar[swap]{dl}{\outl} \ar{dr}{\outr} & \\
A \ar{rr}{f} & & B
\end{tikzcd}
\end{equation}
\end{rem}

\takeout{

\section{More on Predicate Lifting}

\begin{lem}
  \label{lem:can-comp}
Given $B\c {\C^\opp\times \C \to \C}$ and $H\c\C\to\C$, if $H$ preserves weak pullbacks and strong 
epimorphisms then the composition of canonical liftings $\ol H\,\ol B$ is the 
canonical lifting $\ol{H B}$.
\end{lem}
\begin{proof}
  Consider two objects $p \colon P \pred{} X$ and $q \colon Q \pred{} Y$ of $\Pred\C$. %
  First we compute $\overline B$ as below on the left, and then we factorise  $H(\overline B(p,q))$ to obtain $\overline H( \overline B(p,q))$ on the right: 
  \[
  \begin{tikzcd}[column sep=tiny, row sep=.2ex]
  	&[1em]& {T_{P,Q}}
  	\pullbackangle{-45}
  	&&&& {B(P,Q)} \\
  	\\
  	{\overline{B}(P,Q)} \\
  	\\
  	&& {B(X,Y)} &&&& {B(P,Y)}
  	\arrow["e"', two heads, from=1-3, to=3-1]
  	\arrow["{\overline B(p,q)}"', tail, pos=.8, from=3-1, to=5-3]
  	\arrow["a", from=1-3, to=1-7]
  	\arrow[from=1-3, to=5-3]
  	\arrow["{B(P,q)}", from=1-7, to=5-7]
  	\arrow["{B(p,Y)}", from=5-3, to=5-7]
  \end{tikzcd}
  \qquad
\begin{tikzcd}[column sep=1em, row sep=4.5ex]
	H(\overline B(P,Q)) \ar[rr,"{H(\overline B(p,q))}"] \ar[dr,two heads] & &  HB(X,Y) \\
	&  \overline H (\overline B(P,Q))  \ar[ur, tail, "{\overline H(\overline B(p,q))}"']
\end{tikzcd}
  \]
  The lifting of the composite $HB$ is instead computed in the following diagram:
  \[
   \begin{tikzcd}[column sep=tiny, row sep=.3ex]
  	&[1em]& {T_{P,Q}'}
  	\pullbackangle{-45}
  	&&&& {HB(P,Q)} \\
  	\\
  	{\overline{HB}(P,Q)} \\
  	\\
  	&& {HB(X,Y)} &&&& {HB(P,Y)}
  	\arrow[two heads, from=1-3, to=3-1]
  	\arrow["{\overline{HB}(p,q)}"', tail, pos=.8, from=3-1, to=5-3]
  	\arrow[from=1-3, to=1-7]
  	\arrow[from=1-3, to=5-3]
  	\arrow["{HB(P,q)}", from=1-7, to=5-7]
  	\arrow["{HB(p,Y)}", from=5-3, to=5-7]
  \end{tikzcd}
  \]
  The following diagram illustrates the whole situation.
  \[
  \begin{tikzcd}[row sep=0em]
  	H(T_{P,Q}) \ar[rrrd, bend left=10, "H(a)"] \ar[ddd,two heads,"H(e)"'] \ar[rrd,dashed,"k"] \\
  	& & T_{P,Q}' \ar[r] \ar[dl,two heads] \ar[dd] \pullbackangle{-45} & HB(P,Q) \ar[dd,"{HB(P,q)}"] \\
  	& \overline{HB}(P,Q) \ar[dr,tail,"{\overline{HB}(p,q)}"] \\[1em]
  	H(\overline B(P,Q)) \ar[rr,"{H(\overline B(p,q))}"]  \ar[dr,two heads]& & HB(X,Y) \ar[r,"{HB(p,Y)}"] & HB(P,Y) \\[.5em]
  	& \overline H( \overline B( P,Q)) \ar[ur, tail, "{\overline H (\overline B( p,q))}"']
  \end{tikzcd}
  \]
  First, $H(e)$ is a strong epimorphism because $e$ is and $H$ preserves them. Second, $k$ is the universal arrow provided by the pullback $T_{P,Q}'$. Because $H$ preserves weak pullbacks, the outer diagram (which is the image along $H$ of a pullback) is a weak pullback, which means that $k$ is a split epi, in particular it is a strong epi. Therefore we have that the outer diagram below commutes:
  \[
  \begin{tikzcd}
  	H(T_{P,Q}) 
  	  \ar[r,two heads] 
  	  \ar[d,two heads,"k"'] 
  	  & 
  	\overline H(\overline B(P,Q)) 
  	  \ar[d,dashed,"c"] 
  	  \ar[r,tail,"{\overline H(\overline B(p,q))}"] 
  	  &[2em] 
  	HB(X,Y) 
  	  \ar[d,"\id{}"] 
  	\\
  	T_{P,Q}' 
  	  \ar[r, two heads] 
  	  & 
  	\overline{HB}(P,Q) 
  	  \ar[r,tail,"{\overline {HB}(p,q)}"] 
  	  & 
  	HB(X,Y)
  \end{tikzcd}
  \]
  thus we have two strong epi-mono factorizations of the same morphism. Hence there is a (unique) fill-in isomorphism $c$ which identifies $\overline H(\overline B(P,Q))$ and $\overline{HB}(P,Q)$ as subobjects of $HB(X,Y)$.
  
  The fact that $\overline{HB}$ and $\overline H \, \overline B$ coincide on morphisms follows immediately from the definition of lifting both in the covariant and in the mixed-variance case.
\end{proof}

\begin{expl}
  The canonical lifting for $\Pow B$ %
  can be obtained by \Cref{lem:can-comp}, because $\Pow$ preserves weak pullbacks and, assuming the axiom of choice, it preserves all epimorphisms (because they all split). %
  Moreover, $\Pow$ preserves 
  monomorphisms, again because they all split, hence $\ol{\Pow}(p\c P \pred{} Y) = {\Pow p\c \Pow P \pred{} \Pow Y}$.
  Therefore,
  \[
    \ol{\Pow B}_\tau(P,Q) = \ol{\Pow}_\tau(\ol{B}_\tau(P,Q)) = \{Z \subseteq
    B_\tau(X,Y) \mid \forall z \in Z.\, z\in\ol{B}_\tau(P,Q)\}.
  \]
\end{expl}

\begin{prop}[{\cite[Lemma 6.2.8 (ii)]{DBLP:books/cu/J2016}}]
  \label{prop:subalg}
  Let $F \c \C \to \C$ be a functor with canonical predicate lifting $\ol{F} \c
  \Pred{\C} \to \Pred{\C}$.  A predicate $p \c P \pred{} X$ is an
  $\ol{F}$-invariant for $a \c FX \to X$ if and only if $P$ extends to
  an $F$-algebra, such that $p$ is an $F$-algebra morphism:
  \[
    \begin{tikzcd}
      F P \rar[dashed]{}
      \ar[d,"Fp"']
      & P
      \ar[tail]{d}{p}
      \\
      F X
      \ar{r}{a}
      & X
    \end{tikzcd}
  \]
\end{prop}

The proposition below
similarly generalizes the correspondence between coalgebraic invariants for
canonical liftings and \emph{subcoalgebras}~\cite[Th. 6.2.8]{DBLP:books/cu/J2016}.

\begin{prop}
  Let $B\c \C^\opp\times \C\to \C$ be a bifunctor with canonical lifting
  $\ol{B}$ and let $c \c Y \to B(X,Y)$ be a $B$-coalgebra. Given a predicate $s\c S\pred{} X$,
  and assuming that $B$ preserves monomorphisms in the covariant part, then a predicate $q \c Q
  \pred{} Y$ is an $S$-relative invariant for $c$ if and only $q$ extends to a $B(S,\argument)$-subcoalgebra
  to $B(s,Y)\comp c\c Y\to B(S,Y)$, i.e.:
  \[
    \begin{tikzcd}
      Q
      \ar[dashed]{rr}
      \ar[tail,d,"q"']
      & & B(S,Q)
      \ar[tail]{d}{B(S,q)}
      \\
      Y \rar{c}
      & B(X,Y)
      \rar{B(s,Y)}
      & B(S,Y)
    \end{tikzcd}
  \]
\end{prop}

\begin{proof}
Recall that the mono-preservation condition entails that $\ol{B}(S,Q)$
is a pullback of $B(S,q)$ along $B(s,Y)$. Denote by $R=\iimg{c}{\ol{B}(S,Q)}$ be the 
pullback of $\ol B(s,q)$ along $c$. The following diagram will illustrate the argument:
  \[
    \begin{tikzcd}[column sep=5ex, row sep=4ex]
    Q
    \ar[tail, bend right, swap]{ddr}{q}
    \ar[dashed, bend left=13,drr,"v", pos=.7]
    \ar[dashed, bend left=15, drrr, "w", pos=.8]
    \ar[dashed,tail]{dr}{u}
    \\[-2ex]
    &[1em] R
    \rar[dashed]{j}
    \ar[tail, swap]{d}{r}
    \pullbackangle{-45}
    &[1em] \ol{B}(S,Q)
    \rar[dashed]{k}
    \ar[tail, swap]{d}{\ol{B}(s,q)}
    \pullbackangle{-45}
    & B(S,Q)
    \ar[tail]{d}{B(S,q)}
    \\
    & Y
    \ar{r}{c}
    & B(X,Y)
    \ar{r}{B(s,Y)}
    & B(S,Y)
  \end{tikzcd}
\]
First, suppose that $Q \leq \iimg{c}{\ol{B}(S,Q)}$, in other words
assume the existence of a necessarily unique $u \c Q \pred{} R$ such that $r
\comp u = q$. The requisite subcoalgebra structure is given by $k \comp j \comp
u$. For the opposite direction, suppose that there exists necessarily unique 
$w\c Q \to B(S,Q)$, such that $B(S,q) \comp w = B(s,Y) \comp c \comp q$. The
universal property of $\ol{B}(S,Q)$ induces unique $v$ and, subsequently, the
universal property of $R$ induces unique $u \c Q \pred{} R$ such that $n
\comp r = q$, completing the case.
\end{proof}
}

\section{Predicate Liftings of Free Monads}\label{sec:free-monad-lift}
Let $\Ar{\C}$ be the arrow category of $\C$. Its objects are triples $(X,P,m_P)$ consisting of two objects $X,P\in \C$ and a morphism $m_P\colon P\to X$; we usually write $(X,P)$ for $(X,P,m_P)$. A morphism from $(X,P)$ to $(Y,Q)$ is a pair of morphisms $h=(h_0\colon X\to Y, h_1\colon P\to Q)$ such that $h_0\comp m_P=m_Q\comp h_1$. Predicates form a full reflective subcategory $\Pred{\C}\monoto \Ar{\C}$; the reflector $(-)^\dag$ sends $(X,P)\in \Ar{\C}$ to the predicate $(X,P^\dag)\in \Pred{\C}$ obtained via the (strong epi, mono)-factorization of $m_P$:
\[
\begin{tikzcd}
P \ar[shiftarr={yshift=1.5em}]{rr}{m_P} \ar[two heads]{r}{e_{P^\dag}} & P^\dag \ar[tail]{r}{m_{P^\dag}} & X
\end{tikzcd}
\] 
\begin{rem}\label{rem:prod-coprod-vs-dag}
If strong epimorphisms in $\C$ are product-stable, then
\[ ((X,P)\times (Y,Q))^\dag = ((X,P)\times (Y,Q))^\dag \qquad\text{for all $(X,P),(Y,Q)\in \Ar{\C}$}. \]
Dually, if monomorphisms in $\C$ are coproduct-stable (which holds if $\C$ is extensive),
\[ ((X,P)+(Y,Q))^\dag = ((X,P)+(Y,Q))^\dag \qquad\text{for all $(X,P),(Y,Q)\in \Ar{\C}$}. \]
\end{rem}

Every endofunctor $\Sigma\colon \C\to \C$ admits a canonical lifting 
\[\ol{\Sigma}_{\mathbf{Ar}}\colon \Ar{\C}\to\Ar{\C},\qquad (X,P) \mapsto (\Sigma X, \Sigma P, \Sigma m_P).
\]
Note that the canonical predicate lifting $\ol{\Sigma}=\ol{\Sigma}_{\mathbf{Pred}}\colon \Pred{\C}\to \Pred{\C}$ is obtained by taking the canonical lifting to the arrow category and applying the reflector:
\[ \ol{\Sigma}_{\mathbf{Pred}}(X,P) = (\ol{\Sigma}_{\mathbf{Ar}}(X,P))^\dag. \]

\begin{lem}\label{lem:ol-sigma-vs-dag}
If $\Sigma$ preserves strong epimorphisms, then the following diagram commutes:
\[
\begin{tikzcd}
\Ar{\C} \ar{r}{\ol{\Sigma}_{\mathbf{Ar}}} \ar{d}[swap]{(-)^\dag} & \Ar{\C} \ar{d}{(-)^\dag} \\
\Pred{\C} \ar{r}{\ol{\Sigma}_{\mathbf{Pred}}} & \Pred{\C} 
\end{tikzcd}
\]
\end{lem}

\begin{proof}
For each $(X,P)\in \Ar{\C}$ we have
\[
(\ol{\Sigma}_{\mathbf{Ar}}(X,P))^\dag = (\Sigma X,\Sigma P)^\dag \cong (\Sigma X, \Sigma P^\dag)^\dag  = \ol{\Sigma}_{\mathbf{Pred}}((X,P)^\dag) 
\]
The isomorphism in the second step follows from the commutative diagram below and the uniqueness of (strong epi, mono)-factorizations. Note that $\Sigma e_P$ is a strong epimorphism by our assumption that $\Sigma$ preserves strong epimorphisms.
\[
\begin{tikzcd}[column sep=12em, row sep=5em]
\Sigma P \ar[two heads]{d}[swap]{e_{\Sigma P}} \ar{drr}[swap]{\Sigma m_P} \ar[two heads]{r}{\Sigma e_{P}} & \Sigma P^\dag \ar[two heads]{r}{e_{(\Sigma P^\dag)^\dag}} \ar{dr}{\Sigma m_P} & (\Sigma P^\dag)^\dag \ar[tail]{d}{m_{(\Sigma P^\dag)^\dag}} \\
(\Sigma P)^\dag \ar[tail]{rr}{m_{(\Sigma P)^\dag}} & & \Sigma X
\end{tikzcd}
\]
\end{proof}
Since $(-)^\dag$ is a left adjoint, the above lemma and \cite[Thm.~2.14]{hj98} yield

\begin{corollary}\label{cor:free-algebra-adjunction}
If $\Sigma$ preserves strong epimorphisms, there is an adjunction
\[
\begin{tikzcd}
\Alg(\ol{\Sigma}_{\mathbf{Ar}}) \ar[phantom]{r}{\bot} \ar[bend left=1em]{r}{F} & \Alg(\ol{\Sigma}_{\mathbf{Pred}}) \ar[bend left=1em]{l}{U}
\end{tikzcd}
\]
where the right adjoint $U$ is given by
\[  
(\ol{\Sigma}_{\mathbf{Pred}}(A,P)\xto{a} (A,P)) \;\;\mapsto\;\; (\ol{\Sigma}_{\mathbf{Ar}}(A,P) \xto{(\id,e_{\Sigma P})} \ol{\Sigma}_{\mathbf{Pred}}(A,P)\xto{a} (A,P))\qqand h\mapsto h,   
\]
the left adjoint $F$ by
\[ (\ol{\Sigma}_{\mathbf{Ar}}(A,P)\xto{a} (A,P)) \;\;\mapsto\;\;  (\ol{\Sigma}_{\mathbf{Pred}}(A,P^\dag) \xto{a^\dag} (A,P^\dag)) \qqand h\mapsto h^\dag, \]
and the unit at $((A,P),a)\in \Alg(\ol{\Sigma}_{\mathbf{Ar}})$ by
\[ (\id,e_{P})\c ((A,P),a)\to UF((A,P),a). \]
\end{corollary}

We can now state our main result about predicate liftings of free monads, which in particular proves~\Cref{prop:free-monad-lift}.
\begin{prop}\label{prop:free-monad-lift-pred}
If $\Sigma$ preserves strong epimorphisms, then the canonical liftings $\ol{\Sigma}_{\mathbf{Ar}}$ and $\ol{\Sigma}_{\mathbf{Pred}}$ generate free monads, and they are given by the canonical liftings of the free monad generated by $\Sigma$:
\[ (\ol{\Sigma}_{\mathbf{Ar}})^\star = \ol{\Sigmas}_{\mathbf{Ar}} \qqand (\ol{\Sigma}_{\mathbf{Pred}})^\star = \ol{\Sigmas}_{\mathbf{Pred}} \]
\end{prop}

\begin{proof}
Let $(\Sigmas,\eta,\mu)$ be the free monad generated by $\Sigma$, with 
associated natural transformation \[\iota\c \Sigma\Sigmas\to
  \Sigmas;\] that is, $\iota_X\c \Sigma\Sigmas X\to \Sigmas X$ is the
free $\Sigma$-algebra on $X$ with the universal morphism $\eta_X\c X
\to \Sigmas X$.
\begin{enumerate} 
\item\label{lem:free-monad-lift-1} We first consider the canonical lifting $\ol{\Sigma}=\ol{\Sigma}_{\mathbf{Ar}}$ to the arrow category. It suffices to show that a free $\ol{\Sigma}$-algebra on $(X,P)\in \Ar{\C}$ is given by 
\[ 
  \ol{\Sigma}\,\barSigmas(X,P)=(\Sigma\Sigmas X,\Sigma\Sigmas P)
  \xra{(\iota_X,\iota_P)} 
(\Sigmas X, \Sigmas P) = \barSigmas(X,P) 
\] 
with the universal morphism
\[ 
(X,P)
\xra{(\eta_X,\eta_P)}
 (\Sigmas X,\Sigmas P)=\barSigmas(X,P). 
\]
To prove this, we verify the required universal property. Thus suppose that we are given a $\ol{\Sigma}$-algebra $a\c \ol{\Sigma}(A,Q)\to (A,Q)$ and an $\Ar{\C}$-morphism $h\c (X,P)\to (A,Q)$. The universal property of the free $\Sigma$-algebra $(\Sigmas X,\iota_X)$ yields a unique $\Sigma$-algebra morphism
\[ \ol{h}_0\c (\Sigmas X,\iota_X)\to (A,a_0)\qquad\text{such that}\qquad \ol{h}_0\comp \eta_X=h_0, \]
and similarly the universal property of $(\Sigmas P,\iota_P)$ yields a unique $\Sigma$-algebra morphism
\[ \ol{h}_1\c (\Sigmas P,\iota_P) \to (Q,a_1)\qquad\text{such that}\qquad \ol{h}_1\comp \eta_P=h_1.\]
This implies that the $\ol{\Sigma}$-algebra morphism
\[\ol{h}=(\ol{h}_0,\ol{h}_1)\c (\Sigmas(X,P),(i_X,i_P))\to ((A,Q),a)\]
satisfies $\ol{h}\comp (\eta_X,\eta_P) = h$, and it is clearly unique with that property. 
\item Now consider the canonical predicate lifting ${\widetilde{\Sigma}}=\ol{\Sigma}_{\mathbf{Pred}}$. It suffices to show that a free $\widetilde{\Sigma}$-algebra on $(X,P)\in \Pred{\C}$ is given by
\[  
\widetilde{\Sigma}\widetilde{\Sigmas}(X,P) \xto{(\iota_X,\iota_P)^\dag} \widetilde{\Sigmas}(X,P)
\]
with the universal morphism 
\[ (X,P) \xto{(\eta_X,\eta_P)^\dag} \widetilde{\Sigmas}(X,P). \]
But this is immediate from part \ref{lem:free-monad-lift-1} above and \Cref{cor:free-algebra-adjunction}.\qedhere
\end{enumerate}
\end{proof}

\section{Predicate Liftings of Higher-Order GSOS Laws}\label{appsec:liftings}

\begin{prop}
If $\Sigma$ preserves strong epimorphisms, and strong epimorphisms are pullback-stable, then $\rho$ lifts to higher-order GSOS law of $\ol{\Sigma}$ over $\ol{B}$. 
\end{prop}

\begin{proof}
It suffices to show that for each $(X,P),(Y,Q)\in \Pred{\C}$ there exists a (necessarily unique) morphism 
\begin{equation}\label{eq:rho-lift}\ol{\rho}_{(X,P), (Y,Q)}\colon \ol{\Sigma}((X,P)\times \ol{B}((X,P),(Y,Q)) \to \ol{B}((X,P), \ol{\Sigma}^\star((X,P)+(Y,Q)))
\end{equation}
in $\Pred{\C}$ such that \[(\ol{\rho}_{(X,P),(Y,Q)})_0=\rho_{X,Y}.\] (Di)naturality of $\ol{\rho}$ then follows from that $\rho$, using that the forgetful functor from $\Pred{\C}$ to $\C$ is faithful. For this purpose, consider first arbitrary arrows $(X,P), (Y,Q)\in \Ar{\C}$ and form the following pullback:
\begin{equation}\label{eq:t-p-q}
\begin{tikzcd}[column sep=3em]
T_{P,Q} \pullbackangle{-45} \ar{r}{s_{P,Q}} \ar{d}[swap]{r_{P,Q}} & B(P,Q) \ar{d}{B(\id,\out_Q)} \\
B(X,Y) \ar{r}{B(\out_P,\id)} & B(P,Y) 
\end{tikzcd}
\end{equation}
In particular, the universal property of the pullback $T_{P,\Sigmas(P+Q)}$ yields the unique dashed arrow $\wt{\rho}_{P,Q}$ making the upper and the left-hand cell of the diagram below commute; note that the outside and all other cells commute either by definition or by (di)naturality of $\rho$.  
\[
\begin{tikzcd}[scale cd=.7, column sep=2em]
T_{P,\Sigmas(P+Q)} \ar[bend left=1em]{drrrr}{s_{P,\Sigmas(P+Q)}} \ar[shiftarr={xshift=-4em}]{ddd}[swap]{r_{P,\Sigmas(P+Q)}} & & & & \\
\Sigma(P\times T_{P,Q}) \ar[dashed]{u}{\wt{\rho}_{P,Q}} \ar{d}[swap]{\Sigma(m_P\times r_{P,Q})} \ar{rr}{\Sigma(\id\times s_{P,Q})} \ar{dr}{\Sigma(\id\times r_{P,Q})} & & \Sigma(P\times B(P,Q)) \ar{rr}{\rho_{P,Q}} \ar{d}{\Sigma(\id\times B(\id,m_Q))} & & B(P,\Sigmas(P+Q)) \ar{dd}{B(\id,\Sigmas(m_P+m_Q))} \ar{dl}[swap]{B(\id,\Sigmas(\id+m_Q))} \\
\Sigma(X\times B(X,Y) \ar{d}[swap]{\rho_{X,Y}} & \Sigma(P\times B(X,Y))  \ar{l}[swap]{\Sigma(m_P\times \id)} \ar{r}[yshift=.5em]{\Sigma(\id\times B(m_P,\id))} & \Sigma(P\times B(P,Y)) \ar{r}{\rho_{P,Y}} & B(P,\Sigmas(P+Y)) \ar{dr}[description]{B(\id,\Sigmas(m_P+\id))} &  \\
B(X,\Sigmas(X+Y)) \ar{rrrr}{B(m_P,\id)} & & & & B(P,\Sigmas(X+Y)) 
\end{tikzcd}
\]
The left-hand cell states that
\[(\rho_{X,Y},\wt{\rho}_{P,Q})\colon (\Sigma(X\times B(X,Y)), \Sigma(P\times T_{P,Q}), \Sigma(m_P\times r_{P,Q})) \to (B(X,\Sigmas(X+Y)), T_{P,\Sigmas(P+Q)}, r_{P,\Sigmas(P+Q)})\]
is a morphism of $\Ar{\C}$. Note that the domain is equal to $\ol{\Sigma}_{\mathbf{Ar}}((X,P)\times (B(X,Y), T_{P,Q})$. Assuming $(X,P),(Y,Q)\in \Pred{\C}$ and applying $(-)^\dag$ yields
\begin{align*} 
\ol{\Sigma}((X,P)\times \ol{B}((X,P),(Y,Q)) &= (\ol{\Sigma}_{\mathbf{Ar}}((X,P)\times (B(X,Y), T_{P,Q})^\dag))^\dag \\
 &= (\ol{\Sigma}_{\mathbf{Ar}}((X,P)^\dag\times (B(X,Y), T_{P,Q})^\dag))^\dag \\
&= (\ol{\Sigma}_{\mathbf{Ar}}(((X,P)\times (B(X,Y), T_{P,Q}))^\dag)^\dag \\
&= (\ol{\Sigma}_{\mathbf{Ar}}((X,P)\times (B(X,Y), T_{P,Q}))^\dag.
\end{align*}
Here, we use \autoref{rem:prod-coprod-vs-dag} and the third step and \autoref{lem:ol-sigma-vs-dag} in the fourth step. 

Next, the universal property of the pullback $T_{P,(\Sigmas(P+Q))^\dag}$ yields the unique dashed morphism making the diagram below commute:
\[
\begin{tikzcd}
T_{P,\Sigmas(P+Q)} \ar{r}{s_{P,\Sigmas(P+Q)}} \ar[dashed]{d} \ar[shiftarr={xshift=-7em}]{dd}[swap]{r_{P,\Sigmas(P+Q)}} & B(P,\Sigmas(P+Q)) \ar{d}{B(\id,e_{(\Sigmas(P+Q))^\dag})} \ar[shiftarr={xshift=7em}]{dd}{B(\id,\Sigmas(m_P+m_Q))} \\
T_{P,(\Sigmas(P+Q))^\dag} \ar{r}{s_{P,(\Sigmas(P+Q))^\dag}} \ar{d}[swap]{r_{P,(\Sigmas(P+Q))^\dag}} & B(P,(\Sigmas(P+Q))^\dag \ar{d}{B(\id,m_{(\Sigmas(P+Q))^\dag})} \\
B(X,\Sigmas(X+Y)) \ar{r}{B(m_P,\id)} & B(P,\Sigmas(X+Y))
\end{tikzcd}
\]
In particular, this proves that the image of $r_{P,\Sigmas(P+Q)}$ is contained in that of $r_{P,(\Sigmas(P+Q))^\dag}$; in other words, $(B(X,\Sigmas(X+Y)), T_{P,\Sigmas(P+Q)}, r_{P,\Sigmas(P+Q)})^\dag$ is contained in $\ol{B}((X,P),\ol{\Sigma}^\star((X,P)+(Y,Q))$. 

We conclude that $(\rho_{X,Y},\wt{\rho}_{P,Q})^\dag$ is a $\Pred{\C}$-morphism with domain $\ol{\Sigma}((X,P)\times \ol{B}((X,P),(Y,Q))$ and codomain contained in $\ol{B}((X,P),\ol{\Sigma}^\star((X,P)+(Y,Q))$, which proves the proposition. 
\end{proof}

\section{\texorpdfstring{On $\lambda$-laws and the Simply Typed
    $\lambda$-Calculus}
  {On lambda-laws and the Simply Typed Lambda Calculus}}
\label{app:lambdashenanigans}

The goal of this section is to expand on \Cref{sec:lambda-laws}; in particular, we will introduce $\lambda$-laws and their properties. We
begin by equipping our ambient category $\C$, subject to \Cref{assumptions},
with a closed monoidal structure $(\C, \Pt, \mon,\monto)$, meant to abstract
from the internal mechanisms of variable
management~\cite{DBLP:conf/lics/Fiore08}.
We introduce the following notation for transposes:
\begin{equation}\label{eq:natBijAdj}
  \begin{tikzcd}
    \C(X \mon Y, Z) \ar[bend left=1em]{r}{(\argument)^{\flat}}
    & \C(X, Y\monto Z) \ar[bend left=1em]{l}{(\argument)^{\sharp}}
  \end{tikzcd}
\end{equation}
Let us denote by
$\mathrm{ev}_{X,Y}$ the induced evaluation transformation $\id_{X\monto
  Y}^\sharp\c (X\monto Y) \mon X \to Y$, 
$\mathrm{comp}_{X,Y}$ the corresponding internal composition 
$(\mathrm{ev}_{Y,Z}\comp ((Y\monto
Z)\mon\mathrm{ev}_{X,Y})\comp\mathrm{assoc}_{Y\monto Z,X\monto Y,X})^\flat \c
(Y\monto Z) \mon (X\monto Y) \to (X\monto Z)$ and $\mathrm{assoc}$ the obvious
associativity transformation.

\subsection{Predicate liftings of closed monoidal structures}

The following properties related to predicate
liftings of the closed monoidal structures will be of use later on.

\begin{lem}\label{lemma:monPreservesStrongEpis}
	Let $(\C, \Pt, \mon,\monto)$ be a monoidal closed category. Then $\mon$ preserves strong epis in the first variable.
\end{lem}
\begin{proof}
	Let $e \colon A \sto B$ be a strong epi and $R$ be an object of $\C$. We need to prove that for a commutative square as below, with $m$ mono,  there exists a diagonal fill-in arrow from $B \mon R$ to $C$ making both triangles commute:
	\[
	\begin{tikzcd}
		A \mon R \ar[r,"e \mon R"] \ar[d,"u"'] & B \mon R \ar[d,"v"] \ar[dl,dashed] \\
		C \ar[r,>->,"m"] & D
	\end{tikzcd}
	\]
	Because $Y \monto \argument$ is a right adjoint for every $Y$, $\monto \colon \C^\opp \times \C \to \C$ preserves limits in the covariant argument, in particular pullbacks. Since any morphism $f \colon A \to B$ is mono if and only if
	\[
	\begin{tikzcd}
		A \ar[r,"\id_A"] \ar[d,"\id_A"'] & A \ar[d,"f"] \\
		A \ar[r,"f"'] & B
	\end{tikzcd}
	\]
	is a pullback, we have that $\monto$ preserves monos in the covariant argument, hence $R \monto m$ is mono. Therefore, given that $e$ is a strong epi, there exists an arrow $d \colon B \to (R \monto C)$ making everything commute in the diagram below:
	\[
	\begin{tikzcd}[column sep=3.5em]
	A \ar[r,"e",two heads] \ar[d,"u^\flat"'] & B \ar[d,"v^\flat"] \ar[dl,dashed,"d"'] \\
	R \monto C \ar[r,>->,"R \monto m"] & R \monto D
	\end{tikzcd}
	\]
	Our desired map is $d^\sharp \colon B \mon R \to C$. Indeed, by naturality of $(\argument)^\sharp \colon \C(B,R \monto C) \to \C(B \mon R, C)$ in~$B$ and $R$, we obtain
	\[
	d^\sharp \comp (e \mon R) = (d \comp e)^\sharp = (u^\flat)^\sharp = u,
	\]
	and by naturality of $(\argument)^\sharp$ in the third variable $C$:
	\[
	m \comp d^\sharp = ((R \monto m) \comp d)^\sharp = (v^\flat)^\sharp = v. \qedhere
	\] 
\end{proof}

\begin{prop}
	Let $(\C, \Pt, \mon,\monto)$ be a monoidal closed category with a strong-epi mono factorization system and with pullbacks, such that $\mon$ preserves strong epis in the second variable. Then $(\Pred{\C}, \Pt,  \tensorlift, \monto*)$ (where $\tensorlift$ and $\monto*$ are the canonical predicate liftings) is also monoidal closed.
\end{prop}
\begin{proof}
	Recall that the predicate lifting of $\mon$ is given by image factorization as follows: for $p \colon P \pred{}X$ and $q \colon Q \pred{} Y$,
	\[
	\begin{tikzcd}
	P \mon Q \ar[d,"p \mon q"'] \ar[r,two heads,"e_{p,q}"] & P \mathbin{\ol\mon} Q \ar[dl,>->,"p \mathbin{\ol\mon} q"] \\
	X \mon Y
	\end{tikzcd}	
	\]
	To ease the notation, in this proof we write $p \mathbin{\ol\mon} q$ for the image along $\mathbin{\ol\mon}$ of objects $p$ and $q$ of $\Pred\C$, instead of $m_{P,Q}$ as in~\eqref{eq:liftingpb}. 
	Because $\monto$ preserves monos in the covariant argument (see proof of Lemma~\ref{lemma:monPreservesStrongEpis}), its predicate lifting on subobjects $q \colon Q \pred{} Y$ and $r \colon R \pred{} Z$ is given, by~\eqref{eq:liftingpb}, as the following pullback:
		\[
		\begin{tikzcd}[column sep=3.5em]
				Q \monto* R \ar[r] \ar[d,>->,"q \monto* r"']    \pullbackangle{-45} & Q \monto R \ar[d,>->,"Q \monto r"] \\
				Y \monto Z \ar[r,"q \monto Z"] & Q \monto Z
			\end{tikzcd}
		\]
	First we prove that the monoidal structure of $\C$ lifts to $\Pred\C$. The monoidal unit of $\Pred\C$ is the total predicate $\id_\Pt \colon \Pt \to \Pt$. Given the associator $\alpha_{X,Y,Z} \colon (X \mon Y) \mon Z \to X \mon (Y \mon Z)$ and the unitors $\lambda_X \colon \Pt \mon X \to X$ and $\rho_X \colon X \mon \Pt \to X$ in $\C$, we prove that they define morphisms of $\Pred \C$ using the assumptions that $\C$ is monoidal \emph{closed} and that $\mon$ preserves strong epis in the second variable. 
	
	To that end, consider the following diagram:
	\[
	\begin{tikzcd}
		(P \mon Q) \mon R \ar[r,"\alpha_{P,Q,R}"] \ar[d,two heads, "e_{p,q} \mon R"']
		\arrow["(p \mon q) \mon r"',rounded corners,
		to path={[pos=0.75]
		-| ([yshift=-.5cm,xshift=-1.75cm]\tikztotarget.west)\tikztonodes 
		-| (\tikztotarget)}]{ddd}
		 & P \mon (Q \mon R) \ar[d,two heads,"P \mon e_{q,r}"] 
		 \arrow["p \mon{ (q \mon r)}",rounded corners,
		 to path={[pos=0.75]
		 	-| ([yshift=-.5cm,xshift=1.75cm]\tikztotarget.east)\tikztonodes 
		 	-| (\tikztotarget)}]{ddd}
		 \\
		(P \tensorlift Q) \mon R  \ar[d,two heads,"e_{p \tensorlift q,r}"'] 
		\arrow["(p \tensorlift q) \mon r"',rounded corners,
		to path={[pos=0.75]
			-| ([xshift=-.25cm]\tikztotarget.west)\tikztonodes
			-- (\tikztotarget)} ]{dd}
		& P \mon (Q \tensorlift R) \ar[d,two heads,"e_{p,q \tensorlift r}"] 
		\arrow["p \tensorlift (q \mon r)",rounded corners,
		to path={[pos=0.75]
			-| ([xshift=.25cm]\tikztotarget.east)\tikztonodes
			-- (\tikztotarget)} ]{dd}
		\\
		(P \tensorlift Q) \tensorlift R \ar[d,>->,"(p \tensorlift q) \tensorlift r"'] & P \tensorlift ( Q \tensorlift R) \ar[d,>->,"p \tensorlift (q \tensorlift r)"] \\
		(X \mon Y) \mon Z \ar[r,"\alpha_{X,Y,Z}"] & X \mon (Y \mon Z)
	\end{tikzcd}
	\]
	It commutes by functoriality of $\mon$ and by naturality of $\alpha$; moreover, $e_{p,q} \mon R$ is a strong epi by Lemma~\ref{lemma:monPreservesStrongEpis}, while $P \mon e_{q,r}$ is a strong epi by assumption. Hence there is a fill-in arrow from $(P \tensorlift Q) \tensorlift R$ to $P \tensorlift (Q \tensorlift R)$, exhibiting $\alpha_{X,Y,Z}$ as an arrow in $\Pred \C$ from $(p \mon q) \mon r$ to $p \mon (q \mon r)$. Analogously, one proves that its inverse is an arrow in $\Pred \C$ as well.
	
	Given $p \colon P \to X$, we have that $\lambda_X \colon \Pt \mon X \to X$ restricts at the level of predicates because of the diagonal fill-in property of $e_{I,p}$ and naturality of $\lambda$:
	\[
	\begin{tikzcd}
	\Pt \mon P \ar[r,"\lambda_P"] \ar[d,two heads,"e_{\Pt,p}"'] 
	\arrow["\Pt \mon p"',rounded corners,
	to path={[pos=.75]
	-| ([xshift=-.5cm]\tikztotarget.west)\tikztonodes
	-- (\tikztotarget)}]{dd}
	& P \ar[d,equal] \\
	\Pt \tensorlift P \ar[d,>->,"\Pt \tensorlift p"'] \ar[r,dashed] & P \ar[d,>->,"p"] \\
	\Pt \mon X \ar[r,"\lambda_X"] & X
	\end{tikzcd}
	\]
	Similarly for $\rho_X$. The fact that their inverses restrict as well is shown similarly. Naturality of $\alpha$, $\lambda$ and $\rho$ as well as the commutativity of the coherence diagrams required by the definition of monoidal category are all immediate from their corresponding properties at the level of $\C$. Hence $\Pred\C$ is monoidal, with the structure given by the predicate lifting of that of $\C$.
	
	Next, we aim to prove that there is a bijection of sets
	\[
	\Pred\C\Bigl(
	\begin{tikzcd}[cramped,sep=small]
		P \mathbin{\ol\mon} Q \ar[d,>->,"{p \mathbin{\ol\mon} q}"] \\ X \mon Y
	\end{tikzcd},
	\begin{tikzcd}[cramped,sep=small]
		R \ar[d,>->,"r"] \\ Z
	\end{tikzcd}
	\Bigr)
	\cong
	\Pred\C \Bigl(
	\begin{tikzcd}[cramped,sep=small]
		P \ar[d,>->,"p"] \\ X
	\end{tikzcd},
	\begin{tikzcd}[cramped,sep=small]
		Q \monto* R \ar[d,>->,"q \monto* r"] \\ Y \monto Z
	\end{tikzcd}
	\Bigr)
	\]
	natural in $p \colon P \pred{}X$, $q \colon Q \pred{} Y$, and $r \colon R \pred{} Z$, knowing that there is a natural bijection~\eqref{eq:natBijAdj}.

	First, let $f \colon X \mon Y \to Z$ be a morphism in $\C$ such that there is a necessarily unique $f|_{P \tensorlift Q}$ making the diagram below on the left commute:
	\begin{equation}\label{eq:ale1}
	\begin{tikzcd}
		P \mon Q \ar[r,two heads,"e"] \ar[dr,"p \mon q"'] & P \tensorlift Q \ar[r,"f\mid_{P \tensorlift Q}"] 
		\ar[d,>->,"p \tensorlift q"] &[1em] R \ar[d,>->,"r"] 
		\\
		& X \mon Y \ar[r,"f"'] & Z
	\end{tikzcd}
	\qquad \qquad
	\begin{tikzcd}
		P \ar[d,>->,"p"'] \ar[r,dashed,"?"] & Q \monto* R \ar[d,>->,"q \monto* r"] \\
		X \ar[r,"f^\flat"'] & Y \monto Z
	\end{tikzcd}
	\end{equation}
	We will show that $f^\flat$ restricts to $P$ as in the diagram above on the right, thus exhibiting it as a morphism in $\Pred\C$. The argument is summarised in the following diagram:
	\[
	\begin{tikzcd}
		P \ar[dd,>->,"p"'] %
		\arrow["(f\mid_{P \tensorlift Q} \comp e)^\flat",rounded corners,
		to path={[pos=0.25]
		-| (\tikztotarget)\tikztonodes
		}]{rrd}
		\ar[dr,dashed,"{f^\flat}\mid_P"]\\
		& Q \monto* R \ar[r] \ar[d,>->,"q \monto* r"']    \pullbackangle{-45} & Q \monto R \ar[d,>->,"Q \monto r"] \\
		X \ar[r,"f^\flat"'] & Y \monto Z \ar[r,"q \monto Z"'] & Q \monto Z
	\end{tikzcd}
	\]
	The existence of the dashed arrow follows from the universal property of the pullback, which can be invoked because the outer diagram commutes, as we show now. Using the naturality of~\eqref{eq:natBijAdj} in $X$ and $Y$ %
	instantiated at the $\C$-morphisms $p \colon P \to X$ and $q \colon Y$, we obtain
	$(q \monto Z) \comp f^\flat \comp p = (f \comp (p \mon q))^\flat$, the latter
  being the same as $(r \comp {f|_{P \tensorlift Q}} \comp e)^\flat$ by the left diagram of~\eqref{eq:ale1}. On the other hand, naturality of~\eqref{eq:natBijAdj} in $Z$ instantiated at $r \colon R \to Z$ tells us that the following square commutes:
	\[
	\begin{tikzcd}
	\C(P \mon Q, R) \ar[r,"(\argument)^\flat"] \ar[d,"{\C(P \mon Q, r)}"'] & \C(P,Q\monto R) \ar[d,"{\C(P,Q \monto r)}"] \\
	\C(P \mon Q, Z) \ar[r,"(\argument)^\flat"] & \C(P,Q \monto Z)
	\end{tikzcd}
	\]
	so, for $f|_{P \tensorlift Q} \comp e \in \C(P \mon Q, R)$, we have $(Q \monto r) \comp (f|_{P \tensorlift Q} \comp e)^\flat = (r \comp {f|_{P \tensorlift Q}} \comp e)^\flat$. Hence,  $(Q \monto r) \comp (f|_{P \tensorlift Q} \comp e)^\flat =(q \monto Z) \comp f^\flat \comp p$, as desired.
	
	Next, let $g \colon X \to Y \monto Z$ be a morphism in $\C$ that restricts to $g|_P \colon P \to Q \monto* R$, making therefore the following diagram commutative:
	\[
	\begin{tikzcd}
		P \ar[r,"g\mid_P"] \ar[d,>->,"p"'] & Q \monto* R \ar[r,"u"] \ar[d,>->,"q \monto* r"']    \pullbackangle{-45} & Q \monto R \ar[d,>->,"Q \monto r"] \\
		X \ar[r,"g"'] & Y \monto Z \ar[r,"q \monto Z"'] & Q \monto Z
	\end{tikzcd}
	\]
	We will prove that $g^\sharp \colon X \mon Y \to Z$ restricts to ${g^\sharp}|_{P \tensorlift Q} \colon P \tensorlift Q \to R$, as in:
	\[
	\begin{tikzcd}
		P \mon Q \ar[r,two heads,"e"] \ar[dr,"p \mon q"'] & P \tensorlift Q \ar[r,dashed,"{g^\sharp}\mid_{P \tensorlift Q}"] \ar[d,"p \tensorlift q"] & R \ar[d,"r"]\\
		& X \mon Y \ar[r,"g^\sharp"] & Z
	\end{tikzcd}
	\]
	The desired dashed arrow follows from the diagonal fill-in property of $e$ as a strong epi, due to the commutativity (to be proved in a moment) of the following diagram:
	\[
	\begin{tikzcd}
	P \mon Q \ar[r,two heads,"e"] \ar[dd,"(u \comp g\mid_P)^\sharp"'] & P \tensorlift Q \ar[d,"p \tensorlift q"] \ar[ddl,dashed] \\
	& X \mon Y \ar[d,"g^\sharp"] \\
	R \ar[r,>->,"r"] & Z
	\end{tikzcd}
	\]
	Again using the naturality of~\eqref{eq:natBijAdj} in $X$ and $Y$, this time with $(\argument)^\sharp$ instead of $(\argument)^\flat$, applied to $g \in \C(X,Y \monto Z)$, we obtain: 
	\[
	g^\sharp \comp (p \mon q) = ((q \monto Z) \comp g \comp p)^\sharp = ((Q \monto r) \comp u \comp g|_P)^\sharp,
	\]
	while naturality in $Z$ applied to $u \comp g|_P \in \C(P,Q \monto R)$ yields:
	\[
	r \comp (u \comp g|_P)^\sharp =  ((Q \monto r) \comp u \comp g|_P)^\sharp.
	\]
	Recalling that $p \mon q = (p \tensorlift q) \comp e$, we can conclude.
	
	The fact that the $\sharp$ and $\flat$ constructions yield a natural bijection at the level of $\Pred \C$ is immediate from the definition of predicate liftings of $\mon$ and $\monto$ on morphisms of $\Pred \C$.
\end{proof}

\subsection{\texorpdfstring{$\lambda$-laws}
  {Lambda-laws}}

We are now ready to present the theory of $\lambda$-laws. First, we recall the
following technical notion of pointed strength used in their definition.

\begin{notation}
We let $j\colon V/X\to X$ denote the projection functor $(X,p_X)\mapsto X$. We often leave points implicit and write $Y$ for $jY$.
\end{notation}

\begin{defn}[Pointed
  Strength~\cite{DBLP:conf/lics/Fiore08,10.1007/978-3-030-99253-8_20}]
  \label{def:pointed-str}
A \emph{$\Pt$-pointed strength} on an endofunctor $F\c\C\to\C$ is a family of morphisms
$\mathrm{st}_{X,Y}\c FX\mon jY\to F(X\mon jY)$, natural in $X\in \C$ and
$Y\in \Pt /\C$, such that the following diagrams commute:
\begin{equation*}
\begin{tikzcd}[column sep=0em, row sep=normal]
 & FX\\
 FX\mon\Pt  & & F(X\mon\Pt) 
 \arrow[iso, rotate=180, from=1-2, to=2-1]
 \arrow["\mathrm{st}_{X,\Pt}", from=2-1, to=2-3]
 \arrow[iso, from=1-2, to=2-3] 
\end{tikzcd}
\quad
\begin{tikzcd}[column sep=3em, row sep=normal]
  (F X\mon jY)\mon jZ & F(X\mon jY)\mon jZ & F((X\mon jY)\mon jZ) \\
  F X\mon (jY\mon jZ) & & F(X\mon (jY\mon jZ))
  \arrow[from=1-1,to=2-1, iso] 
  \arrow[from=1-1,to=1-2, "\mathrm{st}_{X,Y}\mon jZ"]
  \arrow[from=1-2,to=1-3, "\mathrm{st}_{X\mon jY,Z}"] 
  \arrow[from=2-1,to=2-3, "\mathrm{st}_{X,Y\mon Z}"]
  \arrow[from=1-3,to=2-3, iso]
\end{tikzcd}
\end{equation*}
(eliding the names of the canonical isomorphisms).
\end{defn}

Next, we give the definition of a $\lambda$-law.

\begin{defn}[$\lambda$-law]\label{def:lambdalaw}
  Let $B \c \C^{\opp} \times \C \to C$ be a bifunctor in $\C$, $\Theta \c \C \to
  \C$ be a $\Pt$-pointed strong endofunctor and let $\Xi\c \C \to \C$ be an
  endofunctor of the form $\Pt+\Xi'$ where $\Xi'$ is $\Pt$-pointed strong.
    A \emph{$\lambda$-law of\, $(\Xi, \Theta)$ over $B$} is a pair
    $(\xi, \theta)$ consisting of
    \begin{enumerate}
    \item A ($0$-pointed) higher-order GSOS law $\xi_{X,Y} \c \Xi(X
      \times B(X,Y)) \to B(X, \Xi^{\star}(X+Y))$.
    \item A family of morphisms $\theta_{X,Y}  \c \Theta (jX\monto Y) \to jX
      \monto B(jX,Y),$
    dinatural in $X \in \Pt/\C$ and natural in $Y\in \C$ %
    , such that the following pentagon commutes:
    \begin{equation}
      \label{eq:thetacompat}
      \begin{tikzcd}[column sep=1em, row sep=4ex]
        \Theta (X \monto Y) \mon (X \monto X) \mon X
        \ar[rr, "\mathrm{st}^{\Theta}_{X \monto Y,X \monto X} \mon X"]\
        \dar["\Theta (X \monto Y) \mon \mathrm{ev}_{X,X}"'] & &
        \Theta ((X \monto Y) \mon (X \monto X)) \mon X
        \dar["{\Theta \mathrm{comp}_{X,X,Y} \mon X}"]\\
        \Theta (X \monto Y) \mon X
        \ar[dr, "\theta^{\sharp}_{X,Y}"']
        & &\Theta (X \monto Y) \mon X
        \ar[dl, "\theta^{\sharp}_{X,Y}"]\\[-4ex]
        & B(X,Y) & 
      \end{tikzcd}
    \end{equation}
    (eliding the obvious associativity transformation). Note the the use of
    strength here is legit, since $X\monto X$ is canonically pointed.
  \end{enumerate}
\end{defn}

In a $\lambda$-law, functors $\Xi$ and $\Theta$ model the syntax of the
language and $\xi$ and $\theta$ the semantics. The $0$-pointed higher-order GSOS
law $\xi$ follows the classic (higher-order) GSOS paradigm. As for $\theta$, it
abstractly captures the fact that the behaviour of a $\lambda$-abstraction is
determined by the internal substitution structure of its body in a well-behaved
manner. Concretely, it allows one to show that if the body of a
$\lambda$-abstraction satisfies the logical predicate, then so does the
$\lambda$-abstraction itself. The $\lambda$-laws can be understood as a proper
subset of higher-order GSOS laws for languages with variable binding.
In particular, any $\lambda$-law $(\Xi,\Theta, B, \xi,\theta)$ determines a
$\Pt$-pointed higher-order GSOS law 
\begin{equation}
  \label{eq:lamtogsos}
    \rho_{X,Y} \c \Sigma(X \times (X\monto Y) \times B(X,Y)) \to
      (X\monto \Sigma^{\star}(X + Y)) \times B(X, \Sigma^{\star}(X + Y)) 
\end{equation}
of $\Sigma = \Xi + \Theta$ over
$(\argument\monto\argument) \times B$, and therefore the previously developed 
theory of abstract higher order GSOS laws can (and will) be reused. For any 
$X \in \Pt/\C$, and $Y\in\C$ we define the trivial substitution on the substitution structure~$X\monto Y$:
\begin{align*}
  \mathrm{triv}_{X,Y} \c X \monto Y \iso (X\monto Y) \mon \Pt 
  \xto{(\mathrm{var}_{X}\monto Y)^\sharp} Y.
\end{align*}
Then we define the components of~\eqref{eq:lamtogsos} as
copairs $[\brks{\rho_{11}^{\flat}, \rho_{12}}, \brks{\rho_{21}^{\flat}, \rho_{22}}]$
where
\begin{equation*}
  \begin{aligned}
    \rho_{11} \c \Xi(X \times (X\monto Y) \times B(X,Y)) \mon X
    &
      \xrightarrow{\Xi(\fst\comp\snd) \mon X} \Xi(X\monto Y) \mon X\iso X+\Xi'(X\monto Y) \mon X\\
    & \xrightarrow{X+\mathrm{st}^{\Xi'}_{X\monto Y,X}} X+\Xi'((X\monto Y) \mon X)\\
    & \xrightarrow{X+\Xi'\mathrm{ev}_{X,Y}} X+\Xi' Y \xrightarrow{[\eta \comp \inl,\Sigma^{\star}\inr]} \Sigma^{\star}(X + Y),\\[1ex]
    \rho_{12} \c \Xi(X \times (X\monto Y) \times B(X,Y))
    &
      \xrightarrow{\Xi(\id\times\snd)} \Xi(X \times B(X,Y))
      \xrightarrow{\xi_{X,Y}} B(X, \Xi^{\star}(X+Y)) \\
    & \xrightarrow{B(X,\inl^{\klstar})} B(X, \Sigma^{\star}(X+Y)), \\[1ex]
    \rho_{21} \c \Theta(X \times (X\monto Y) \times B(X,Y)) \mon X
    &
      \xrightarrow{\Theta(\fst\comp\snd)\mon X} \Theta (X\monto Y) \mon X
      \xrightarrow{\mathrm{st}^{\Theta}_{X\monto Y,X}} \Theta((X\monto Y) \mon X)\\
    & \xrightarrow{\Theta \mathrm{ev}_{X,Y}} \Theta Y
      \xrightarrow{\inr^{\klstar}} \Sigma^{\star}Y \xrightarrow{\Sigma^{\star}\inr} \Sigma^{\star}(X + Y),\\[1ex]
    \rho_{22} \c \Theta(X \times (X\monto Y) \times B(X,Y))
    & \xrightarrow{\Theta(\fst\comp\snd)} \Theta(X\monto Y) \xrightarrow{\theta_{X,Y}} X\monto B(X,Y)\\
    &  \xrightarrow{\mathrm{triv}_{B(X,Y)}} B(X, Y)
      \xrightarrow{B(X,\eta \comp \inr)} B(X, \Sigma^{\star}(X + Y)).
  \end{aligned}
\end{equation*}
The family $\rho$ is dinatural in $X$ because all of its components are. We consider
the canonical operational model of a $\lambda$-law to be the canonical model of
the derived higher-order GSOS law, which in this case is a pair of morphisms
\begin{align*}
  \brks{\phi,\gamma} \c \mS  \to   (\mS\monto\mS) \times B(\mS, \mS).
\end{align*}

Map $\gamma \c \mS \to B(\mS, \mS)$ models the computational behaviour terms. In
addition, the initial $(\Xi + \Theta)$-algebra $\mS$ is a $\mon$-monoid with
$\phi \c \mS \to (\mS \monto \mS)$ as the multiplication and the inclusion of
variables as unit. In fact, $\mS$ is a $\Xi' + \Theta$-monoid~\cite[\textsection
4]{DBLP:conf/lics/FiorePT99}, a fact we prove in the next lemma.

\begin{notation}
  In the sequel, we use $\iota^{\Xi'} = \iota
  \comp \inr \comp \inl \c \Xi'\mS\to\mS$, $\iota^\Xi =
  \iota\comp\inl\c\Xi\mS\to\mS$, and $\iota^\Theta=
  \iota\comp\inr\c\Theta\mS\to\mS$. The canonical point of $\mathrm{var}_{\mS} \c
  \Pt \to \mS$ is $\mathrm{var}_{\mS} = \iota^{\Pt} = \iota \comp \inl \comp
  \inl$. In addition, as $\mS = (\Pt + \Xi' + \Theta)^{\star}(0) = (\Xi' +
  \Theta)^{\star}(\Pt)$, $\mathrm{var}_{\mS} = \eta_{\Pt}$.
\end{notation}

\begin{lem}\label{lem:phi_prop}
  $\mS$ is a $\Xi' + \Theta$-monoid~\cite[\textsection 4]{DBLP:conf/lics/FiorePT99}. In
  particular, it is a $\mon$-monoid with the canonical point
  $\mathrm{var}_{\mS}\c\Pt\to\mS$ as unit and $\phi^{\sharp}\c\mS\mon\mS\to\mS$
  as multiplication, and the following two diagrams commute:
\begin{equation}
  \label{eq:sigmamonoid}
  \begin{tikzcd}[column sep=normal, row sep=normal]
    {\Xi'}\mS\mon\mS
    \rar["\iota^{{\Xi'}}\mon\mS"]
    \dar["\mathrm{st}^{\Xi'}_{\mS,\mS}"'] & \mS\mon\mS\ar[dd, "\phi^\sharp"]\\
    {\Xi'}(\mS\mon\mS)
    \dar["{\Xi'}\phi^\sharp"'] & \\
    {\Xi'}\mS
    \rar["\iota^{\Xi'}"] & \mS
  \end{tikzcd}
  \qquad
  \begin{tikzcd}[column sep=normal, row sep=normal]
    \Theta\mS\mon\mS
    \rar["\iota^{\Theta}\mon\mS"]
    \dar["\mathrm{st}^\Theta_{\mS,\mS}"'] & \mS\mon\mS\ar[dd, "\phi^\sharp"]\\
    \Theta(\mS\mon\mS)
    \dar["\Theta\phi^\sharp"'] & \\
    \Theta\mS
    \rar["\iota^\Theta"] & \mS
  \end{tikzcd}
\end{equation}
\end{lem}

\begin{proof}
  First, by the pointed strength $\mathrm{st}^{\Xi'}_{X,Y} \c \Xi' X \mon Y \to \Xi'(X
  \mon Y)$ and $\mathrm{st}^{\Theta}_{X,Y} \c \Theta X \mon Y \to \Theta(X \mon Y)$
  we can canonically equip $\Xi + \Theta$ with a pointed strength $\mathrm{st}_{X,Y}
  \c (\Xi' + \Theta)X \mon Y \to (\Xi' + \Theta)(X \mon Y)$, as $\argument \mon Y$
  preserves colimits.
  Next, recall that the canonical operational model $\brks{\phi,\gamma}\c \mS  \to
  (\mS\monto\mS) \times B(\mS, \mS)$ is obtained as a unique solution of the
  identity represented by the diagram below (with $\Sigma = \Pt + \Xi + \Theta$):
  \begin{equation*}
    \begin{tikzcd}[column sep=normal, row sep=normal]
      \Sigma\mS
      \ar[dd, "\iota"']
      \rar["{\Sigma\brks{\mS,\brks{\phi,\gamma}}}"]
      & \Sigma(\mS \times ((\mS \monto \mS) \times B(\mS,\mS)))
      \dar["{[\brks{\rho_{11}, \rho_{12}^\flat},
        \brks{\rho_{21}, \rho_{22}^\flat}]}"]
      \\
      & (\mS \monto \Sigma^{\star}(\mS + \mS)) \times B(\mS, \Sigma^{\star}(\mS +
      \mS))
      \dar["{(\mS\monto\hat\iota\comp\Sigma^\star\nabla)\times B(\mS,\hat\iota\comp\Sigma^\star\nabla)}"]
      \\
      \mS
      \rar["{\brks{\phi,\gamma}}"]
      & (\mS \monto \mS) \times B(\mS, \mS)
    \end{tikzcd}
  \end{equation*}
  By the definitions of $\rho_{11}$ and $\rho_{21}$, we can project the first
  component and simplify to
  \begin{equation*}
    \begin{tikzcd}[column sep=normal, row sep=normal]
      {(\Pt + \Xi' + \Theta)\mS}
      \ar[dd, "\iota"']
      \ar[rr, "{(\Pt + \Xi' + \Theta)\phi}"]
      & & \Pt + \Xi'(\mS \monto \mS) + \Theta(\mS \monto \mS)
      \dar["{([\cong,(\mathrm{ev} \comp \mathrm{st}^{\Xi'}),(\mathrm{ev} \comp \mathrm{st}^{\Theta} ]\comp \cong)^{\flat}}"]
      \\
      & &
      \mS \monto (\mS + \Xi'\mS + \Theta\mS)
      \dar["{\mS \monto [\mS,\iota^{\Xi'},\iota^{\Theta}]}"]
      \\
      \mS
      \ar[rr, "{\phi}"]
      & & (\mS \monto \mS)
    \end{tikzcd}
  \end{equation*}
  Transposing the diagram and breaking it into the three components for
  $\Pt \mon \mS$, $\Xi'\mS \mon \mS$ and $\Theta \mS \mon \mS$, we obtain the
  diagrams in \eqref{eq:sigmamonoid} as well as unit diagram
  \begin{equation*}
  \begin{tikzcd}
    \Pt \mon \mS
    \ar[dr, "\cong"]
    \ar[d, "\mathrm{var}_{\mS} \mon \mS"']
    \\
    \mS \mon \mS
    \ar[r, "{\phi^{\sharp}}"']
    & \mS
  \end{tikzcd}
\end{equation*}
Putting the three diagrams \eqref{eq:sigmamonoid} and the above back
together, we obtain
\begin{equation*}
  \begin{tikzcd}[column sep=normal, row sep=normal]
    {(\Xi' + \Theta)\mS \mon \mS}
    \ar[dd,"{[\iota^{\Xi'},\iota^{\Theta}] \mon \mS}"']
    \ar[r,"{\mathrm{st}}"]
    & (\Xi' + \Theta)(\mS \mon \mS)
    \ar[r,"(\Xi' + \Theta)\phi^{\sharp}"]
    & (\Xi' + \Theta)\mS
    \ar[dd,"{[\iota^{\Xi'},\iota^{\Theta}]}"]
    \\
    \\
    \mS \mon \mS
    \ar[rr,"{\phi^{\sharp}}"]
    & & \mS
    \\
    \\
    \Pt \mon \mS
    \ar[uu,"\mathrm{var}_{\mS} \mon \mS"]
    \ar[rruu,"\cong"]
  \end{tikzcd}
\end{equation*}
Since $\mS = (\Xi' + \Theta)^{\star}(\Pt)$, $\mathrm{var}_{\mS} = \eta_{\Pt}$,
the above shows $\phi^{\sharp} \c \mS \mon \mS \to \mS$ is indeed the unique
inductive extension of the isomorphism $\Pt \mon \mS \cong \mS$.
By~\cite[Th. 4.1]{DBLP:conf/lics/FiorePT99}, $\mS$ is a $(\Xi' + \Theta)$-monoid.
\end{proof}

\begin{defn}
  A $\lambda$-law $(\xi, \theta)$ is \emph{relatively flat}
  when $\xi$ is so (see \Cref{def:relativelyflat}).
\end{defn}
Let $\rho_{X,Y} \c \Sigma(X \times (X\monto Y) \times B(X,Y)) \to
  (X\monto \Sigma^{\star}(X + Y)) \times B(X, \Sigma^{\star}(X + Y))$
be a GSOS law induced by a lambda law $(\Xi,\Theta,\xi,\theta)$. The fundamental
problem motivating the introduction of $\lambda$-laws is that, in order to use
\Cref{th:main2} on $\rho$, one needs the bifunctor $\argument \monto \argument \times
B(\argument, \argument)$ to be contractive while, in practice, \emph{only} $B$ is
contractive. However, $\lambda$-laws come with their own simplifying up-to
principle, which works similarly to \Cref{th:main2}, only that we are looking to
prove the \emph{open extension} of $\square^{\gamma,\ol{B}}{P}$:
\[\blacksquare^{\gamma,\ol{B}} P =
  \iimg{\phi}{\square^{\gamma,\ol{B}}{P} \mathbin{\ol{\monto}} \square^{\gamma,\ol{B}}{P}}.\]

\begin{theorem}[Induction up to
  $\blacksquare$]\label{thm:ind-up-to-blackquaregen}
  Let $\brks{\phi, \gamma} \c \mS \to \mS \monto \mS \times B(\mS,\mS)$ be the
  operational model of a relatively flat $\lambda$-law $(\Xi, \Theta, \xi,
  \theta)$ that admits a predicate lifting. Then for every predicate $P \pred{} \mS$ and every
  locally maximal logical refinement $\logp^{\gamma,\ol{B}}P$,
  \[ \ol{\Xi+\Theta}(\square^{\gamma,\ol{B}})\leq \iimg{\ini}{P}
    \qquad\text{implies}\qquad  \ol{\Xi+\Theta}(\blacksquare^{\gamma,\ol{B}}P)\leq
    \iimg{\ini}{\blacksquare^{\gamma,\ol{B}}P} \quad\text{(hence
      $\blacksquare^{\gamma,\ol{B}}P=\mS$)}. \]
\end{theorem}
\begin{rem}
  The property $\blacksquare^{\gamma,\ol{B}}P = \mS$ is precisely what is known
  in the literature as the \emph{fundamental property} of the logical predicate.
\end{rem}
For the proof of \Cref{thm:ind-up-to-blackquaregen}, we first note that if
$(\Xi, \Theta,z \xi, \theta)$ lifts and $\Pt \leq \iimg{\ini}{P}$, $\logp P$ is
automatically \emph{pointed}, in that 
there is a morphism $(\Pt \monoto \Pt) \to (\logp P \monoto \mS)$ in
$\Pred{\C}$. In particular, $(\Pt \lor \logp P) \monoto \mS$ can be shown to be a $\logp
P$-relative invariant through the lifting of $\xi_{X,Y} \c \Xi(X
\times B(X,Y)) \to B(X, \Xi^{\star}(X+Y))$.

Next, we make use of the following technical lemma:
\begin{lem}
  \label{lem:zetaphi}
  Let $(\Xi,\Theta,B,\xi,\theta)$ be a $\lambda$-law. The following diagram commutes:
\begin{equation*}
 \begin{tikzcd}[column sep=2em, row sep=1ex]
    \Theta\mS\mon\mS
    \ar[ddddd, "\Theta\phi \mon \mS"']
    \ar[rrrr, "\iota^\Theta\mon\mS"]
    &[1em] &[-1em] &[.5em] 
    & \mS\mon\mS\ar[dd, "\phi^\sharp"]
    \\[2ex]
    & &  
    \\
    & & & & \mS \ar[ddd, "{\gamma}"]
    \\[1ex]
    & & & 
    \\
    & & & &
    \\
    \Theta(\mS \monto \mS) \mon \mS
    \ar[rrrr,"\theta^{\sharp}_{\mS,\mS}"] & & & & B(\mS, \mS)
  \end{tikzcd}
\end{equation*}
\end{lem}

\begin{proof}
The proof is displayed in the diagram:
\[
  \adjustbox{scale=0.90,center}{
    \begin{tikzcd}[column sep=0em, row sep=4ex]
      \Theta\mS\mon\mS
      \ar[ddddr, phantom, "(1)\qquad"] 
      \ar[ddrrrr, phantom, "\hspace{5em}(2)"]
      \ar[ddddd, "\Theta\phi \mon \mS"']
      \ar[dr,"{\mathrm{st}^{\Theta}_{\mS,\mS}}"]
      \ar[rrrr, "\iota^\Theta\mon\mS"]
      &[1em] &[-1em] &[.5em] & \mS\mon\mS\ar[dd, "\phi^\sharp"]
      \\[1ex]
      & \Theta(\mS\mon\mS) 
        \ar[dr,"\Theta\phi^\sharp"]
        \dar[iso] &  
      \\
      & \Theta(\mS\mon\mS)\mon\Pt\dar["\Theta(\phi\comp\phi^\sharp)\mon\Pt"'] & \Theta\mS
      \ar[d,"\Theta\phi"]
      \ar[rr, "\iota^\Theta"]
      \ar[dddrr, phantom, "\hspace{8em}(3)", pos=.3]
      & & \mS \ar[ddd, "{\gamma}"]
      \\[1ex]
      & \Theta(\mS\monto\mS)\mon\Pt\rar[iso]\ar[dr, "\theta_{\mS,\mS}\mon\Pt"]\dar["\Theta(\mS\monto\mS)\mon\mathrm{var}_{\mS}"'] & \Theta(\mS \monto \mS) 
      \ar[rd, "{\theta_{\mS,\mS}}"] & 
      \\[2ex]
      & \Theta(\mS\monto\mS)\mon\mS~~\ar[d,"\theta_{\mS,\mS}^\sharp"'] & ~~(\mS \monto B(\mS,\mS))\mon\Pt\rar[iso]\ar[drr,"{(\mathrm{var}_{\mS}\monto B(\mS,\mS))^\sharp}"'] &  \mS \monto B(\mS,\mS) \ar[dr, pos=.3, "{\mathrm{triv}_{B(\mS,\mS)}}"]  &
      \\[3ex]
      \Theta(\mS \monto \mS) \mon \mS
      \ar[r,"\theta^{\sharp}_{\mS,\mS}"] & B(\mS, \mS)\ar[rrr,equal] & & & B(\mS, \mS)
    \end{tikzcd}
    }
  \]
  Here, (1), (2) and (3) are the only non-trivially commuting cells. Commutativity of~(2)
  is shown in \Cref{lem:phi_prop}.

  To show (3), we recall that the operational model $\brks{\phi,\gamma}\c \mS  \to
  (\mS\monto\mS) \times B(\mS, \mS)$ makes the following diagram commute:
  \[
    \begin{tikzcd}[column sep=normal, row sep=normal]
      \Sigma\mS
      \ar[dd, "\iota"']
      \rar["{\Sigma\brks{\mS,\brks{\phi,\gamma}}}"]
      & \Sigma(\mS \times ((\mS \monto \mS) \times B(\mS,\mS)))
      \dar["{[\brks{\rho_{11}, \rho_{12}^\flat},
              \brks{\rho_{21}, \rho_{22}^\flat}]}"]
      \\
      & (\mS \monto \Sigma^{\star}(\mS + \mS)) \times B(\mS, \Sigma^{\star}(\mS +
      \mS))
      \dar["{(\mS\monto\hat\iota\comp\Sigma^\star\nabla)\times B(\mS,\hat\iota\comp\Sigma^\star\nabla)}"]
      \\
      \mS
      \rar["{\brks{\phi,\gamma}}"]
      & (\mS \monto \mS) \times B(\mS, \mS)
    \end{tikzcd}
  \]
We obtain (3) by precomposing both sides of the identity with $\inr$, and postcomposing 
it with $\fst$, and by calling the definition of $\rho_{2,1}$.
Finally, commutativity of $(3)$ is proven as follows:
\begin{equation*}
\begin{tikzcd}[column sep=.75em, row sep=5ex]
    \Theta\mS \mon \mS
    \ar[rr, "\mathrm{st}_{\mS,\mS}^\Theta"]
    \dar[iso]
    \dar[ddd,shiftarr = {xshift=-75}, pos=-.05, "\quad\Theta\phi\mon\mS"]
    & &
    \Theta (\mS \mon\mS)
    \dar[iso]
    \\[-2ex]
    \Theta\mS \mon \mS \mon \Pt
    \ar[rr, "\mathrm{st}_{\mS,\mS}^\Theta \mon \Pt"]
    \dar["\Theta\phi\mon\phi\mon\mathrm{var}_{\mS}"'] & &
    \Theta (\mS \mon\mS) \mon \Pt
    \dar["\Theta(\phi\mon\phi)\mon\mathrm{var}_{\mS}"]
    \dar[dd,shiftarr = {xshift=85}, pos=-.075, "\hspace{-5.5em}\Theta(\phi\comp\phi^\sharp)\mon\mathrm{var}_{\mS}"]
    \\
    \Theta (\mS\monto\mS) \mon (\mS\monto\mS) \mon \mS
    \ar[rr, "\mathrm{st}_{\mS \monto \mS,\mS \monto \mS}^\Theta \mon \mS"]
    \dar["\Theta (\mS \monto \mS) \mon \mathrm{ev}_{\mS,\mS}"'] & &
    \Theta ((\mS \monto \mS) \mon (\mS \monto \mS)) \mon \mS
    \dar["{\Theta \mathrm{comp}_{\mS,\mS,\mS} \mon \mS}"]
    \\
    \Theta (\mS \monto \mS) \mon \mS
    \ar[dr, "\theta^{\sharp}_{\mS,\mS}"']
    & &\Theta (\mS \monto \mS) \mon \mS
    \ar[dl, "\theta^{\sharp}_{\mS,\mS}"]\\[-3ex]
    & B(\mS,\mS) & 
\end{tikzcd}
\end{equation*}  
Here, we are using the fact that $\phi^\sharp$ is a monoid multiplication, and 
the pentagonal coherence condition from \Cref{def:lambdalaw}.
\end{proof}

We first prove a slightly more general statement, a corollary of which is
\Cref{thm:ind-up-to-blackquaregen}.

\begin{lem}
  \label{lem:lambdamain}
  Let $\brks{\phi, \gamma} \c \mS \to \mS \monto \mS \times B(\mS,\mS)$ be the
  operational model of a relatively flat $\lambda$-law $(\Xi, \Theta, \xi,
  \theta)$ that admits a predicate lifting. Then for every predicate $P \pred{} \mS$ and every
  locally maximal logical refinement $\logp^{\gamma,\ol{B}}P$, if
  \begin{enumerate}
  \item $\fimg{\iota^{\Xi}}{\ol{\Xi}(\logp^{\gamma,\ol{B}} P)} \leq {P}$ and
  \item $\fimg{\iota^{\Theta}}{\ol{\Theta}\blacksquare^{\gamma,\ol{B}} P}
      \leq \iimg{\phi}{{\logp^{\gamma,\ol{B}}P\monto* P}}$,
  \end{enumerate}
  we have that $\blacksquare^{\gamma,\ol{B}} P = \mS$.
\end{lem}

\begin{proof}
  It suffices to prove that $\fimg{\iota}{(\ol\Xi + \ol{\Theta})
  (\iimg{\phi}{\logp P\monto*\logp P})}
  \leq \iimg{\phi}{\logp P\monto*\logp P}$, from which the goal 
  follows by induction. Equivalently, we proceed with proving that
  \begin{equation*}
    \fimg{(\phi \comp \iota)}{(\ol\Xi + \ol{\Theta})
      (\iimg{\phi}{\logp P\monto*\logp P})}
    \leq \logp P\monto*\logp P.
  \end{equation*}
  Using the lifted adjunction $(\argument)\mathbin{\ol{\mon}}\logp P\dashv\logp
  P\monto*(\argument)$, the above is equivalent to
  \begin{equation*}
    \fimg{(\phi \comp\iota)^{\sharp}}{(\ol\Xi + \ol{\Theta})
      (\iimg{\phi}{\logp P\monto*\logp P}) \ol\mon \logp P}
    \leq \logp P.
  \end{equation*}
  We proceed with the following sub-goals (recall $\Xi = \Pt + \Xi'$ and the
  identity $(\phi \comp \iota)^{\sharp} = \phi^{\sharp} \comp (\iota \mon \mS)$):
  \begin{align}
    & \fimg{(\phi^\sharp\comp(\mathrm{var}_{\mS}\mon\mS))}
      {\Pt \mathbin{\ol\mon} \logp P} \label{eq:mainlam1}
    \leq \logp P, \\
    & \fimg{(\phi^{\sharp} \comp (\iota^{\Xi'} \mon \mS))}{\ol\Xi'
      (\iimg{\phi}{\logp P\monto*\logp P}) \ol\mon \logp P}
      \leq \logp P, \label{eq:mainlam2}\\
    & \fimg{(\phi^{\sharp} \comp (\iota^{\Theta} \mon \mS))}{\ol\Theta
      (\iimg{\phi}{\logp P\monto*\logp P}) \ol\mon \logp P}
      \leq \logp P\label{eq:mainlam3}
  \end{align}
  For \eqref{eq:mainlam1}, we observe that $\mS$ being a monoid implies
  $(\phi^\sharp\comp(\iota^{\Xi}\mon\mS))$ is equivalent to the
  canonical isomorphism $\Pt \mon \mS \cong \mS$, hence
  \begin{equation}
    \label{eq:unitandsquare}
    \fimg{(\phi^\sharp\comp(\mathrm{var}_{\mS}\mon\mS))}{\Pt \mathbin{\ol\mon} \logp P}
    \leq \logp P.
  \end{equation}
  For \eqref{eq:mainlam2}, we simplify further by \Cref{lem:phi_prop}:
  \begin{flalign*}
    && &\fimg{(\phi^{\sharp} \comp (\iota^{\Xi'} \mon \mS))}{\ol\Xi'
         (\iimg{\phi}{\logp P\monto*\logp P}) \ol\mon \logp P}
    & \\
    &&  \;= \quad & \fimg{(\iota^{\Xi'} \comp \Xi'\phi^{\sharp} \comp
                    \mathrm{st}^{\Xi'})}{\ol\Xi'
                    (\iimg{\phi}{\logp P\monto*\logp P}) \ol\mon \logp P}
    & \by{\Cref{lem:phi_prop}} \\
    && \;\leq \quad & \fimg{(\iota^{\Xi'} \comp \Xi'\phi^{\sharp})}{\ol\Xi'
                      (\iimg{\phi}{\logp P\monto*\logp P} \ol\mon \logp P)}
    & \by{$\logp P$ is pointed}\\
    && \;\leq \quad & \fimg{(\iota^{\Xi'})}{\ol\Xi'
                      (\overline{\mathrm{ev}}\comp({(\logp P\monto*\logp P)}
                      \ol\mon \logp P))} \\
    && \;\leq \quad & \fimg{(\iota^{\Xi'})}{\ol\Xi'\logp P}
  \end{flalign*}
  By the flatness of $\xi$ and assumption~(1), we follow the same procedure as
  in \Cref{th:main2} to get
  \[
    \fimg{(\iota^{\Xi})}{\ol\Xi\logp P} \leq \logp P.
  \]
  Restricting the above to $\Xi'$ finishes the proof of \eqref{eq:mainlam2}.
  We are left to prove \eqref{eq:mainlam3}, which we break down to the sub-goals
  (recall $\logp P = P \land \iimg{\gamma}{\ol{B}(\logp P, \logp P)}$)
  \begin{align}
    & \fimg{(\phi^{\sharp} \comp (\iota^{\Theta} \mon \mS))}{\ol\Theta
      (\iimg{\phi}{\logp P\monto*\logp P}) \ol\mon \logp P}
      \leq P, \label{eq:lam_sg1}\\
    & \fimg{(\phi^{\sharp} \comp (\iota^{\Theta} \mon \mS))}{\ol\Theta
      (\iimg{\phi}{\logp P\monto*\logp P}) \ol\mon \logp P}
      \leq \iimg{\gamma}{\ol{B}(\logp P, \logp P)}. \label{eq:lam_sg2}
  \end{align}
  Adjoint correspondences allows to write \eqref{eq:lam_sg1} as
  \[
    \fimg{\iota^{\Theta}}{\ol{\Theta}
      (\iimg{\phi}{\logp P\monto*\logp P})}
    \leq \iimg{\phi}{{\logp P\monto* P}},
  \]
  which holds by assumption (2). Finally,~\eqref{eq:lam_sg2} is equivalent to
\begin{align*}
\fimg{(\gamma\comp\phi^{\sharp}\comp(\iota^{\theta}\mon\mS))}{\ol{\Theta}
      (\iimg{\phi}{\logp P\monto*\logp P}) \ol\mon \logp P}\leq&\; \ol{B}(\logp P, \logp P),
\end{align*}
and we prove it as follows:
\begin{flalign*}
  && (\gamma\comp\phi^{\sharp}&\fimg{
  \comp(\iota^{\Theta}\mon\mS))}{\ol{\Theta}
    (\iimg{\phi}{\logp P\monto*\logp P}) \mon \logp P}
  &\\
&&  &\;= \fimg{(\theta^{\sharp} \comp (\Theta\phi\mon\mS))}{\ol{\Theta}
      (\iimg{\phi}{\logp P\monto*\logp P}) \mon \logp P}
  &\by{\Cref{lem:zetaphi}}\\
&&  &\;\leq \fimg{\theta^{\sharp}}{\ol{\Theta}
        (\logp P\monto*\logp P) \mon \logp P}
  &\\
  &&  &\;\leq \ol{B}(\logp P, \logp P)
  &\qedhere
\end{flalign*}
\end{proof}

\begin{proof}[Proof of \Cref{thm:ind-up-to-blackquaregen}]
  The goal is to instantiate \Cref{lem:lambdamain}. All we need to prove is
  \[
    \fimg{\iota^{\Theta}}{\ol{\Theta}
      (\iimg{\phi}{\logp P\monto*\logp P})}
    \leq \iimg{\phi}{{\logp P\monto* P}}.
  \]
  By the lifted adjunction $(\argument)\mathbin{\ol{\mon}}\logp P\dashv\logp
  P\monto*(\argument)$, the above is equivalent to
  \begin{equation*}
    \fimg{(\phi^{\sharp} \comp (\iota^{\Theta} \mon \mS))}{\ol\Theta
      (\iimg{\phi}{\logp P\monto*\logp P}) \ol\mon \logp P}
    \leq P.
  \end{equation*}
  We finish the proof with
  \begin{flalign*}
    && &\fimg{(\phi^{\sharp} \comp (\iota \mon \mS))}{\ol\Theta
         (\iimg{\phi}{\logp P\monto*\logp P}) \ol\mon \logp P}
    & \\
    &&  \;= \quad & \fimg{(\iota^{\Theta} \comp \Theta\phi^{\sharp} \comp \mathrm{st}^{\Theta})}{\ol\Theta
                    (\iimg{\phi}{\logp P\monto*\logp P}) \ol\mon \logp P}
    & \by{\Cref{lem:phi_prop}} \\
    && \;\leq \quad & \fimg{(\iota^{\Theta} \comp \Theta\phi^{\sharp})}{\ol\Theta
                      (\iimg{\phi}{\logp P\monto*\logp P} \ol\mon \logp P)}
    & \by{$\logp P$ is pointed}\\
    && \;\leq \quad & \fimg{(\iota^{\Theta})}{\ol\Theta
                      (\overline{\mathrm{ev}}\comp({(\logp P\monto*\logp P)}
                      \ol\mon \logp P))} \\
    && \;\leq \quad & \fimg{(\iota^{\Theta})}{\ol\Theta\logp P} \leq P. \qedhere
  \end{flalign*}
\end{proof}

\subsection{\texorpdfstring{\stlc{} as a $\lambda$-law}
  {\stlc{} as a lambda-law}}

\paragraph{The category $(\Set^{\fset/{\Tyl}})^{\Tyl}$.}
Let $\fset/{\Tyl}$ be the slice category of the category of finite
cardinals $\fset$, the skeleton of the category of finite sets, over the set of
types $\Tyl$. Category $\fset/{\Tyl}$ can be thought of as the category of
variable contexts over $\Tyl$: an object $\Gamma \in \fset/{\Tyl}$ is a map
$n=\{0,1,\dots,{n-1}\} \to \Tyl$ that assigns variables, represented by natural numbers, to  
their respective types (instead of finite ordinals, we could just as well use any
finite sets, for the price of slightly complicating the mathematics). Let
$|\Gamma|$ denote the domain of $\Gamma$ 
and let $\varnothing$ denote the empty variable context. Morphisms $\Gamma_{1} \to \Gamma_{2}$, i.e.
functions $r \c |\Gamma_{1}| \to |\Gamma_{2}|$ such that
$\Gamma_{1} = \Gamma_{2} \comp r$, are type-respecting \emph{renamings}. Each
type $\tau \in \Tyl$ determines the obvious single-variable context $\check \tau\c
\to \Tyl$; coproducts in~$\fset/{\Tyl}$ are formed by copairing:
\[
  (\Gamma_{1} \c |\Gamma_{1}| \to \Tyl) + (\Gamma_{2} \c |\Gamma_{1}| \to \Tyl)
  = [\Gamma_{1}, \Gamma_{2}] \c |\Gamma_{1}| + |\Gamma_{2}| \to \Tyl.
\]
For an important example, consider the operation of \emph{context extension} $(-
+ \check\tau) \c \fset/{\Tyl} \to \fset/{\Tyl}$ which, as the name suggests,
extends a variable context with a new variable that has a type of $\tau \in \Tyl$,
i.e. $\Gamma, x \c \tau$ in a type-theoretic notation.

Following~\cite{DBLP:conf/fossacs/Fiore05,DBLP:conf/csl/FioreH10,DBLP:journals/mscs/Fiore22}, $(\Set^{\fset/{\Tyl}})^{\Tyl}$ is the
mathematical universe used to model the simply type $\lambda$-calculus . The objects of
this category are type-indexed covariant presheaves over $\fset/{\Tyl}$ (equivalently,
of course, covariant presheaves over $(\fset/{\Tyl})\times\Tyl$) -- in
other words sets indexed by types and variable contexts that furthermore respect
context renaming. We briefly recall the core constructions in $(\Set^{\fset/{\Tyl}})^{\Tyl}$ 
that enable categorical modelling of typed languages with variable binding.
First, the so-called presheaf of
\emph{variables} $V_{\tau} \in \Set^{\fset/{\Tyl}}$ is given in terms of
the Yoneda embedding $V_{\tau} = \mathbf{y}(\check\tau)$. More explicitly,
\[
  V_{\tau}(\Gamma) \cong \{x \in |\Gamma| \mid \Gamma(x) = \tau\}.
\]
Exponentials $Y^{X}$ and their evaluation morphisms $\ev\c Y^X \times X \to Y$
in $\Set^{\fset/{\Tyl}}$
are respectively computed by the following formulas, standard for presheaf toposes~\cite[Sec.~I.6]{MacLaneMoerdijk92}:
\[
  Y^{X}(\Gamma) = \Set^{\fset/{\Tyl}}\bigl(\Set^{\fset/{\Tyl}}(\Gamma,
  \argument) \times X, Y\bigr) \quad \text{and} \quad
  \ev_{\Gamma}(f, x) = f_{\Gamma}(\id_{\Gamma}, x) \in Y(\Gamma),
\]
for $f\c (\argument)^{\Gamma} \times X \to Y$ and an element $x
\in X(\Gamma)$. For simplicity, we write \(f(x) \,:=\,\ev_{\Gamma}(f,x).\)

The context extension $(- + \check\tau) \c \fset/{\Tyl} \to
\fset/{\Tyl}$ operator yields an endofunctor $\delta^{\tau} \c
(\Set^{\fset/{\Tyl}})^{\Tyl} \to (\Set^{\fset/{\Tyl}})^{\Tyl}$ in
the category of presheaves by precomposition:
\begin{equation*}
  (\delta^{\tau}_{\tau'}X)(\Gamma) = X_{\tau'}(\Gamma + \check\tau).
\end{equation*}
This provides a categorical mechanism for typed \emph{variable binding}:
Informally, with $(\delta^{\tau}_{\tau' }X)(\Gamma)$, we render 
terms in context $\Gamma$ from terms in context $\Gamma,x\c\tau$, by abstracting 
the last variable. Category $(\Set^{\fset/{\Tyl}})^{\Tyl}$ is monoidal closed by
the so-called \emph{substitution monoidal structure}~([p.
6]\cite{DBLP:conf/csl/FioreH10}). The tensor product is given by
\begin{equation}
  (X \otimes_{\tau} Y) (\Gamma) = \int^{\Delta \in \Set^{\fset/{\Tyl}}} X_{\tau}(\Delta)
  \times \prod_{i \in |\Delta|}Y_{\Delta(i)}(\Gamma).
\end{equation}
The tensor unit is the presheaf of variables $V$. The internal hom
functor $\llangle -,-\rrangle \c ( (\Set^{\fset/{\Tyl}})^{\Tyl})^{\opp} \times
(\Set^{\fset/{\Tyl}})^{\Tyl} \to (\Set^{\fset/{\Tyl}})^{\Tyl}$, which models
\emph{simultaneous substitution}, is
([Eq. (4)]\cite{DBLP:conf/csl/FioreH10}):
\begin{equation}
  \llangle X,Y \rrangle_{\tau}(\Gamma) = \Set^{\fset/{\Tyl}}\bigl(\prod_{i \in
    |\Gamma|}X_{\Gamma(i)}, Y_{\tau}\bigr).
\end{equation}
An element of $\llangle X,Y \rrangle_{\tau}(\Gamma)$ is thus a natural in $\Delta$ family 
$(f_{\Delta} \c \prod_{i \in |\Gamma|}X_{\Gamma(i)}(\Delta) \to Y_{\tau}(\Delta))_{\Delta\in \fset/{\Tyl}}$, which,
given a tuple of terms $\sigma=(\sigma_{0} \in X_{\Gamma(0)}(\Delta), \ldots,
\sigma_{n-1} \in X_{\Gamma(n-1)}(\Delta))$ that
represents the substitution $[\sigma_0/0,\ldots,\sigma_{n-1}/n-1]$ with
$n=|\Gamma|$, returns a term $t \in Y_{\tau}(\Delta)$. The naturality condition
means that $f_{\Delta}$ commutes with context renaming. The substitution tensor
with the hom $\llangle -,-\rrangle$ forms the following adjoint situation:
\[
  \begin{tikzcd}
    (\Set^{\fset/{\Tyl}})^{\Tyl}(Y \otimes X, Z) \ar[bend left=1em]{r}{(\argument)^{\flat}}
    & (\Set^{\fset/{\Tyl}})^{\Tyl}(Y, \llangle X , Z \rrangle) \ar[bend
    left=1em]{l}{(\argument)^{\sharp}}
\end{tikzcd}
\]

\paragraph{Categorical modelling of \stlc{}.}
We shall now construct a $\lambda$-law in $(\Set^{\fset/{\Tyl}})^{\Tyl}$ for the
semantics of \stlc{}. Recall the syntax endofunctor $\Sigma$ and
behaviour bifunctor $B^{\lambda}$:
\begin{equation*}
  \begin{aligned}
    \Sigma_{\utype}X =
    &\; V_{\utype} + K_{1}
      + \coprod_{\tau \in \Tyl}X_{\arty{\tau}{\utype}}
      \times X_{\tau},\\
    \Sigma_{\arty{\tau_{1}}{\tau_{2}}} X=
    &\; V_{\arty{\tau_{1}}{\tau_{2}}} + \delta^{\tau_{1}}_{\tau_{2}} X
      + \coprod_{\tau \in \Tyl}X_{\arty{\tau}{\arty{\tau_{1}}{\tau_{2}}}}
      \times X_{\tau},
  \end{aligned}
\end{equation*}
and
\begin{flalign*}
  &&B^{\lambda}(X,Y)=\;&\llangle X,Y\rrangle \times B(X,Y)&&\\[1ex] 
  \text{where}&& 
                 \llangle X,Y \rrangle_{\tau}(\Gamma) =\;& \Set^{\fset/{\Tyl}}\Bigl(\prod_{x \in
                                                           |\Gamma|}X_{\Gamma(x)}, Y_{\tau}\Bigr),\notag\\
  && B(X,Y) =\;& (K_{1} + Y + D(X,Y)),\notag\\
  && D_{\utype}(X,Y) =\;& K_{1} \qquad \text{and}
                          \qquad D_{\arty{\tau_{1}}{\tau_{2}}}(X,Y) = Y_{\tau_{2}}^{X_{\tau_{1}}},\notag
\end{flalign*}

First, we break the signature endofunctor $\Sigma \c
(\Set^{\fset/{\Tyl}})^{\Tyl} \to (\Set^{\fset/{\Tyl}})^{\Tyl}$ into the three parts
$\Sigma = V + \Xi' + \Theta$, with $V \c (\Set^{\fset/{\Tyl}})^{\Tyl} \to
(\Set^{\fset/{\Tyl}})^{\Tyl}$ being the presheaf of variables, $\Theta$ the part
of $\lambda$-abstraction and $\Xi'$ that of application and the $\utype$
expression.
Typed signature endofunctors such as $\Xi'$ and $\Theta$ canonically admit a
$V$-pointed strength~\cite[p. 6]{DBLP:conf/csl/FioreH10} (see also
\cite{DBLP:conf/lics/Fiore08}). The most interesting clause is that of
$\Theta$, which underlines the need for the point $\nu \c V \to Y$,
representing the inclusion of variables in some presheaf $Y$ (skipping the
void case on $\utype$):
\[
  \begin{aligned}
    \mathrm{st}^{\Theta}_{X,\nu \c V \to Y}
    & \c \Theta X \otimes Y \to \Theta (X
      \otimes Y) \\
    \mathrm{st}^{\Theta}_{X,\nu \c V \to Y,\arty{\tau_{1}}{\tau_{2}},\Gamma}(\Delta, t \in
    X_{\tau_{2}}(\Delta + \check\tau_{1}), \vec{s} \in \prod_{\tau' \in
    \Delta + \check\tau_{1}}Y_{\tau'}(\Gamma))
    & = (\Delta + \check\tau_{1}, t,
      (\vec{\mathsf{wk}}_{\tau_{1}}(\vec{s}),
      \nu_{\Gamma + \check\tau_{1}}(\mathrm{new}))),
  \end{aligned}
\]
where $\vec{\mathsf{wk}}_{\tau_{1}}$ extends the typing context
of each member $\Gamma \vdash s_{i} \c \tau'$ of the substitution $\vec{s}$
by $\tau_{1}$: $\Gamma, x \c \tau_{1} \vdash \mathsf{wk}_{\tau_{1}}(s_{i}) \c
\tau'$ for each $\tau' \in \Delta$. 

Let $\Gamma,x \c \tau_{1} \vdash t \c \tau_{2}$ be a $\lambda$-term. The
substitution structure of $t$ is the function accepting a
$|\Gamma + \check\tau_{1}|$-sized substitution, meaning a tuple of terms
$(\Delta \vdash \sigma_{0} \c (\Gamma + \check\tau_{1})(0),\ldots,\Delta \vdash
\sigma_{|\Gamma + \check\tau_{1}|- 1} \c \tau_{1})$ and
returning the term $\Delta \vdash t[x_{0}/\sigma_{0},\dots,x_{|\Gamma + \check\tau_{1}|
  - 1}/\sigma_{|\Gamma + \check\tau_{1}| - 1}] \c \tau_{2}$. In this example of \stlc{}, the
natural transformation  $\theta \c \Theta\llangle X, Y \rrangle \to
B(X,\llangle X,Y \rrangle)$ expresses the fact that the substitution structure
of $t$ is equivalently a function accepting a term $\Gamma \vdash e \c
\tau_{1}$ and returning the substitution structure
\[
  \sigma \in \prod_{i \in \Gamma}\Trl_{\Gamma(i)}(\Delta) \mapsto \Delta \vdash
  t[x_{|\Gamma + \check\tau_{1}| - 1}/e][x_{0}/\sigma_{0},\dots] \c \tau_{2}.
\]

In more abstract terms, $\theta$ is defined as follows (again skipping the
void case of $\theta_{\utype}$):
\begin{equation*}
  \begin{aligned}
    \mathrm{\theta}_{\arty{\tau_{1}}{\tau_{2}},\Gamma}
    \c \Theta_{\arty{\tau_{1}}{\tau_{2}}}\llangle X, Y
    \rrangle(\Gamma)
    & = \delta_{\tau_{2}}^{\tau_{1}}\llangle X , Y\rrangle(\Gamma)
      = \Set^{\fset/{\Tyl}}(\prod_{i \in
      |\Gamma + \check\tau_{1}|}X_{(\Gamma + \check\tau_{1})(i)}, Y_{\tau_{2}}\bigr)
    \\
    & =
      \Set^{\fset/{\Tyl}}(\prod_{i \in
      |\Gamma|}X_{\Gamma(i)} \times X_{\tau_{1}}, Y_{\tau_{2}})
      \cong
      \Set^{\fset/{\Tyl}}(\prod_{i \in
      |\Gamma|}X_{\Gamma(i)}, Y_{\tau_{2}}^{X_{\tau_{1}}}) \\
    & =
      \llangle X,Y^{X_{\tau_{1}}}\rrangle_{\tau_{2}}(\Gamma)
      \hookrightarrow \llangle X, B(X,Y)\rrangle_{\arty{\tau_{1}}{\tau_{2}}}(\Gamma).
  \end{aligned}
\end{equation*}
where the last isomorphism is given by the cartesian closed structure of
$\Set^{\fset/{\Tyl}}$. The law $\xi_{X,Y} \c (I + \Xi')(X \times B(X,Y)) \to B(X, (I +
\Xi')^{\star}(X+Y))$ specifies the semantics of variables, the unit expression
and application.
\begin{align*}
  \xi_{X,Y,\tau,\Gamma} \c \quad
  & (I + \Xi')(X \times B(X,Y))_{\tau}(\Gamma)
  & \to \quad
  & B(X, (I + \Xi')^\star (X+Y))(\Gamma) \\
  \xi_{X,Y,\tau,\Gamma}(tr) = \quad
  & \texttt{case}~tr~\texttt{of} \\
  & \mathsf{e} & \mapsto \quad & * & \\
  & x \c \tau & \mapsto \quad & \inl(*) & \\
  & (p,f)\app{}{}(q,g) & \mapsto \quad
  &
    \begin{cases}
      f(q) \text{\quad if $f \c Y_{\tau}^{X_{\tau_{1}}}(\Gamma)$} \\
      f \app{}{} q \text{\quad if $f \c Y_{\arty{\tau_{1}}{\tau}}(\Gamma)$} \\
      \inl(*) \text{\quad if $f \in 1$}.
    \end{cases} &
\end{align*}

Finally, we explain the compatibility
condition \eqref{eq:thetacompat}. Consider the following data:
\begin{enumerate}
\item Types $\tau_{1}$, $\tau_{2}$, typing contexts $\Gamma, \Delta$ and $E$,
  and a term $\Gamma \vdash e \c \tau_{1}$.
\item A natural transformation $f \in \Set^{\fset/{\Tyl}}\bigl(\prod_{i \in
    |E+\check\tau_{1}|}X_{(E + \check\tau_{1})(i)} , Y_{\tau_{2}}\bigr)$.
\item A substitution $(g_{0},g_{1},\dots) \in \prod_{i \in |E|}
  \Set^{\fset/{\Tyl}}\bigl(\prod_{j \in
    |\Delta|}X_{\Delta(j)} , X_{E(i)}\bigr)$.
\item A substitution $(\sigma_{0},\sigma_{1},\dots)
  \in \prod_{i \in |\Delta|}X_{\Delta(i)}(\Gamma)$.
\end{enumerate}
Recall that the \emph{weakening} of a natural transformation $g \in
\Set^{\fset/{\Tyl}}\bigl(\prod_{j \in
  |\Delta|}X_{\Delta(j)} , X_{\tau}\bigr)$ is
\[
  {wk}^{\tau_{1}}_{\tau,\Delta}(g) \in
  \Set^{\fset/{\Tyl}}\Bigl(\prod_{j \in
    |\Delta + \check\tau_{1}|}X_{(\Delta + \check\tau_{1})(j)} , X_{\tau}\Bigr),
\]
given by
$
\mathsf{wk}^{\tau_{1}}_{\tau,\Delta}(g)(\sigma_{0},\dots,e) =
g(\sigma_{0},\dots).
$
The compatibility condition implies that
\[
  \begin{aligned}
    & f(\mathsf{wk}^{\tau_{1}}_{\tau_{2},\Delta}(g_{0})(\sigma_{0},\sigma_{1},\dots,e),\mathsf{wk}^{\tau_{1}}_{\tau_{2},\Delta}(g_{1})(\sigma_{0},\sigma_{1},\dots,e),\dots,\pi_{|\Delta|}(\sigma_{0},\sigma_{1},\dots,e))
      = \\
    & f(g_{0}(\sigma_{0},\sigma_{1},\dots),g_{1}(\sigma_{0},\sigma_{1},\dots),\dots,e).
  \end{aligned}
\]

\end{adjustwidth}
\end{document}